\documentclass[english]{article}
\usepackage[T1]{fontenc}
\usepackage[latin9]{inputenc}
\usepackage{geometry}
\usepackage{float}
\usepackage{mathtools}
\usepackage{amsmath}
\usepackage{amsthm}
\usepackage{amssymb}

\makeatletter

\floatstyle{ruled}
\newfloat{algorithm}{tbp}{loa}
\providecommand{\algorithmname}{Algorithm}
\floatname{algorithm}{\protect\algorithmname}

\numberwithin{figure}{section}
\numberwithin{equation}{section}
\newcommand{\lyxaddress}[1]{
	\par {\raggedright #1
	\vspace{1.4em}
	\noindent\par}
}
\theoremstyle{plain}
\newtheorem{thm}{\protect\theoremname}
\theoremstyle{plain}
\newtheorem{lem}[thm]{\protect\lemmaname}
\theoremstyle{plain}
\newtheorem{prop}[thm]{\protect\propositionname}
\theoremstyle{remark}
\newtheorem*{rem*}{\protect\remarkname}
\theoremstyle{definition}
\newtheorem{defn}[thm]{\protect\definitionname}
\theoremstyle{plain}
\newtheorem*{assumption*}{\protect\assumptionname}
\theoremstyle{plain}
\newtheorem{assumption}[thm]{\protect\assumptionname}

\usepackage{algorithm,algpseudocode}

\usepackage{tikz}
\usetikzlibrary{shapes.geometric, arrows.meta,arrows, positioning}
\tikzstyle{block} = [rectangle, draw, fill=blue!20, text width=5em, text centered, rounded corners, minimum height=2em]
\tikzstyle{line} = [draw, -latex']
\tikzstyle{plainline} = [draw]
\tikzset{
	block/.style={rectangle, draw, fill=blue!20, text width=10em, font=\footnotesize,  text centered, rounded corners, minimum height=3em},
	lblock/.style={rectangle, draw, fill=blue!20, text width=15em, ,font=\footnotesize, text centered, rounded corners, minimum height=.8em},
	doublearrow/.style={draw, latex'-latex'}  % Style for lines with arrows at both ends
}
\tikzset{multiline node/.style={ % Define a style for multiline nodes
		align=center,
		text width=5cm % Ensuring enough text width for wrapping
	},
	line/.style={
		draw, -Latex  % Default arrow size
	},
}

\makeatother

\usepackage{babel}
\providecommand{\assumptionname}{Assumption}
\providecommand{\definitionname}{Definition}
\providecommand{\lemmaname}{Lemma}
\providecommand{\propositionname}{Proposition}
\providecommand{\remarkname}{Remark}
\providecommand{\theoremname}{Theorem}

\begin{document}
\title{Coupled Stochastic-Statistical Equations for Filtering Multiscale
Turbulent Systems }
\author{Di Qi\textsuperscript{a} and Jian-Guo Liu\textsuperscript{b} }
\maketitle

\lyxaddress{\textsuperscript{a}Department of Mathematics, Purdue University,
150 North University Street, West Lafayette, IN 47907, USA}

\lyxaddress{\textsuperscript{b}Department of Mathematics and Department of Physics,
Duke University, Durham, NC 27708, USA}
\begin{abstract}
We present a new strategy for filtering high-dimensional
multiscale systems characterized by high-order non-Gaussian statistics
using observations from leading-order moments. A closed stochastic-statistical
modeling framework suitable for systematic theoretical analysis and
efficient numerical simulations is designed. Optimal filtering
solutions are derived based on the explicit coupling structures of
stochastic and statistical equations subject to linear operators,
which satisfy an infinite-dimensional Kalman-Bucy filter with conditional Gaussian dynamics.
To facilitate practical implementation, we develop a finite-dimensional stochastic filter model that approximates the optimal filter solution. 
We prove that this approximating filter effectively captures key
non-Gaussian features, demonstrating consistent statistics with the optimal filter first in
its analysis step update, then at the long-time limit guaranteeing stable convergence
to the optimal filter. Finally, we build a practical ensemble filter
algorithm based on the approximating filtering model, which enables accurate recovery of the true model statistics. 
The proposed modeling and filtering strategies are applicable to a wide range challenging problems in science and engineering, particularly for statistical
prediction and uncertainty quantification of multiscale turbulent
states.
\end{abstract}

\section{Introduction\protect\label{sec:Intro}}

Complex turbulent phenomena are widely observed in various science
and engineering systems \cite{salmon1998lectures,diamond2010modern,nazarenko2016wave}.
Such systems are characterized by a wide spectrum of nonlinearly coupled
spatiotemporal scales and high degrees of inherent internal instability
\cite{majda2016introduction,frisch1995turbulence}. A probabilistic
formulation containing highly non-Gaussian statistics is required
to quantify uncertainties in the high-dimensional turbulent state
\cite{pardoux2007stochastic,majda2018strategies}. Traditional ensemble
approaches using a particle system representation to approximate the
probability evolution quickly become computationally prohibitive since
a sufficiently large sample size is necessary to capture the extreme
non-Gaussian outliers even for relatively low dimensional systems
\cite{leutbecher2008ensemble,tong2021extreme}. As a result, rigorous
analysis often becomes intractable and direct numerical simulations
are likely to be expensive and inaccurate.

Filtering strategies \cite{reich2015probabilistic,law2015data,sanz2023inverse}
have long been used for finding the optimal probability estimates
of a stochastic state with large uncertainty based on partial and
noisy observation data. Applications to improve the probability forecast
can be found in various practical dynamical systems \cite{ghil1981applications,kalnay2003atmospheric,balakrishna2022determining,biswas2024unified}.
In predicting nonlinear turbulent signals, ensemble Kalman filters
\cite{evensen1994sequential,houtekamer1998data,evensen2009ensemble}
as well as the related particle methods \cite{del1997nonlinear,doucet2001sequential,calvello2022ensemble}
provide an effective tool for state and parameter estimations. Filtering
theories \cite{yau1998finite,crisan2010approximate,pathiraja2021mckean}
and corresponding numerical solutions \cite{kurtz2001numerical,bergemann2012ensemble,crisan2014numerical}
for general nonlinear systems have been investigated through different
approaches. Despite wide applications \cite{ernst2015analysis,majda2018performance,de2020analysis},
difficulties persist for accurate statistical forecast of turbulent
states especially when non-Gaussian features are present in the target
probability distribution. Conventional ensemble-based approaches often
suffer inherent difficulties in estimating the crucial higher-order
moment statistics and maintaining stable prediction with finite number
of particles \cite{snyder2008obstacles,gottwald2013mechanism,surace2019avoid}.
On the other hand, in many situations observations of the statistical
states, such as the mean and covariance, are easier to extract from
various sources of data \cite{majda2019linear}. Therefore, a promising
research direction is to propose new ensemble filtering models to
recover the highly non-Gaussian probability distributions using statistical
observations from the leading-order mean and covariance \cite{bach2023filtering}.

\subsection{General problem setup}

We start with a general mathematical formulation \cite{majda2018strategies}
modeling a high-dimensional stochastic state variable $u_{t}\in\mathbb{R}^{d}$
involving nonlinear multiscale interactions satisfying the following
stochastic differential equation (SDE)
\begin{equation}
\frac{\mathrm{d}u_{t}}{\mathrm{d}t}=\Lambda u_{t}+B\left(u_{t},u_{t}\right)+F_{t}+\sigma_{t}\dot{W}_{t}.\label{eq:abs_formu}
\end{equation}
On the right hand side of the above equation, the linear operator,
$\Lambda:\mathbb{R}^{d}\rightarrow\mathbb{R}^{d}$, represents linear
dispersion and dissipation effects. The nonlinear effect in the dynamical
system is introduced via a bilinear quadratic operator, $B:\mathbb{R}^{d}\times\mathbb{R}^{d}\rightarrow\mathbb{R}^{d}$.
The system is subject to external forcing effects that are decomposed
into a deterministic component, $F_{t}$, and a stochastic component
represented by a Gaussian white noise $W_{t}\in\mathbb{R}^{s}$ with
$\sigma_{t}\in\mathbb{R}^{d\times s}$. The model emphasizes the important
role of quadratic interactions through $B\left(u,u\right)$. This
typical structure is inherited from a finite-dimensional truncation
of the continuous system, for example, a spectral projection of the
nonlinear advection in fluid model \cite{lesieur1987turbulence}.
Many realistic systems with wide applications \cite{hu2021initial,qi2023random,gao2024mean}
can be categorized in the general dynamical equation \eqref{eq:abs_formu}. 

The evolution of the model state $u_{t}$ depends on the sensitivity
to the randomness in initial conditions and external stochastic effects,
which will be further amplified in time by the inherent internal instability
due to the nonlinear coupling term \cite{majda2016introduction,qi2018rigorous}.
Assuming that the stochastic solution satisfies a continuous probability
density function (PDF), the time evolution of the PDF $p_{t}$ is
given by the associated Fokker-Planck equation (FPE) 
\begin{equation}
\frac{\partial p_{t}}{\partial t}=\mathcal{L}_{\mathrm{FP}}p_{t}\coloneqq-\nabla_{u}\cdot\left[\left(\Lambda u+B\left(u,u\right)+F_{t}\right)p_{t}\right]+\frac{1}{2}\nabla_{u}\cdot\left[\nabla_{u}\cdot\left(\sigma_{t}\sigma_{t}^{\intercal}p_{t}\right)\right],\;p_{t=0}=p_{0},\label{eq:FPE}
\end{equation}
where $\mathcal{L}_{\mathrm{FP}}$ represents the Fokker-Planck operator
with $\nabla\cdot\left(\nabla\cdot A\right)=\sum_{k,l}\frac{\partial^{2}A_{kl}}{\partial u_{k}\partial u_{l}}$.
However, it remains a challenging task for directly solving the FPE
\eqref{eq:FPE} as a high-dimensional PDE. As an alternative approach,
ensemble forecast by tracking the Monte-Carlo solutions estimates
the essential statistics through empirical averages among a group
of samples drawn i.i.d. from the initial distribution $u^{\left(i\right)}\left(0\right)\sim p_{0}$
at the starting time $t=0$. The PDF solution $p_{t}\left(u\right)$
and the associated statistical expectation of any test function $\varphi\left(u\right)$
at each time instant $t>0$ are then approximated by the empirical
ensemble representation
\begin{equation}
p_{t}\left(u\right)\simeq p_{t}^{N}\left(u\right)\coloneqq\frac{1}{N}\sum_{i=1}^{N}\delta\left(u-u_{t}^{\left(i\right)}\right),\quad\mathbb{E}_{p_{t}}\varphi\left(u\right)\simeq\mathbb{E}_{p_{t}^{N}}\varphi\left(u\right)=\frac{1}{N}\sum_{i=1}^{N}\varphi\left(u_{t}^{\left(i\right)}\right),\label{eq:pdf_MC}
\end{equation}
In practice, large errors will be introduced in the above empirical
estimations since only a finite sample approximation is available
in modeling the probability distribution and statistics in a high
dimensional system.

It is expected that the model errors in the predicted PDF of the stochastic
states from the finite ensemble simulation can be effectively corrected
through the available observation data. In designing new filtering
strategies, we propose to use statistical observations from mean and
covariance to improve the accuracy and stability in the finite ensemble
forecast using the finite ensemble approximation. However, the general
formulation \eqref{eq:abs_formu} as well as the associated FPE \eqref{eq:FPE}
becomes inconvenient to use since all the multiscale stochastic processes
in the equation are mixed together containing all high-order statistics.
The main goal of this paper is thus to develop a systematic modeling
framework with strategies to accurately capture the (potentially highly
non-Gaussian) PDF $p_{t}$ with the help of statistical measurements
in the leading moments.

\subsection{Overview of the paper}

In this paper, we study nonlinear filtering  of the general multiscale
turbulent system \eqref{eq:abs_formu} according to the two-step filtering
procedure. In the first forecast step, a new modeling framework is
designed combining statistical equations \eqref{eq:dyn_stat} for
the leading moments and stochastic equations \eqref{eq:dyn_stoc}
for capturing higher-order statistics. In the second analysis step,
we propose the observation process from the statistical equations
of the leading-order mean and covariance to improve the probability
prediction of the stochastic model. The statistical equations for
observation processes \eqref{eq:fpf_obs} then become linearly dependent
on the PDF of stochastic state, while the evolution of the PDF is
also subject to conditional linear dynamics \eqref{eq:fpf_model}.
Therefore, optimal filtering equations can be explicitly derived based
on the linearity in both signal and observation processes. Finally,
the expensive optimal filter equations are approximated by an efficient
filtering model demonstrating equivalent statistics according to the
most important observation functions \eqref{eq:operators_obs}. 

The new multiscale nonlinear filtering model is constructed under
the following step-by-step procedure, which will finally be combined
to build a practical ensemble filtering strategy able to recover non-Gaussian
PDFs: 
\begin{itemize}
\item First, we propose a precise coupled stochastic-statistical formulation
\eqref{eq:closure_model} as the forecast model for multiscale turbulence:
the stochastic dynamics will serve as the signal process in filtering
including high-order non-Gaussian features, while the reinforced statistical
equations provide the observation process; 
\item Second, the optimal filtering problem is formulated based on the particular
coupling structure in the multiscale stochastic-statistical model:
the optimal filter solution \eqref{eq:KB-fpf} is given based on the
PDF of the stochastic state, combined with the mean and covariance
as a natural choice of the observed state;
\item Third, an effective statistical filtering model \eqref{eq:mf-fpf}
is developed approximating the optimal filter solution in key leading-order
statistics: an equivalent McKean-Vlasov SDE containing higher moments
feedbacks is adapted from the optimal filter solution;
\item Last, discrete ensemble filtering schemes \eqref{eq:model_full} are
constructed as a particle approximation of the statistical filtering
model: the filter SDE is approximated by interacting particles and
the statistical moments are computed by empirical ensemble averages.
\end{itemize}
The coupled stochastic-statistical model \eqref{eq:closure_model}
by itself can serve as an effective tool for statistical forecasts
and uncertainty quantification. Efficient computational algorithms
have been developed \cite{qi2023high,qi2023random} that demonstrate
high skill in capturing crucial non-Gaussian phenomena such as extreme
events and fat-tailed PDFs. Then combined with the observation data,
the resulting filtering McKean-Vlasov SDE \eqref{eq:mf-fpf} is only
implicitly dependent on the probability distribution, which can be
computed directly from the statistical equations. This enables efficient
computational strategies to effectively improve the accuracy and stability
in capturing high-order non-Gaussian features based on only observation
from the lower moments. We illustrate the main steps in building effective
models for capturing probability distributions in Fig.~\ref{fig:Diagram}.

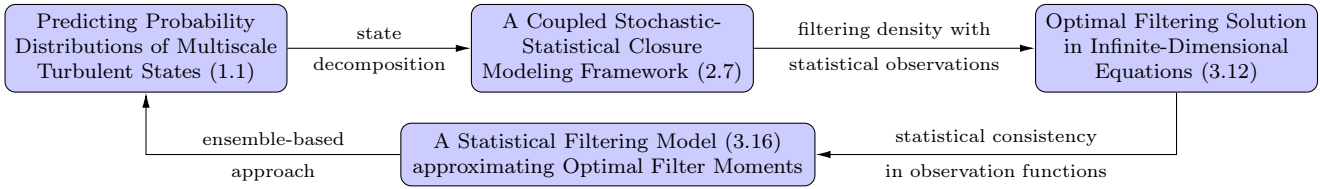
\begin{figure}
\begin{tikzpicture}[node distance = 5cm, auto, font=\small] % Adjusted font size here 		
% Place nodes %		
		\node [block] (turbulent) {Predicting Probability Distributions of Multiscale Turbulent States \eqref{eq:abs_formu}};
		\node [block,right of=turbulent,node distance=6.2cm] (model) {A Coupled Stochastic-Statistical Closure Modeling Framework \eqref{eq:closure_model}};
		\node [block, right of=model, node distance=7.5cm] (filter) {Optimal Filtering Solution in Infinite-Dimensional Equations \eqref{eq:KB-fpf}};
 
		\path[draw, -{Latex[length=2mm,width=1mm]}] (turbulent) -- (model) node[midway, above,font=\scriptsize] {state} node[midway, below, multiline node,font=\scriptsize] {decomposition}; 
		\path[draw, -{Latex[length=2mm,width=1mm]}] (model) -- (filter) node[midway, above,font=\scriptsize] {filtering density with} node[midway, below, multiline node,font=\scriptsize] {statistical observations}; 

		\node [lblock,  below right=.4cm and 1.5cm of turbulent] (meanfield) {A Statistical Filtering Model \eqref{eq:mf-fpf} approximating Optimal Filter Moments};

		\path[draw, -{Latex[length=2mm,width=1mm]}]  (filter)|- (meanfield) node[near end,  above,font=\scriptsize] {statistical consistency} node[near end, below, multiline node,font=\scriptsize] {in observation functions}; 
		\path[draw, -{Latex[length=2mm,width=1mm]}]  (meanfield)-| (turbulent) node[near start, above,font=\scriptsize] {ensemble-based} node[near start, below, multiline node,font=\scriptsize] {approach}; 
		
		%\path[plainline](meanfield) -| (ensemble);
\end{tikzpicture}

\caption{Diagram illustrating main ideas in constructing the modeling framework\protect\label{fig:Diagram}}
\end{figure}

Our main focus of this paper is to provide a detailed description
of the forecast and analysis step of the filtering method so that
various applications can directly follow based on this highly adaptive
general framework. We provide detailed discussion in each component
of the proposed model and methods: i) we show in Proposition~\ref{prop:model_equiv}
that the highly tractable closure model \eqref{eq:closure_model}
demonstrates consistent statistics as the original system \eqref{eq:abs_formu}
which plays a fundamental role in the construction of the filtering
methods; ii) the filter solution \eqref{eq:KB-fpf} is shown optimal
in the mean square sense by exploiting the conditional Gaussian structure
of the forward equation, and the approximation model is shown to recover
the same key statistics during the analysis step update in Theorem~\ref{thm:consist_analy};
iii) the long-time convergence in statistics to the optimal filter
is demonstrated in Theorem~\ref{thm:mean_conv} concerning the complete
filtering procedure using the statistical filtering  model \eqref{eq:mf-fpf};
and iv) the ensemble approximation with discrete time step is shown
to recover to the true statistics at the large ensemble limit.

The structure of the paper is organized as follows: In Section~\ref{sec:A-statistically-consistent},
we first set up a statistically consistent formulation for the general
multiscale system \eqref{eq:abs_formu} that is suitable for the construction
of the statistical filtering models. Section~\ref{sec:Data-assimilation}
shows the major ideas of finding the optimal filter solution and constructing
approximate filtering models. Long-time convergence and stability
of the approximation filter model is then discussed in Section~\ref{sec:Stability-and-convergence}.
Combining each component of the ideas, the final ensemble statistical
filtering model is developed in Section~\ref{sec:Ensemble-approximation}
Finally, a summary of this paper is given in Section~\ref{sec:Summary}.

\section{A statistically consistent modeling framework for multiscale dynamical
systems\protect\label{sec:A-statistically-consistent}}

In this section, we first introduce a new formulation for the original
system \eqref{eq:abs_formu} using an explicit macroscopic and microscopic
decomposition of the multiscale state. In particular, we show that
the new formulation provides consistent statistics with the original
system in the explicit leading moment equations, as well as higher-order
statistics in the stochastic equation. In addition, the new formulation
also enjoys a more tractable dynamical structure to be adapted to
the filtering framework.

\subsection{The statistical and stochastic equations as macroscopic and microscopic
states}

In order to identify the detailed multiscale interactions in the general
system \eqref{eq:abs_formu}, we decompose the random model state
$u_{t}$ into the composition of a statistical mean $\bar{u}_{t}$
and stochastic fluctuations $u_{t}^{\prime}$ in a finite-dimensional
representation under an orthonormal basis $\left\{ \hat{v}_{k}\right\} _{k=1}^{d}$
with $\hat{v}_{k}\cdot\hat{v}_{l}=\delta_{kl}$ 
\begin{equation}
u_{t}=\bar{u}_{t}+u_{t}^{\prime}=\sum_{k=1}^{d}\bar{u}_{k,t}\hat{v}_{k}+\sum_{k=1}^{d}Z_{k,t}\hat{v}_{k},\quad\mathrm{with}\;\bar{u}_{k,t}=\hat{v}_{k}\cdot\bar{u}_{t},\;Z_{k,t}=\hat{v}_{k}\cdot\left(u_{t}-\bar{u}_{t}\right).\label{eq:decomp}
\end{equation}
Above, $\bar{u}_{t}\in\mathbb{R}^{d}$ represents the dominant largest-scale
structure (for example, the zonal jets in geophysical turbulence or
the coherent radial flow in fusion plasmas); and $Z_{t}=\left[Z_{1,t},\cdots,Z_{d,t}\right]^{\intercal}\in\mathbb{R}^{d}$
are stochastic coefficients projected on each eigenmode $\hat{v}_{k}$,
whose randomness illustrates the uncertainty in each single scale
of $u_{t}^{\prime}$. In particular, we will show that the dynamics
of the stochastic modes $Z_{t}$ contain nonlinear interactions among
a large number of coupled multiscale fluctuations, which demonstrate
the charactering feature in turbulent systems demonstrating forward
and backward energy cascades \cite{frisch1995turbulence,majda2016introduction}. 

\subsubsection{Statistical equations for the macroscopic states}

First, we define the leading-order mean and covariance according to
the state decomposition \eqref{eq:decomp} 
\begin{equation}
\bar{u}_{t}=\mathbb{E}_{p_{t}}\left[u\right]\coloneqq\int_{\mathbb{R}^{d}}up_{t}\left(\mathrm{d}u\right),\quad R_{kl,t}=\mathbb{E}_{p_{t}}\left[Z_{k,t}Z_{l,t}\right]\coloneqq\int_{\mathbb{R}^{d}}\hat{v}_{k}\cdot\left(u-\bar{u}_{t}\right)\left(u-\bar{u}_{t}\right)\cdot\hat{v}_{l}p_{t}\left(\mathrm{d}u\right),\;1\leq k,l\leq d,\label{eq:macro_states}
\end{equation}
where the expectations are with respect to the law $p_{t}$ from the
PDF solution of \eqref{eq:FPE}. Above, statistical mean $\bar{u}_{t}$
and covariance $R_{t}$ can represent the \emph{macroscopic} physical
quantities that are easiest to achieve from direct measurements. The
dynamical equations for the mean and covariance can be derived by
the following statistical equations for $1\leq k,l\leq d$\addtocounter{equation}{0}\begin{subequations}\label{eq:dyn_stat} 
\begin{align}
\frac{\mathrm{d}\bar{u}_{k,t}}{\mathrm{d}t}= & \:\hat{v}_{k}\cdot\left[\Lambda\bar{u}_{t}+B\left(\bar{u}_{t},\bar{u}_{t}\right)\right]+\sum_{m,n=1}^{d}\gamma_{kmn}\mathbb{E}_{p_{t}}\left[Z_{m,t}Z_{n,t}\right]+\hat{v}_{k}\cdot F_{t},\label{eq:dyn_mean}\\
\frac{\mathrm{d}R_{kl,t}}{\mathrm{d}t}= & \sum_{m=1}^{d}\left[L_{km}\left(\bar{u}_{t}\right)R_{t,ml}+R_{t,km}L_{lm}\left(\bar{u}_{t}\right)\right]+Q_{t,kl},\label{eq:dyn_cov}\\
+ & \sum_{m,n=1}^{d}\gamma_{kmn}\mathbb{E}_{p_{t}}\left[Z_{m,t}Z_{n,t}Z_{l,t}\right]+\gamma_{lmn}\mathbb{E}_{p_{t}}\left[Z_{m,t}Z_{n,t}Z_{k,t}\right].\nonumber 
\end{align}
\end{subequations}Above, we define the nonlinear coupling coefficients
$\gamma_{kmn}=\hat{v}_{k}\cdot B\left(\hat{v}_{m},\hat{v}_{n}\right)$,
and the white noise coefficient is defined as $\Sigma_{t}=\left[\left(\hat{v}_{1}^{\intercal}\sigma_{t}\right)^{\intercal},\cdots,\left(\hat{v}_{d}^{\intercal}\sigma_{t}\right)^{\intercal}\right]^{\intercal}$
with $Q_{t}=\Sigma_{t}\Sigma_{t}^{\intercal}\in\mathbb{R}^{d\times d}$.
The mean-fluctuation coupling operator $L\left(\bar{u}_{t}\right)\in\mathbb{R}^{d\times d}$
dependent on the statistical mean state $\bar{u}_{t}$ is defined
as
\begin{equation}
L_{kl}\left(u\right)=\hat{v}_{k}\cdot\left[\Lambda\hat{v}_{l}+B\left(u,\hat{v}_{l}\right)+B\left(\hat{v}_{l},u\right)\right].\label{eq:quasi_linear}
\end{equation}
Notice that the right hand side of \eqref{eq:dyn_cov} involves the
fluctuation modes $Z_{t}$ defined from $u_{t}$ in \eqref{eq:decomp},
then the expectations on third moments are taken w.r.t. the PDF $p_{t}$.
Therefore, the resulting statistical equations \eqref{eq:dyn_stat}
are not closed and need to be combined with the FPE \eqref{eq:FPE}
to achieve a closed formulation for the leading-order mean and covariance
in the nonlinear system.

\subsubsection{Stochastic equations for the microscopic processes }

Second, we introduce the SDE describing the time evolution of the
multiscale stochastic processes $Z_{t}$ as the \emph{microscopic}
state consisting of the many subscale fluctuations
\begin{equation}
\mathrm{d}Z_{t}=L\left(\bar{u}_{t}\right)Z_{t}\mathrm{d}t+\varGamma\left(Z_{t}Z_{t}^{\intercal}-R_{t}\right)\mathrm{d}t+\Sigma_{t}\mathrm{d}W_{t}.\label{eq:dyn_stoc}
\end{equation}
Above, $L\left(\bar{u}_{t}\right)$ is the same mean-fluctuation coupling
operator defined in \eqref{eq:quasi_linear} involving the statistical
mean $\bar{u}_{t}$, and we define the quadratic coupling operator
$\varGamma:\mathbb{R}^{d\times d}\rightarrow\mathbb{R}^{d}$ as a
linear combination of the entries of the input matrix $R\in\mathbb{R}^{d\times d}$
describing the nonlinear multiscale coupling involving the covariance
$R_{t}$
\begin{equation}
\varGamma_{k}\left(R\right)=\sum_{m,n=1}^{d}\hat{v}_{k}\cdot B\left(\hat{v}_{m},\hat{v}_{n}\right)R_{mn}.\label{eq:coupling_coeff}
\end{equation}
Similar to the statistical equations, the dynamics on the right hand
side of \eqref{eq:dyn_stoc} is linked to the macroscopic quantities
$\bar{u}_{t}$ and $R_{t}$, which in tern requires additional information
of $p_{t}$ from the original model state $u_{t}$. This makes the
stochastic equation also unclosed requiring additional information
from the PDF solution \eqref{eq:FPE} from the original system. The
derivations of the above statistical and stochastic equations \eqref{eq:dyn_stat}
and \eqref{eq:dyn_stoc} from the original general formulation \eqref{eq:abs_formu}
are shown in the proof of Proposition~\ref{prop:model_equiv}. 

\subsection{A coupled stochastic-statistical closure model with explicit higher-order
feedbacks}

Combining the ideas in the closely related stochastic equation \eqref{eq:dyn_stoc}
and the statistical equations \eqref{eq:dyn_stat}, we propose a \emph{statistically
consistent stochastic-statistical closure model} based on the following
self-consistent coupling of the microscopic stochastic processes $Z_{t}$
and the macroscopic statistics $\bar{u},R_{t}$ 
\begin{equation}
\begin{aligned}\mathrm{d}Z_{t}= & \:L\left(\bar{u}_{t}\right)Z_{t}\mathrm{d}t+\varGamma\left(Z_{t}Z_{t}^{\intercal}-R_{t}\right)\mathrm{d}t+\Sigma_{t}\mathrm{d}W_{t},\\
\frac{\mathrm{d}\bar{u}_{t}}{\mathrm{d}t}= & \:M\left(\bar{u}_{t}\right)+Q_{m}\left(\mathbb{E}\left[Z_{t}\otimes Z_{t}\right]\right)+F_{t},\\
\frac{\mathrm{d}R_{t}}{\mathrm{d}t}= & \:L\left(\bar{u}_{t}\right)R_{t}+R_{t}L\left(\bar{u}_{t}\right)^{\intercal}+Q_{v}\left(\mathbb{E}\left[Z_{t}\otimes Z_{t}\otimes Z_{t}\right]\right)\\
 & +\Sigma_{t}\Sigma_{t}^{\intercal}+\epsilon^{-1}\left(\mathbb{E}\left[Z_{t}Z_{t}^{\intercal}\right]-R_{t}\right).
\end{aligned}
\label{eq:closure_model}
\end{equation}
Above, the expectations are all w.r.t. the PDF $\rho_{t}$ of the
stochastic states $Z_{t}$. In the first moment equation for $\bar{u}_{t}$,
with a bit abuse of notation, we denote $\bar{u}_{t}=\left[\bar{u}_{1,t},\cdots,\bar{u}_{d,t}\right]^{\intercal}\in\mathbb{R}^{d}$
with each component $\bar{u}_{k,t}=\bar{u}_{t}\cdot\hat{v}_{k}$,
$M=\left[M_{1},\cdots,M_{d}\right]^{\intercal}\in\mathbb{R}^{d}$
where $M_{k}\left(\bar{u}_{t}\right)=\sum_{p,q}\hat{v}_{k}\cdot\left[\Lambda\hat{v}_{p}\bar{u}_{p,t}+B\left(\hat{v}_{p},\hat{v}_{q}\right)\bar{u}_{p,t}\bar{u}_{q,t}\right]$
for $1\leq k\leq d$ and $F_{t}=\left[\hat{v}_{1}\cdot F_{t},\cdots,\hat{v}_{d}\cdot F_{t}\right]$;
in the second moment equation for $R_{t}\in\mathbb{R}^{d\times d}$,
the operator $L\left(\bar{u}_{t}\right)\in\mathbb{R}^{d\times d}$
indicates mean-fluctuation interactions defined in \eqref{eq:quasi_linear};
and in the stochastic equation for $Z_{t}$, $\varGamma:\mathbb{R}^{d\times d}\rightarrow\mathbb{R}^{d}$
is the quadratic coupling operator defined in \eqref{eq:coupling_coeff}.
The two higher-moment feedbacks for the mean and covariance, $Q_{m},Q_{v}$,
related to the second and third moments of $Z_{t}$ respectively,
are defined as
\begin{equation}
\begin{aligned}Q_{m,k}\left(\mathbb{E}\left[Z_{t}\otimes Z_{t}\right]\right)= & \sum_{p,q=1}^{d}\gamma_{kpq}\mathbb{E}\left[Z_{p,t}Z_{q,t}\right],\\
Q_{v,kl}\left(\mathbb{E}\left[Z_{t}\otimes Z_{t}\otimes Z_{t}\right]\right)= & \sum_{p,q=1}^{d}\gamma_{kpq}\mathbb{E}\left[Z_{p,t}Z_{q,t}Z_{l,t}\right]+\gamma_{lpq}\mathbb{E}\left[Z_{p,t}Z_{q,t}Z_{k,t}\right],
\end{aligned}
\label{eq:stat_feedbacks}
\end{equation}
for $1\leq k,l\leq d$ with coupling coefficients $\gamma_{kmn}=\hat{v}_{k}\cdot B\left(\hat{v}_{m},\hat{v}_{n}\right)$.
Above, $Q_{m}$ models the feedback in the mean equation due to the
second moments $\mathbb{E}\left[Z_{t}\otimes Z_{t}\right]$, and $Q_{v,kl}$
is the symmetric feedback in the covariance equation due to all the
third moments $\mathbb{E}\left[Z_{t}\otimes Z_{t}\otimes Z_{t}\right]$.
Notice that $Q_{m},Q_{v}$ can be both viewed as linear operators
w.r.t. $\rho_{t}$. Besides, in the second moment equation an additional
relaxation term with a parameter $\epsilon>0$ is introduced. This
term will not modify the model dynamics with statistical consistency,
but is playing a crucial role as a `reinforcement' term in maintaining
stable numerical performance (for example, see \cite{qi2023high})
especially with highly instability induced by the strong mean-fluctuation
coupling from $L\left(\bar{u}_{t}\right)$.

Different from the inherently coupled stochastic and statistical equations
\eqref{eq:dyn_stoc} and \eqref{eq:dyn_stat} through the PDF $p_{t}$
of the original model state $u_{t}$, the new closure model \eqref{eq:closure_model}
provides a clean self-consistent formulation for tractable theoretical
analysis and direct numerical implementations. A new PDF $\rho_{t}$
of the stochastic process $Z_{t}$ is introduced to close the system.
The statistical states $\bar{u}_{t},R_{t}$ are first treated as new
individual processes subject to higher-order moments w.r.t. $\rho_{t}$.
Then, the microscopic stochastic equation for $Z_{t}$ models the
high-dimensional multiscale process with explicit dependence on the
macroscopic states $\bar{u}_{t}$ and $R_{t}$. Thus, the whole system
is closed by the law of the stochastic state $Z_{t}$ itself. No additional
information about the intractable PDF $p_{t}$ will be needed for
solving the FPE of the original system. 

In the rest part of this section, we built precise link between the
new closure model \eqref{eq:closure_model} and the coupled equations
\eqref{eq:dyn_stoc} and \eqref{eq:dyn_stat} from the original system.
First, the following lemma provides the self-consistency in the leading
moments of the stochastic modes $Z_{k}$ and statistical states $\bar{u}_{t},R_{t}$
in \eqref{eq:closure_model}.
\begin{lem}
\label{lem:consist_model}With consistent initial conditions $\mathbb{E}\left[Z_{0}\right]=0$
and $\mathbb{E}\left[Z_{0}Z_{0}^{\intercal}\right]=R_{0}$, the leading
moments of the stochastic modes $Z_{t}$ of the closure model \eqref{eq:closure_model}
satisfy that for all $t>0$
\begin{equation}
\mathbb{E}\left[Z_{t}\right]=0,\quad\mathbb{E}\left[Z_{t}Z_{t}^{\intercal}\right]=R_{t},\label{eq:consistent}
\end{equation}
where the expectation is taken w.r.t. the PDF $\rho_{t}$ of $Z_{t}$,
and $R_{t}$ is the solution of the second moment equation in \eqref{eq:closure_model}. 
\end{lem}

The proof can be found through the direct application of It\^{o}'s
formula and we put the detailed proof in Appendix~\ref{sec:Detailed-proofs}.
Lemma~\ref{lem:consist_model} demonstrates that the law of $Z_{t}$
indeed recovers the same covariance of the fluctuation modes, while
it also contains more information in the higher-order statistics.
Furthermore, the closed coupled stochastic-statistical model generates
the same statistical solution as the original system \eqref{eq:abs_formu}.
The following proposition describes the statistical consistency between
the coupled model \eqref{eq:closure_model} and the original system
\eqref{eq:abs_formu}.
\begin{prop}
\label{prop:model_equiv}Assume that the solution $u_{t}$ of the
system \eqref{eq:abs_formu} has a continuous PDF $p_{t}$, and the
solution $\left\{ \bar{u}_{t},R_{t};Z_{t}\right\} $ of the stochastic-statistical
closure model \eqref{eq:closure_model} has the continuous PDF $\rho_{t}$
for $Z_{t}$ together with the deterministic solutions for $\bar{u}_{t}$
and $R_{t}$. Then from the same initial conditions, the two models
give the same statistical solution, that is, for all $t>0$
\begin{equation}
\mathbb{E}_{p_{t}}\left[u_{t}\right]=\bar{u}_{t},\quad\mathbb{E}_{p_{t}}\left[u_{t}^{\prime}u_{t}^{\prime\intercal}\right]=R_{t},\label{eq:consist_lower}
\end{equation}
where $u_{t}^{\prime}=u_{t}-\mathbb{E}_{p_{t}}\left[u_{t}\right]$.
Furthermore, for both models with any function $\varphi\in C^{2}\left(\mathbb{R}^{d}\right)$
we have
\begin{equation}
\mathbb{E}_{p_{t}}\left[\varphi\left(u_{t}^{\prime}\right)\right]=\mathbb{E}\left[\varphi\left(\sum_{k=1}^{d}Z_{k,t}\hat{v}_{k}\right)\right].\label{eq:consist_higher}
\end{equation}
\end{prop}

The proof of Proposition~\ref{prop:model_equiv} can be found in
Appendix~\ref{sec:Detailed-proofs} through detailed computation
of each moments of \eqref{eq:abs_formu} compared with that of \eqref{eq:closure_model}.
Notice that the left hand sides of \eqref{eq:consist_lower} and \eqref{eq:consist_higher}
consist of the statistics requiring solving the PDF $p_{t}$ of the
original system, while the right hand sides are purely from the PDF
$\rho_{t}$ of the closure model. This confirms the direct link between
the new stochastic-statistical closure model \eqref{eq:closure_model}
to the statistics in the original fully coupled multiscale system.
From this detailed formulation, quantification of non-Gaussian statistics
in $Z_{t}$ relies on the accurate estimation of the leading-order
mean and covariance, which can be assisted from the observation data.
This leads to effective filtering methods for improved probability
forecast for $Z_{t}$ that will be discussed next in Section~\ref{sec:Data-assimilation}.
\begin{rem*}
The second-order closure model \eqref{eq:closure_model} provides
explicit statistical dynamics for the mean and covariance. We can
also propose a first-order closure model involving only the statistical
mean equation coupled with the McKean-Vlasov SDE
\begin{equation}
\begin{aligned}\mathrm{d}Z_{t} & =L\left(\bar{u}_{t}\right)Z_{t}\mathrm{d}t+\varGamma\left(Z_{t}Z_{t}^{\intercal}-\mathbb{E}\left[Z_{t}Z_{t}^{\intercal}\right]\right)\mathrm{d}t+\Sigma_{t}\mathrm{d}W_{t},\\
\frac{\mathrm{d}\bar{u}_{t}}{\mathrm{d}t} & =M\left(\bar{u}_{t}\right)+\sum_{p,q}\mathbb{E}\left[Z_{p,t}Z_{q,t}\right]B\left(\hat{v}_{p},\hat{v}_{q}\right)+F_{t},
\end{aligned}
\label{eq:closure_mean}
\end{equation}
However, the above equations may introduce an unbalanced approximation
to the mean and fluctuation interactions since the dynamics of the
SDE for $Z_{t}$ will directly involve expectation w.r.t. its law
$\rho_{t}$. Thus, this model will be prone to numerical errors in
practical applications. Still, \eqref{eq:closure_mean} can serve
as an intermediate model for the analysis of filtering  schemes discussed
in the next section.
\end{rem*}

\subsection{The stochastic McKean-Vlasov equation as a multiscale interacting
system}

From the closed stochastic-statistical formulation \eqref{eq:closure_model},
the SDE for $Z_{t}$ is given by a stochastic McKean-Vlasov equation
depending on its own probability distribution $\rho_{t}$. In particular,
the resulting McKean-Vlasov SDE can be viewed naturally as the mean-field
limit of the ensemble approximation of $N$ individual trajectories
\begin{equation}
\mathrm{d}Z_{t}^{\left(i\right)}=L\left(\bar{u}_{t}^{N}\right)Z_{t}\mathrm{d}t+\varGamma\left(Z_{t}^{\left(i\right)}Z_{t}^{\left(i\right)\intercal}-R_{t}^{N}\right)\mathrm{d}t+\Sigma_{t}\mathrm{d}W_{t}^{\left(i\right)},\quad i=1,\cdots,N,\label{eq:ensemble_model}
\end{equation}
where $\left\{ W_{t}^{\left(i\right)}\right\} _{i=1}^{N}$ are independent
white noise processes, and the initial samples $\left\{ Z_{0}^{\left(i\right)}\right\} _{i=1}^{N}$
are drawn from the same initial measure $\rho_{0}$ of $Z_{t}$. Notice
that the ensemble members $Z_{t}^{\left(i\right)}$ as interacting
particles are not evolving independently with each other, but are
coupled through feedbacks of the leading-order statistics $\bar{u}_{t}^{N}$
and $R_{t}^{N}$ according to the statistical equations
\begin{equation}
\begin{aligned}\frac{\mathrm{d}\bar{u}_{t}^{N}}{\mathrm{d}t}= & \:M\left(\bar{u}_{t}^{N}\right)+Q_{m}\left(\mathbb{E}^{N}\left[\mathbf{Z}_{t}\otimes\mathbf{Z}_{t}\right]\right)+F_{t},\\
\frac{\mathrm{d}R_{t}^{N}}{\mathrm{d}t}= & \:L\left(\bar{u}_{t}^{N}\right)R_{t}^{N}+R_{t}^{N}L^{\intercal}\left(\bar{u}_{t}^{N}\right)+Q_{v}\left(\mathbb{E}^{N}\left[\mathbf{Z}_{t}\otimes\mathbf{Z}_{t}\otimes\mathbf{Z}_{t}\right]\right)\\
 & +\Sigma_{t}\Sigma_{t}^{\intercal}+\epsilon^{-1}\left(\mathbb{E}^{N}\left[\mathbf{Z}_{t}\mathbf{Z}_{t}^{\intercal}\right]-R_{t}^{N}\right).
\end{aligned}
\label{eq:ensemble_stat}
\end{equation}
Above, the expectations are computed through the empirical average
of the interacting particles $\mathbf{Z}_{t}=\left\{ Z_{t}^{\left(i\right)}\right\} _{i=1}^{N}$
from the ensemble simulation 
\[
\mathbb{E}^{N}\left[\varphi\left(\mathbf{Z}_{t}\right)\right]=\frac{1}{N}\sum_{i=1}^{N}\varphi\left(Z_{t}^{\left(i\right)}\right).
\]
The coupled ensemble approximation equations \eqref{eq:ensemble_model}
and \eqref{eq:ensemble_stat} have advantages in practical applications.
Unlike the general McKean-Vlasov SDEs \cite{mckean1967propagation},
\eqref{eq:ensemble_model} avoids the direct inclusion of the PDF
of $Z_{t}$, which is very difficult to approximate accurately from
finite particles. Instead, the PDF only implicitly enters the mean
and covariance equations \eqref{eq:ensemble_stat} through the computation
of higher moments. Effective computational algorithms with consistent
statistics then can be proposed (such as using the efficient random
batch methods \cite{qi2023high,qi2023random}) for the straightforward
ensemble model approximation. Besides in practical computation, the
relaxation term in $R_{t}^{N}$ provides additional restoring forcing
as a correction term to numerical errors with finite sample approximation
to reinforce stable dynamics and consistent statistics especially
in the case where internal instability is involved. 

In particular, it is well-know \cite{graham1996asymptotic} that the
empirical measure converges weakly to the true distribution, $\rho_{t}^{N}\rightarrow\rho_{t}$,
as well as the leading-order statistics in \eqref{eq:ensemble_model},
$\bar{u}^{N}\rightarrow\bar{u},R_{t}^{N}\rightarrow R_{t}$, as $N\rightarrow\infty$
under relatively weak assumptions (see Proposition~ \ref{thm:conv_stat}
in Section~\ref{sec:Ensemble-approximation}). The dynamical equation
for the continuous density function $\rho_{t}$ of $Z_{t}$ is given
by the corresponding equation
\begin{equation}
\frac{\partial\rho_{t}}{\partial t}=\mathcal{L}_{t}^{*}\left(\bar{u}_{t},R_{t}\right)\rho_{t}\coloneqq-\nabla_{z}\cdot\left[L\left(\bar{u}_{t}\right)z\rho_{t}\left(z\right)+\varGamma\left(zz^{\intercal}-R_{t}\right)\rho_{t}\left(z\right)\right]+\frac{1}{2}\nabla_{z}\cdot\left[\nabla_{z}\cdot\left(\Sigma_{t}\Sigma_{t}^{\intercal}\rho_{t}\left(z\right)\right)\right],\label{eq:dyn_pdf}
\end{equation}
where $\mathcal{L}_{t}^{*}$ is the adjoint of the generator $\mathcal{L}_{t}$
that is also dependent on the law of $Z_{t}$ shown in the statistics
of the mean $\bar{u}_{t}$ and covariance $R_{t}$. We postpone the
detailed analysis for the convergence of the ensemble approximation
model in Section~\ref{sec:Ensemble-approximation} together with
the filter approximation. 

In general, the probability density $\rho_{t}$ will demonstrate non-Gaussian
features due to the nonlinear stochastic coupling effects. On the
other hand, though the equations for the stochastic state $Z_{t}$
contain nonlinear structures, its PDF $\rho_{t}$ is subject to only
linear dynamics \eqref{eq:dyn_pdf} with $\mathcal{L}_{t}^{*}\left(\bar{u}_{t},R_{t}\right)$,
which is dependent on the mean and covariance. These desirable features
inspires the construction of effective filtering methods in the next
section to include leading-order statistical observations to improve
forecast of highly non-Gaussian statistics.

\section{Filtering models using observations in mean and covariance\protect\label{sec:Data-assimilation}}

In this section, we exploit the ideas in filtering to propose improved
approximation of probability distributions of the model state containing
highly non-Gaussian statistics. The optimal filter is developed by
combining the stochastic forecast model describing unobserved microscopic
states and leading-order statistics introduced as macroscopic observations.
We start with a precise description of the optimal filter equations
satisfying a conditional Gaussian process, then a new statistical
filtering model are proposed approximating the optimal filtering solution
with equivalent statistics.

\subsection{Filtering probability distributions using statistical observations}

We first formulate the filtering  problem for predicting probability
distributions based on observations from the leading-order statistics.
From the stochastic-statistical equations \eqref{eq:closure_model},
we can reformulate the general multiscale system \eqref{eq:abs_formu}
for $u_{t}$ as a composition of the macroscopic state from the first
two moments $\bar{u}_{t},R_{t}$ and the microscopic stochastic processes
$Z_{t}$. In practice, measurements are often available in the macroscopic
statistical states (such as the locally averaged mean state and variances
as the deviation from the mean). Therefore, it is natural to incorporate
the statistical observation data to improve the estimation of the
unobserved microscopic processes, especially to recover the unobserved
higher-order statistics (such as the deviation from the normal distribution
indicating the occurrence of high impact extreme events).

Using similar idea of the Fokker-Planck filter introduced in \cite{bach2023filtering},
we assume that the signal process of the filtering problem is from
the model state PDF $\rho_{t}\in\mathcal{P}\left(\mathbb{R}^{d}\right)$
belonging to the space of continuous probability density functions
on $\mathbb{R}^{d}$ with all finite moments. The observation process
is generated by the finite-dimensional statistical moments denoted
as $y_{t}\in\mathbb{R}^{p}$. This leads to the general \emph{infinite-dimensional
filtering system with statistical observations}\addtocounter{equation}{0}\begin{subequations}\label{eq:fpf}
\begin{align}
\mathrm{d}\rho_{t}= & \:\mathcal{L}_{t}^{*}\left(y_{t}\right)\rho_{t}\mathrm{d}t,\qquad\qquad\qquad\quad\rho_{t=0}\sim\mu_{0},\label{eq:fpf_model}\\
\mathrm{d}y_{t}= & \left[\mathcal{H}\rho_{t}+h_{t}\left(y_{t}\right)\right]\mathrm{d}t+\Gamma_{t}\mathrm{d}B_{t},\quad y_{t=0}=y_{0},\label{eq:fpf_obs}
\end{align}
\end{subequations}where $\mu_{0}$ is a probability measure of the
$\mathcal{P}\left(\mathbb{R}^{d}\right)$-valued random field $\rho$.
Given $y_{t}$, we will have $\rho_{t}\in\mathcal{P}\left(\mathbb{R}^{d}\right)$
for all $t>0$. Above, $\mathcal{L}_{t}\left(y_{t}\right)$ is the
infinitesimal generator of the corresponding SDE for $Z_{t}$ with
the explicit form given by \eqref{eq:dyn_pdf}; the general observation
process $y_{t}$ satisfies the dynamical equation subject to a linear
observation operator $\mathcal{H}:\mathcal{P}\left(\mathbb{R}^{d}\right)\rightarrow\mathbb{R}^{p}$
acting on the PDF $\rho_{t}$, as well as $h_{t}:\mathbb{R}^{p}\rightarrow\mathbb{R}^{p}$.
Detailed equations for $y_{t}=\left(\bar{u}_{t},R_{t}\right)$ will
be given next in \eqref{eq:dyn_obs} based on the first two moments.
Notice here $\rho_{t}$ becomes a random field (with a precise definition
given in \eqref{eq:def_rf}) due to the randomness in $y_{t}$ as
well as its initial uncertainty. 

\subsubsection*{Statistical observations from leading-order moments containing errors}

Let $p_{t}$ be the (unknown) PDF of the state $u_{t}$ in \eqref{eq:abs_formu},
that is, the deterministic solution of the FPE \eqref{eq:FPE}. Then,
we can assume that observations are drawn from the mean or covariance
of the state $u_{t}$ as
\begin{equation}
\bar{u}_{k,t}=\mathbb{E}_{p_{t}}\left[u_{t}\cdot\hat{v}_{k}\right],\quad R_{kl,t}=\mathbb{E}_{p_{t}}\left[\left(\hat{v}_{k}\cdot u_{t}^{\prime}\right)\left(u_{t}^{\prime}\cdot\hat{v}_{l}\right)\right],\label{eq:obs}
\end{equation}
projected to the observed large-scale modes $\hat{v}_{k},k\leq d^{\prime}$
in \eqref{eq:decomp}. We refer it as the \emph{full observation case}
with $d^{\prime}=d$ and \emph{partial observation case} with $d^{\prime}<d$.
For simplicity, we may always consider the full observation case $d^{\prime}=d$
(that is, $y_{t}=\left(\bar{u}_{t},R_{t}\right)\in\mathbb{R}^{p}$
with $p=d+d^{2}$) in this paper without confusion. The dynamical
equations for $\bar{u}_{t},R_{t}$ can be introduced according to
the closure model \eqref{eq:closure_model}. According to Proposition~\ref{prop:model_equiv},
the statistical equations for $\bar{u}_{t},R_{t}$ in \eqref{eq:closure_model}
provide statistical solutions consistent with the law $p_{t}$ of
the state $u_{t}$. 

On the other hand, model errors are usually introduced due to the
finite truncation of the originally infinite-dimensional continuous
system as well as measurement errors. Therefore, additional correction
terms, using independent white noises, $B_{m,t}$ and $B_{v,t}$,
are added to the equations \eqref{eq:dyn_stat} for observed statistics
accounting for errors from the imperfect model approximations. The
detailed equations for the observations \eqref{eq:fpf_obs} can be
rewritten according to \eqref{eq:dyn_stat} as the following SDEs
for $\bar{u}_{t}\in\mathbb{R}^{d},R_{t}\in\mathbb{R}^{d^{2}}$ 
\begin{equation}
\begin{aligned}\mathrm{d}\bar{u}_{t}= & \left[\mathcal{H}_{m}\rho_{t}+h_{m,t}\left(\bar{u}_{t}\right)\right]\mathrm{d}t\;+\Gamma_{m}\mathrm{d}B_{m,t},\\
\mathrm{d}R_{t}= & \left[\mathcal{H}_{v}\rho_{t}+h_{v,t}\left(\bar{u}_{t},R_{t}\right)\right]\mathrm{d}t+\Gamma_{v}\mathrm{d}B_{v,t}.
\end{aligned}
\label{eq:dyn_obs}
\end{equation}
where $h_{m,t}\left(\bar{u}\right)=M\left(\bar{u}\right)+F_{t}$ and
$h_{v,t}\left(\bar{u},R\right)=L\left(\bar{u}\right)R+RL\left(\bar{u}\right)^{\intercal}+\Sigma_{t}\Sigma_{t}^{\intercal}$
are deterministic functions, while the linear \emph{observation operators},
$\mathcal{H}_{m},\mathcal{H}_{v}$, are defined by the high-order
statistical feedback functions \eqref{eq:stat_feedbacks}
\begin{equation}
\begin{aligned}\mathcal{H}_{m,k}\rho=\int_{\mathbb{R}^{d}}H_{k}^{m}\left(z\right)\rho\left(z\right)\mathrm{d}z= & \sum_{p,q=1}^{d}\gamma_{kpq}\int_{\mathbb{R}^{d}}z_{p}z_{q}\rho\left(z\right)\mathrm{d}z,\\
\mathcal{H}_{v,kl}\rho=\int_{\mathbb{R}^{d}}H_{kl}^{v}\left(z\right)\rho\left(z\right)\mathrm{d}z= & \sum_{p,q=1}^{d}\left[\gamma_{kpq}\int_{\mathbb{R}^{d}}z_{p}z_{q}z_{l}\rho\left(z\right)\mathrm{d}z+\gamma_{lpq}\int_{\mathbb{R}^{d}}z_{p}z_{q}z_{k}\rho\left(z\right)\mathrm{d}z\right].
\end{aligned}
\label{eq:operators_obs}
\end{equation}
We assume $\rho\in\mathcal{P}\left(\mathbb{R}^{d}\right)$ the probability
density with finite moments in all orders, thus we have $\mathcal{H}_{m,k}\rho<\infty$
and $\mathcal{H}_{v,kl}\rho<\infty$ for all $k,l$. Then \eqref{eq:dyn_obs}
fits into the general observation equation \eqref{eq:fpf_obs} by
setting $y_{t}=\left(\bar{u}_{t},R_{t}\right)^{\intercal}\in\mathbb{R}^{p}$
with $p=d+d^{2}$ in a column vector, and letting $\mathcal{H}=\left(\mathcal{H}_{m},\mathcal{H}_{v}\right)^{\intercal}$,
$h_{t}=\left(h_{m,t},h_{v,t}\right)^{\intercal}$, and $\Gamma_{t}=\mathrm{diag}\left(\Gamma_{m},\Gamma_{v}\right)$.
With this explicit setup of the filtering problem, we will consider
the optimal filtering solution for the probability density $\rho_{t}$
of $Z_{t}$ based on the statistical observation data $Y_{t}=\left\{ \left(\bar{u}_{s},R_{s}\right),s\leq t\right\} $. 

\subsection{The optimal filter with conditional Gaussian structure}

Let $\left(\Omega,\mathcal{F},\mathbb{P}\right)$ be the complete
probability space, and denote $\mathcal{P}\left(\mathbb{R}^{d}\right)$
as the space of probability density functions with bounded all moments.
We first define the $\mathcal{P}\left(\mathbb{R}^{d}\right)$-valued
stochastic process $\rho_{t}$ as 
\begin{equation}
\rho_{t}:\mathbb{R}^{d}\times\Omega\rightarrow\mathbb{R}^{+},\;\left(z,\omega\right)\mapsto\rho_{t}\left(z;\omega\right),\;\mathrm{with}\;\rho_{t}\left(\cdot;\omega\right)\in\mathcal{P}\left(\mathbb{R}^{d}\right),\label{eq:def_rf}
\end{equation}
which is thereafter referred to as a \emph{random field}. In contrast
to the standard filtering problem concerning the nonlinear SDE of
the random model states $Z_{t}$, for derivation purpose of the exact
optimal equations, we lift the problem into filtering the random field
$\rho_{t}$ based on the observation information in $y_{s},s\leq t$
as in \eqref{eq:fpf}. A stochastic model on $\mathbb{R}^{d}$ \eqref{eq:mf-fpf}
will be then proposed for practical implementations next in Section~\ref{sec:A-statistically-consistent}.
Let $\mathcal{G}_{t}=\sigma\left\{ \omega:y_{s},s\leq t\right\} $
be the $\sigma$-algebra generated by the observations up to time
$t$. We define the space as the collection of $\mathcal{G}_{t}$-measurable
square-integrable random fields
\begin{equation}
\mathcal{V}_{t}\coloneqq L^{2}\left(\Omega,\mathcal{G}_{t},\mathbb{P};\mathcal{P}\left(\mathbb{R}^{d}\right)\cap L^{2}\left(\mathbb{R}^{d}\right)\right),\label{eq:def_condprob}
\end{equation}
satisfying $\int\left\Vert \nu\left(\cdot;\omega\right)\right\Vert _{L^{2}\left(\mathbb{R}^{d}\right)}^{2}\mathrm{d}\mathbb{P}\left(\omega\right)<\infty$
and $\nu\left(\cdot;\omega\right)\in\mathcal{P}\left(\mathbb{R}^{d}\right)\cap L^{2}\left(\mathbb{R}^{d}\right)$
for $\nu\in\mathcal{V}_{t}$. In this infinite-dimensional filtering
problem, we aim to find the optimal approximation of $\rho_{t}$ in
the space $\mathcal{V}_{t}$. The \emph{optimal filtering solution}
$\hat{\rho}_{t}$ is then introduced as the least-square estimate
with the minimum variance as
\begin{equation}
\hat{\rho}_{t}\coloneqq\underset{\nu\in\mathcal{V}_{t}}{\arg\min}\:\mathbb{E}\left[\left\Vert \rho_{t}-\nu\right\Vert _{L^{2}\left(\mathbb{R}^{d}\right)}^{2}\right]=\mathsf{P}_{\mathcal{V}_{t}}\left[\rho_{t}\right],\label{eq:optim_cost}
\end{equation}
where the optimal solution $\hat{\rho}_{t}$ can be viewed as the
unbiased projection of $\rho_{t}$ onto the space $\mathcal{V}_{t}$.
\eqref{eq:optim_cost} indicates that $\hat{\rho}_{t}$ gives the
PDF estimation closest to the true distribution $\rho_{t}$ in the
mean square sense in agreement with the observations in $\mathcal{G}_{t}$. 

Accordingly, we define the \emph{optimal filter distribution} $\mu_{t}:\mathcal{P}\left(\mathbb{R}^{d}\right)\times\Omega\rightarrow\left[0,1\right]$
as the regular conditional measure of the stochastic process $\rho_{t}$
given $\mathcal{G}_{t}$. That is, for any Borel set $A\in\mathcal{B}\left(\mathcal{P}\left(\mathbb{R}^{d}\right)\right)$,
$\mu_{t}$ gives the conditional probability of $\rho_{t}$ given
$\mathcal{G}_{t}$
\begin{equation}
\mu_{t}\left(A;\omega\right)\coloneqq\mathbb{P}\left(\rho_{t}\in A\mid\mathcal{G}_{t}\right)\left(\omega\right),\quad\mathrm{a.s.}\:\omega\in\Omega.\label{eq:optim_meas}
\end{equation}
Notice that $\mu_{t}\left(A;\cdot\right)\in\mathcal{G}_{t}$ is still
a stochastic process. For any functional $F\in C\left(\mathcal{P}\left(\mathbb{R}^{d}\right)\right)$
and $t>0$, we can introduce the conditional expectation w.r.t. the
measure $\mu_{t}$ given $\mathcal{G}_{t}$ as
\[
\mathbb{E}\left[F\left(\rho_{t}\right)\mid\mathcal{G}_{t}\right]\coloneqq\int_{\mathcal{P}\left(\mathbb{R}^{d}\right)}F\left(\rho\right)\mu_{t}\left(\mathrm{d}\rho\right).
\]
Therefore, the optimal filter solution \eqref{eq:optim_cost} is a
random field
\begin{equation}
\hat{\rho}_{t}=\mathbb{E}\left[\rho_{t}\mid\mathcal{G}_{t}\right]\label{eq:optim_m}
\end{equation}
given by the conditional expectation of $\rho_{t}$ w.r.t. $\mu_{t}$.
For any linear operator $\mathcal{M}:L^{2}\left(\mathbb{R}^{d}\right)\rightarrow\mathbb{R}^{p}$,
we have $\mathcal{M}\hat{\rho}_{t}=\mathbb{E}\left[\mathcal{M}\rho_{t}\mid\mathcal{G}_{t}\right]$.
Furthermore, second moment of $\mathcal{M}\rho_{t}$ is given by
\[
\mathbb{E}\left[\left(\mathcal{M}\rho_{t}-\mathcal{M}\hat{\rho}_{t}\right)\left(\mathcal{M}\rho_{t}-\mathcal{M}\hat{\rho}_{t}\right)^{\intercal}\mid\mathcal{G}_{t}\right]=\mathcal{M}\hat{\mathcal{C}}_{t}\mathcal{M}^{*},
\]
where $\hat{\mathcal{C}}_{t}\left(\omega\right):L^{2}\left(\mathbb{R}^{d};\mathbb{R}^{p}\right)\rightarrow L^{2}\left(\mathbb{R}^{d};\mathbb{R}^{p}\right)$
with $\hat{\mathcal{C}}_{t}^{*}=\hat{\mathcal{C}}_{t}$ is the self-adjoint
covariance operator for any $f\in L^{2}\left(\mathbb{R}^{d};\mathbb{R}^{p}\right)$
\begin{equation}
\hat{\mathcal{C}}_{t}f=\mathbb{E}\left[\left(\rho-\hat{\rho}_{t}\right)\int_{\mathbb{R}^{d}}\left(\rho-\hat{\rho}_{t}\right)\left(z\right)f\left(z\right)\mathrm{d}z\mid\mathcal{G}_{t}\right].\label{eq:optim_v}
\end{equation}
Notice again $\hat{\mathcal{C}}_{t}f\left(z;\cdot\right)$ is also
a random field conditional on $\mathcal{G}_{t}$. For clarification
of notations, we will call $\hat{\rho}_{t}$ and $\hat{\mathcal{C}}_{t}$
the optimal filter solution for the mean and covariance, and $\mu_{t}$
the optimal filter distribution in the rest part of the paper.

In particular, we can characterize the optimal filter solution $\hat{\rho}_{t}$
as the best estimate in each order of moments. The following result
describes the accuracy of the filter approximations under any finite-dimensional
projections.
\begin{prop}
\label{prop:optim_mom}Let $\rho_{t}$ be the random field from the
system \eqref{eq:fpf} and $\hat{\rho}_{t}=\mathbb{E}\left[\rho_{t}\mid\mathcal{G}_{t}\right]$
the optimal filter solution given the observations in $\mathcal{G}_{t}$.
For any linear operator $\mathcal{M}:\mathcal{P}\left(\mathbb{R}^{d}\right)\rightarrow\mathbb{R}^{p}$
defined by $\mathcal{M}\rho=\int M\left(z\right)\rho\left(\mathrm{d}z\right)$
with $M\in C\left(\mathbb{R}^{d};\mathbb{R}^{p}\right)$, $\mathcal{M}\hat{\rho}_{t}=\mathbb{E}\left[\mathcal{M}\rho_{t}\mid\mathcal{G}_{t}\right]$
gives the best unbiased estimate of $\mathcal{M}\rho_{t}$ in the
sense of minimum mean square error, that is,
\begin{equation}
\mathbb{E}\left[\left|\mathcal{M}\rho_{t}-\mathcal{M}\hat{\rho}_{t}\right|^{2}\right]=\min_{\nu\in\mathcal{V}_{t}}\mathbb{E}\left[\left|\mathcal{M}\rho_{t}-\mathcal{M}\nu\right|^{2}\right],\quad\mathrm{with}\quad\mathbb{E}\left[\mathcal{M}\hat{\rho}_{t}\right]=\mathbb{E}\left[\mathcal{M}\rho_{t}\right].\label{eq:optim_mom}
\end{equation}
\end{prop}

\begin{proof}
For any $\mathcal{G}_{t}$-measurable square-integrable stochastic
process $\nu$, we have from direct computation
\begin{align*}
 & \mathbb{E}\left[\left|\mathcal{M}\rho_{t}-\mathcal{M}\nu\right|^{2}\right]-\mathbb{E}\left[\left|\mathcal{M}\rho_{t}-\mathcal{M}\hat{\rho}_{t}\right|^{2}\right]\\
= & \mathbb{E}\left[\left(\mathcal{M}\hat{\rho}_{t}-\mathcal{M}\nu\right)\cdot\left(2\mathcal{M}\rho_{t}-\mathcal{M}\nu-\mathcal{M}\hat{\rho}_{t}\right)\right]\\
= & \mathbb{E}\left\{ \mathbb{E}\left[\left(\mathcal{M}\hat{\rho}_{t}-\mathcal{M}\nu\right)\cdot\left(2\mathcal{M}\rho_{t}-\mathcal{M}\nu-\mathcal{M}\hat{\rho}_{t}\right)\mid\mathcal{G}_{t}\right]\right\} \\
= & \mathbb{E}\left\{ \left(\mathcal{M}\hat{\rho}_{t}-\mathcal{M}\nu\right)\cdot\left[\mathbb{E}\left[2\mathcal{M}\rho_{t}\mid\mathcal{G}_{t}\right]-\left(\mathcal{M}\nu+\mathcal{M}\hat{\rho}_{t}\right)\right]\right\} \\
= & \mathbb{E}\left[\left|\mathcal{M}\hat{\rho}_{t}-\mathcal{M}\nu\right|^{2}\right]\geq0.
\end{align*}
Above, the third equality uses the fact $\mathcal{M}\hat{\rho}_{t}-\mathcal{M}\nu=\mathbb{E}\left[\mathcal{M}\rho_{t}\mid\mathcal{G}_{t}\right]-\mathcal{M}\nu$
is $\mathcal{G}_{t}$-measurable. Thus, we get $\mathcal{M}\hat{\rho}_{t}$
minimizes the mean square error. The consistency under the expectation
in $\mathcal{M}\hat{\rho}_{t}$ and $\mathcal{M}\rho_{t}$ can be
directly implied by definition. 
\end{proof}
By taking the operator $\mathcal{M}$ as the expectation on $M\left(Z_{t}\right)=Z_{t}^{m}$
with any integer $m$, $\mathcal{M}\rho_{t}$ and $\mathcal{M}\hat{\rho}_{t}$
give the $m$-th order moments of $Z_{t}$ under the random field
$\rho_{t}$ in \eqref{eq:def_rf} and optimal filter approximation
$\hat{\rho}_{t}$. A direct implication from \eqref{eq:optim_mom}
shows that we have the unbiased statistics in all high-order moments
$\mathbb{E}\left[\mathcal{M}\hat{\rho}_{t}\right]=\mathbb{E}\left[\mathcal{M}\rho_{t}\right]$
with the minimum error $\mathbb{E}\left[\left|\mathcal{M}\rho_{t}-\mathcal{M}\hat{\rho}_{t}\right|^{2}\right]$
from the optimal filter solution.

Importantly, the model equations \eqref{eq:fpf} satisfy the desirable
conditional Gaussian process \cite{liptser2013statistics}, that is,
given the observations of $Y_{t}=\left\{ y_{s}=\left(\bar{u}_{s},R_{s}\right),s\leq t\right\} $
and Gaussian initial state $\rho_{0}$, the random field $\rho_{t}$
follows a Gaussian distribution at each time $t$. Let $\rho_{t}=\rho_{t}\left(\cdot;\omega\right)$
be the (unknown) signal state satisfying linear dynamics \eqref{eq:fpf_model},
and $y_{t}=y_{t}\left(\omega\right)$ the observed statistical process
subject to linear observation operators \eqref{eq:fpf_obs}. The optimal
filter distribution $\mu_{t}$ \eqref{eq:optim_meas} conditional
on $Y_{t}$ then becomes an infinite-dimensional Gaussian distribution,
$\mu_{t}\left(\cdot,\omega\right)=\mathcal{N}\left(\hat{\rho}_{t},\hat{\mathcal{C}}_{t}\right)\left(\omega\right)$,
where the mean $\hat{\rho}_{t}\left(\cdot;\omega\right)$ and covariance
$\hat{\mathcal{C}}_{t}\left(\omega\right)$ give the solution to \eqref{eq:optim_m}
and \eqref{eq:optim_v} respectively. Therefore, the equations for
the mean and covariance are given by the generalized version of Kalman-Bucy
(KB) filter for the infinite-dimensional conditional Gaussian processes 

\begin{equation}
\begin{aligned}\mathrm{d}\hat{\rho}_{t}= & \:\mathcal{L}_{t}^{*}\left(\bar{u}_{t},R_{t}\right)\hat{\rho}_{t}\mathrm{d}t+\mathcal{\hat{C}}_{t}\mathcal{H}_{m}^{*}\Gamma_{m}^{-2}\left\{ \mathrm{d}\bar{u}_{t}-\left[\mathcal{H}_{m}\hat{\rho}_{t}+h_{m,t}\left(\bar{u}_{t}\right)\right]\mathrm{d}t\right\} \\
 & +\hat{\mathcal{C}}_{t}\mathcal{H}_{v}^{*}\Gamma_{v}^{-2}\left\{ \mathrm{d}R_{t}-\left[\mathcal{H}_{v}\hat{\rho}_{t}+h_{v,t}\left(\bar{u},R\right)\right]\mathrm{d}t\right\} ,\\
\mathrm{d}\hat{\mathcal{C}}_{t}= & \left[\mathcal{L}_{t}^{*}\left(\bar{u}_{t},R_{t}\right)\mathcal{\hat{C}}_{t}+\mathcal{\hat{C}}_{t}\mathcal{L}_{t}\left(\bar{u}_{t},R_{t}\right)\right]\mathrm{d}t-\mathcal{\hat{C}}_{t}\left(\mathcal{H}_{m}^{*}\Gamma_{m}^{-2}\mathcal{H}_{m}+\mathcal{H}_{v}^{*}\Gamma_{v}^{-2}\mathcal{H}_{v}\right)\mathcal{\hat{C}}_{t}\mathrm{d}t.
\end{aligned}
\label{eq:KB-fpf}
\end{equation}
The equations of the conditional Gaussian processes and the uniqueness
of the solutions are developed in Chapter 12 of \cite{liptser2013statistics}
for finite dimensional systems. The results are then generalized to
infinite-dimensional Hilbert space \cite{curtain1975survey,falb1967infinite}
(see a summary of the main results in Appendix~\ref{sec:Background-about-filtering}).
The system \eqref{eq:KB-fpf} gives a closed set of coupled SPDEs
(due to the randomness in $\bar{u}_{t},R_{t}$) for the optimal filter
solution $\hat{\rho}_{t}$ enabling more detailed analysis and development
of practical methods for computing the optimal solution.
\begin{rem*}
The filtering problem using statistical observations is originally
introduced as the ensemble Fokker-Planck filter in \cite{bach2023filtering}.
We propose the filtering equations \eqref{eq:KB-fpf} as a further
generalization where nonlinearly coupled conditional processes are
involved. In particular, this new filtering model directly fits into
the coupled stochastic-statistical modeling framework \eqref{eq:closure_model}.
\end{rem*}

\subsection{A surrogate filtering model for approximating the optimal filter solution\protect\label{subsec:A-statistical-filtering}}

The resulting optimal filtering problem from \eqref{eq:fpf} requires
to solve the infinite-dimensional system \eqref{eq:KB-fpf} concerning
the function $\hat{\rho}_{t}\in\mathcal{P}\left(\mathbb{R}^{d}\right)$
and operator $\hat{\mathcal{C}}_{t}$ on $\mathcal{P}\left(\mathbb{R}^{d}\right)$.
It usually becomes intractable in finding such infinite-dimensional
solutions from direct methods. In developing practical strategies
to realize the optimal filter solution, it is more useful to find
a surrogate model for the stochastic process $\tilde{Z}_{t}$, based
on which effective ensemble-based approaches can be built. Therefore,
we aim to construct an approximating filtering model from designing
a new dynamical equation for $\tilde{Z}_{t}$, whose PDF $\tilde{\rho}_{t}$
can effectively represent that of the optimal filter solution $\hat{\rho}_{t}$.

\subsubsection{Filtering updating cycle in a split two-step procedure}

For a clear characterization of the filtering process, we follow the
general procedure in \cite{calvello2022ensemble} to first describe
the filtering process by concatenated iterations of transporting maps
on the corresponding probability distribution. We propose a new stochastic
process $\tilde{Z}_{t}$, whose law $\tilde{\rho}_{t}\in\mathcal{V}_{t}$
is still a $\mathcal{P}\left(\mathbb{R}^{d}\right)$-valued random
field in $\mathcal{G}_{t}$ dependent on the same statistical observation
$y_{t}$ as in the optimal filter \eqref{eq:KB-fpf}. Thus, the filtering
updating cycle during the time interval $\left[t,t+\tau\right]$ can
be characterized by the transport of the probability density $\tilde{\rho}_{t}$
of $\tilde{Z}_{t}$. 

The discrete time update of the approximation filter PDF $\tilde{\rho}_{t}$
is carried out in a split two-step procedure. First, the \emph{forecast
step} can be viewed as the push-forward operator acting on the probability
density function at time instant $t$ with time step $\tau>0$
\begin{equation}
\tilde{\rho}_{t}\rightarrow\tilde{\rho}_{t+\tau}^{-}\coloneqq\mathcal{F}_{t}^{\tau}\tilde{\rho}_{t}=e^{\int_{t}^{t+\tau}\mathcal{L}_{s}^{*}\left(y_{s}\right)\mathrm{d}s}\tilde{\rho}_{t},\label{eq:for_op}
\end{equation}
where $\mathcal{F}_{t}^{\tau}$ represents the forecast updating operator
with forward time step $\tau$, and $\mathcal{L}_{t}\left(y_{t}\right)$
is the same generator as in \eqref{eq:dyn_pdf}. Second, the \emph{analysis
step} updates the prior distribution $\tilde{\rho}^{-}$ to the posterior
distribution $\tilde{\rho}^{+}$ by incorporating the observation
data up to $Y_{t+\tau}=\left\{ y_{s},s\leq t+\tau\right\} $, that
is
\begin{equation}
\tilde{\rho}_{t+\tau}^{-}\rightarrow\tilde{\rho}_{t+\tau}^{+}\coloneqq\mathcal{A}_{t}^{\tau}\left(\tilde{\rho}_{t+\tau}^{-};Y_{t+\tau}\right),\label{eq:ana_op}
\end{equation}
where $\mathcal{A}_{t}^{\tau}$ represents the analysis updating operator.
Therefore, the full filtering cycle from $t$ to $t+\tau$ can be
summarized as the composition of the forecast and analysis maps
\begin{equation}
\tilde{\rho}_{t+\tau}=\mathcal{A}_{t}^{\tau}\left(\mathcal{F}_{t}^{\tau}\tilde{\rho}_{t};Y_{t+\tau}\right).\label{eq:filter_op}
\end{equation}
Notice that $\mathcal{F}_{t}^{\tau}$ is a linear operator on $\tilde{\rho}_{t}$,
while $\mathcal{A}_{t}^{\tau}$ could contain nonlinear actions due
to the normalization of the probability distribution. Continuous equation
for $\partial_{t}\tilde{\rho}_{t}=\lim_{\tau\rightarrow0}\frac{1}{\tau}\left(\tilde{\rho}_{t+\tau}-\tilde{\rho}_{t}\right)$
is then achieved by letting the discrete time step $\tau\rightarrow0$.
Next, we first propose the general structure of the new filtering
model for $\tilde{Z}_{t}\sim\tilde{\rho}_{t}$ as a combination of
the above two-step procedure, then detailed analysis can be done according
to the design of the forecast and analysis step operators $\mathcal{F}_{t}^{\tau}$
and $\mathcal{A}_{t}^{\tau}$ accordingly. 

\subsubsection{Construction of equivalent statistical approximating filter}

For simplicity of notations, we still use the general statistical
observation processes \eqref{eq:dyn_obs} for $y_{t}=\left(\bar{u}_{t},R_{t}\right)$
taking the compact form
\[
\mathrm{d}y_{t}=\left[\mathcal{H}\rho_{t}+h_{t}\left(y_{t}\right)\right]\mathrm{d}t+\Gamma_{t}\mathrm{d}B_{t},
\]
where the general observation operator \eqref{eq:operators_obs},
$\mathcal{H}\rho_{t}=\int H\left(z\right)\rho_{t}\left(z\right)\mathrm{d}z$,
is defined with the general observation function $H\in C\left(\mathbb{R}^{d};\mathbb{R}^{p}\right)$
acting on the density function $\rho_{t}$. Following the general
construction in \cite{crisan2010approximate,pathiraja2021mckean},
we seek the approximating filter model adopting the following McKean-Vlasov
representation with undetermined functionals $a_{t},K_{t}$ 
\begin{equation}
\begin{aligned}\mathrm{d}\tilde{Z}_{t}= & \:L\left(\bar{u}_{t}\right)\tilde{Z}_{t}\mathrm{d}t+\varGamma\left(\tilde{Z}_{t}\tilde{Z}_{t}^{\intercal}-R_{t}\right)\mathrm{d}t+\Sigma_{t}\mathrm{d}\tilde{W}_{t}\\
+ & a_{t}\left(\tilde{Z}_{t};\tilde{\rho}_{t}\right)\mathrm{d}t+K_{t}\left(\tilde{Z}_{t};\tilde{\rho}_{t}\right)\left\{ \mathrm{d}y_{t}-\left[H\left(\tilde{Z}_{t}\right)+h_{t}\left(y_{t}\right)\right]\mathrm{d}t-\Gamma_{t}\mathrm{d}\tilde{B}_{t}\right\} .
\end{aligned}
\label{eq:mf-fpf}
\end{equation}
We use the name `statistical filtering' to refer the above new model
emphasizing our main goal of filtering statistical moments different
from the common filtering case. The first row of the above equation
models the forecast step of the filtering process, while the second
row is the analysis step. The forecast step accepts the same dynamical
model of \eqref{eq:fpf_model} dependent on the mean and covariance
$\left(\bar{u}_{t},R_{t}\right)$. On the other hand, the analysis
step serves as an additional control correction over statistical observations
$y_{t}$. New functionals known as the drift $a_{t}:\mathbb{R}^{d}\times\mathcal{P}\left(\mathbb{R}^{d}\right)\rightarrow\mathbb{R}^{d}$
and the control gain operator $K_{t}:\mathbb{R}^{d}\times\mathcal{P}\left(\mathbb{R}^{d}\right)\rightarrow\mathbb{R}^{d\times p}$
are introduced, resulting in an approximating filtering model about
the process $\tilde{Z}_{t}$. Most importantly, as will be shown next
in Theorem~\ref{thm:consist_analy} under proper condition, $a_{t}$
and $K_{t}$ are only implicitly dependent on $\tilde{\rho}_{t}$
through the leading moments, without the need to compute the (potentially
highly non-Gaussian) density function $\tilde{\rho}_{t}$ explicitly.

In a more clear identification of the filtering updates involving
several levels of approximations, we take the split-step strategy
to analyze the coupled forecast step and analysis step of the filtering
equation \eqref{eq:mf-fpf} separately. In particular, the forecast
step in the first row of \eqref{eq:mf-fpf} is given by the exactly
same form as the stochastic-statistical closure equations \eqref{eq:closure_model}
developed in Section~\ref{sec:A-statistically-consistent}. Thus
in practice, the updating step with the forecast operator can be implemented
adopting the efficient uncertainty prediction methods such as \cite{majda2018strategies,qi2023high}.
Then, the remaining task is to propose proper analysis step update
in the second line of \eqref{eq:mf-fpf} concerning consistent statistics
with the optimal solution $\hat{\rho}_{t}$ in \eqref{eq:KB-fpf}.

\subsection{Statistical consistency in analysis step update of the approximating
filter\protect\label{subsec:Statistical-consistency}}

Now, we focus on updating posterior PDF $\tilde{\rho}_{t}$ in \eqref{eq:ana_op}
of the proposed approximating filter \eqref{eq:mf-fpf} based on the
statistical observation $y_{t}$ satisfying \eqref{eq:dyn_obs}. Concentrating
on the analysis step, the resulting optimal filter equations \eqref{eq:KB-fpf}
for the mean and covariance $\left(\hat{\rho}_{t},\hat{\mathcal{C}}_{t}\right)$
become
\begin{equation}
\begin{aligned}\mathrm{d}\hat{\rho}_{t}= & \:\mathcal{\hat{C}}_{t}\mathcal{H}^{*}\Gamma_{t}^{-2}\left\{ \mathrm{d}y_{t}-\left[\mathcal{H}\hat{\rho}_{t}+h_{t}\left(y_{t}\right)\right]\mathrm{d}t\right\} ,\\
\mathrm{d}\mathcal{\hat{C}}_{t}= & -\hat{\mathcal{C}}_{t}\mathcal{H}^{*}\Gamma_{t}^{-2}\mathcal{H}\mathcal{\hat{C}}_{t}\mathrm{d}t.
\end{aligned}
\label{eq:KB-quad}
\end{equation}
Correspondingly, the approximating statistical filtering model for
$\tilde{Z}_{t}$ satisfies the second line of the SDE \eqref{eq:mf-fpf}
as
\begin{equation}
\begin{aligned}\mathrm{d}\tilde{Z}_{t}= & \:a_{t}\left(\tilde{Z}_{t}\right)\mathrm{d}t+K_{t}\left(\tilde{Z}_{t}\right)\left\{ \mathrm{d}y_{t}-\left[H\left(\tilde{Z}_{t}\right)+h_{t}\left(y_{t}\right)\right]\mathrm{d}t-\Gamma_{t}\mathrm{d}\tilde{B}_{t}\right\} .\end{aligned}
\label{eq:mf-fpf-quad}
\end{equation}
Following the similar idea in the McKean-Vlasov representation of
the filtering equation \cite{pathiraja2021mckean,yang2014continuous},
we expect the PDF $\tilde{\rho}_{t}$ of $\tilde{Z}_{t}$ to satisfy
the following Kushner-Stratonovich-type equation (with requirements
on $a_{t},K_{t}$ given next in \eqref{eq:coeff_sde})
\begin{equation}
\frac{\partial\tilde{\rho}_{t}}{\partial t}=\left[H\left(z\right)-\mathcal{H}\tilde{\rho}_{t}\right]^{\intercal}\Gamma_{t}^{-2}\left[\frac{\mathrm{d}y_{t}}{\mathrm{d}t}-\mathcal{H}\tilde{\rho}_{t}-h_{t}\left(y_{t}\right)\right]\tilde{\rho}_{t}.\label{eq:KS-fpf-quad}
\end{equation}
Again, the goal here is to approximate the optimal filter mean $\hat{\rho}_{t}$
in \eqref{eq:KB-quad} by $\tilde{\rho}_{t}$ generated by the surrogate
SDE model \eqref{eq:mf-fpf-quad} in the sense of consistent statistics. 

Unfortunately, the approximation \eqref{eq:KS-fpf-quad} and the optimal
filtering equation \eqref{eq:KB-quad} will in general have different
continuous solutions for $\tilde{\rho}_{t}$ and $\hat{\rho}_{t}$
due to their distinctive dynamics. In order to compare the statistical
moments of the two distributions, we apply the linear operator $\mathcal{H}$
to the optimal equations \eqref{eq:KB-quad} as a finite-dimensional
projection on leading moments. The resulting equations optimal mean
and covariance equations become finite dimensional as
\begin{equation}
\begin{aligned}\mathrm{d}\left(\mathcal{H}\hat{\rho}_{t}\right)= & \left(\mathcal{H}\mathcal{\hat{C}}_{t}\mathcal{H}^{*}\right)\Gamma_{t}^{-2}\left\{ \mathrm{d}y_{t}-\left[\mathcal{H}\hat{\rho}_{t}+h_{t}\left(y_{t}\right)\right]\mathrm{d}t\right\} ,\\
\mathrm{d}\left(\mathcal{H}\mathcal{\hat{C}}_{t}\mathcal{H}^{*}\right)= & -\left(\mathcal{H}\mathcal{\hat{C}}_{t}\mathcal{H}^{*}\right)\Gamma_{t}^{-2}\left(\mathcal{H}\mathcal{\hat{C}}_{t}\mathcal{H}^{*}\right)\mathrm{d}t.
\end{aligned}
\label{eq:dyn_h_opt}
\end{equation}
Above, remind that the observation operator $\mathcal{H}:L^{2}\left(\mathbb{R}^{d}\right)\rightarrow\mathbb{R}^{p}$
and its adjoint $\mathcal{H}^{*}:\mathbb{R}^{p}\rightarrow L^{2}\left(\mathbb{R}^{d}\right)$
are defined based on the observation function $H\in C_{b}^{2}\left(\mathbb{R}^{d};\mathbb{R}^{p}\right)$
as
\[
\mathcal{H}\rho=\int H\left(z\right)\rho\left(z\right)\mathrm{d}z,\quad\left[\mathcal{H}^{*}u\right]\left(z\right)=u\cdot H\left(z\right),
\]
and the covariance operator $\hat{\mathcal{C}}_{t}:L^{2}\left(\mathbb{R}^{d}\right)\rightarrow L^{2}\left(\mathbb{R}^{d}\right)$
is defined in \eqref{eq:optim_v}. Therefore, \eqref{eq:dyn_h_opt}
gives the equations for the finite-dimensional quantities $\mathcal{H}\hat{\rho}_{t}\in\mathbb{R}^{p}$
and $\mathcal{H}\mathcal{\hat{C}}_{t}\mathcal{H}^{*}\in\mathbb{R}^{p\times p}$
as the first two moments of $H$ w.r.t. $\hat{\rho}_{t}$. The idea
is to design the analysis step operator $\mathcal{A}_{t}^{\tau}$
for the approximating filter process $\tilde{Z}_{t}$ in \eqref{eq:mf-fpf-quad},
so that consistency in the first and second-order moments can be achieved.

Denote the expectation, $\tilde{\mathbb{E}}\left[\cdot\right]\coloneqq\mathbb{E}\left[\cdot\mid\mathcal{G}_{t}\right]=\mathbb{E}_{\tilde{\rho}_{t}}\left[\cdot\right]$,
w.r.t. the conditional density $\tilde{\rho}_{t}$ in \eqref{eq:KS-fpf-quad}
given the same observation process in $\mathcal{G}_{t}$. We define
$\bar{H}_{t}=\tilde{\mathbb{E}}H\left(\tilde{Z}_{t}\right)=\int H\left(z\right)\tilde{\rho}_{t}\left(z\right)\mathrm{d}z$
and $C_{t}^{H}=\tilde{\mathbb{E}}\left[H_{t}\left(\tilde{Z}_{t}\right)-\bar{H}_{t}\right]\left[H_{t}\left(\tilde{Z}_{t}\right)-\bar{H}_{t}\right]^{\intercal}$
as the first and second-order moments of $H$ w.r.t. $\tilde{\rho}_{t}$
. Assume that the drift $a_{t}$ and gain $K_{t}$ in the SDE approximation
\eqref{eq:mf-fpf-quad} satisfy the following identities
\begin{equation}
a_{t}=\nabla\cdot\left(K_{t}\Gamma_{t}^{2}K_{t}^{\intercal}\right)-K_{t}\Gamma_{t}^{2}\nabla\cdot K_{t}^{\intercal},\quad-\nabla\cdot\left(K_{t}^{\intercal}\tilde{\rho}_{t}\right)=\tilde{\rho}_{t}\Gamma_{t}^{-2}\left(H\left(z\right)-\tilde{\mathbb{E}}H\right),\label{eq:coeff_sde}
\end{equation}
where the divergence on a matrix is defined columnwise as $\left(\nabla\cdot A\right)_{i}=\sum_{j}\partial_{z_{j}}A_{ij}.$We
first have the following result concerning the evolution equations
of $\bar{H}_{t}$ and $C_{t}^{H}$ given the realization $Y_{t}=\left\{ y_{s},s\leq t\right\} $.
\begin{lem}
\label{lem:analysis}Given that $\Gamma_{t}\succ0$ in \eqref{eq:fpf}
and the identities \eqref{eq:coeff_sde} are satisfied, the evolution
equations for the mean and covariance of the observation function
$H\left(\tilde{Z}_{t}\right)$ associated with the SDE \eqref{eq:mf-fpf-quad}
are given by
\begin{equation}
\begin{aligned}\mathrm{d}\bar{H}_{t}= & \:C_{t}^{H}\Gamma_{t}^{-2}\left\{ \mathrm{d}y_{t}-\left[\bar{H}_{t}+h_{t}\left(y_{t}\right)\right]\mathrm{d}t\right\} ,\\
\mathrm{d}C_{t}^{H}= & \:Q_{t}^{H}\Gamma_{t}^{-2}\left\{ \mathrm{d}y_{t}-\left[\bar{H}_{t}+h_{t}\left(y_{t}\right)\right]\mathrm{d}t\right\} -C_{t}^{H}\Gamma_{t}^{-2}C_{t}^{H}\mathrm{d}t,
\end{aligned}
\label{eq:dyn_h}
\end{equation}
where $Q_{t}^{H}:\mathbb{R}^{p}\rightarrow\mathbb{R}^{p\times p}$
is defined as
\[
Q_{t}^{H}=\tilde{\mathbb{E}}\left[\left(H_{t}^{\prime}H_{t}^{\prime\intercal}\right)\otimes H_{t}^{\prime\intercal}\right],
\]
containing third moments of $H_{t}^{\prime}=H\left(\tilde{Z}_{t}\right)-\bar{H}_{t}$.
\end{lem}

We put the detailed derivation of \eqref{eq:dyn_h} in Appendix~\ref{sec:Detailed-proofs}.
Notice that \eqref{eq:dyn_h} goes back to the Kalman-Bucy filter
if we set linear observation $H\left(\tilde{Z}_{t}\right)=\tilde{Z}_{t}$
satisfying a normal distribution as in \cite{bach2023filtering,crisan2014numerical}.
Although here we are considering the more general nonlinear dynamics
and observation functions from \eqref{eq:operators_obs}. 

Comparing \eqref{eq:dyn_h} and \eqref{eq:dyn_h_opt} implies that
the same statistical solution can be reached in $\left(\bar{H}_{t},C_{t}^{H}\right)$
and $\left(\mathcal{H}\hat{\rho}_{t},\mathcal{H}\mathcal{\hat{C}}_{t}\mathcal{H}^{*}\right)$
if we have $Q_{t}^{H}=0$. In order to achieve this, we further introduce
the projection operator on the space of probability distributions
using the Kullback-Leibler (KL) divergence \cite{kullback1951annals}
as an unbiased metric.
\begin{defn}
\label{def:symm_op}Define the operator $\mathcal{S}_{H}$ making
symmetric projection on the probability density $\rho\in\mathcal{P}\left(\mathbb{R}^{d}\right)$
\begin{equation}
\mathcal{S}_{H}\rho=\underset{\nu\in\mathcal{V}_{H}}{\arg\min}\:d_{\mathrm{KL}}\left(\nu\parallel\rho\right),\label{eq:symm_proj}
\end{equation}
where $d_{\mathrm{KL}}$ is the KL divergence between two probability
measures. The minimization is among the probability measures in the
following set
\[
\mathcal{V}_{H}\left(\bar{H},C^{H}\right)=\left\{ \nu\in\mathcal{P}\left(\mathbb{R}^{d}\right):\mathbb{E}_{\nu}H=\bar{H},\;\mathbb{E}_{\nu}\left[H^{\prime}H^{\prime\intercal}\right]=C^{H},\;\mathrm{and}\;\mathbb{E}_{\nu}\left[H_{l}^{\prime}H_{m}^{\prime}H_{n}^{\prime}\right]=0\right\} ,
\]
for all $l,m,n\leq d$ and $H^{\prime}\left(Z\right)=H\left(Z\right)-\mathbb{E}_{\nu}\left[H\left(Z\right)\right]$.
\end{defn}

In Definition~\ref{def:symm_op}, $\mathcal{S}_{H}$ acts as a symmetric
approximation of probability measures with vanishing third-order moments
of the observation function $H$, while maintains consistent first
two leading moments of $H$. It is clear that given $\bar{H},C^{H}$,
the set $\mathcal{V}_{H}$ is closed with respect to weak convergence
of measures. From Proposition 2.1 of \cite{pinski2015kullback} and
\cite{dupuis2011weak}, we have for any $\rho$ and weakly convergent
sequence $\left\{ \nu_{n}\right\} $ to $\nu_{*}$
\[
\underset{n\rightarrow\infty}{\lim\inf}\:d_{\mathrm{KL}}\left(\nu_{n}\parallel\rho\right)\geq d_{\mathrm{KL}}\left(\nu_{*}\parallel\rho\right).
\]
It follows immediately that there exists $\nu_{*}\in\mathcal{V}_{H}$
that reaches the minimum. Therefore, we have the following lemma guaranteeing
the existence of the minimizer in the proposed projection \eqref{eq:symm_proj}. 
\begin{lem}
\label{lem:kl_proj}Assume that there is one $\nu\in\mathcal{V}_{H}$
such that the KL-divergence $d_{\mathrm{KL}}\left(\nu\parallel\rho\right)<\infty$
for given $\rho\in\mathcal{P}\left(\mathbb{R}^{d}\right)$. Then a
minimizer exists in \eqref{eq:symm_proj}.
\end{lem}

Even though Lemma~\ref{lem:kl_proj} does not guarantee the uniqueness
of the minimizer, we can always find one minimizer given the same
mean and covariance $\bar{H}=\tilde{\mathbb{E}}\left[H\left(Z\right)\right],C^{H}=\tilde{\mathbb{E}}\left[H^{\prime}\left(Z\right)H^{\prime}\left(Z\right)^{\intercal}\right]$.
This provides the desirable target density function satisfying the
required symmetric statistics about $H\left(Z_{t}\right)$. Denote
the push-forward operator for the new SDE \eqref{eq:mf-fpf-quad}
with the structure functions \eqref{eq:coeff_sde} as $\tilde{\rho}_{t+s}=\mathcal{Q}_{t}^{s}\left(\tilde{\rho}_{t};Y_{t+s}\right)$
for any $s\geq0$. Under the above construction, we can finally propose
the forward operator in analysis step \eqref{eq:ana_op} as 
\begin{equation}
\mathcal{A}_{t}^{s}\left(\tilde{\rho}_{t}\right)\coloneqq\mathcal{Q}_{t}^{s}\left(\mathcal{S}_{H}\tilde{\rho}_{t};Y_{t+s}\right),\label{eq:analysis_update}
\end{equation}
where the projection $\mathcal{S}_{H}$ in \eqref{eq:symm_proj} is
a linear operator acting on the random fields in space $\mathcal{V}_{t}$
in \eqref{eq:def_condprob}. In a similar fashion as in the proof
of Lemma~\ref{lem:analysis}, by applying the addition projection
$\mathcal{S}_{H}\tilde{\rho}_{t}$ in the expectations on the SDEs
for $H\left(\tilde{Z}_{t}\right)$ and $H_{t}^{\prime}\left(\tilde{Z}_{t}\right)H_{t}^{\prime}\left(\tilde{Z}_{t}\right)^{\intercal}$
(see also \eqref{eq:mean_h} and \eqref{eq:cov_h} in Appendix~\ref{sec:Detailed-proofs}),
$Q_{t}^{H}=0$ is automatically guaranteed in \eqref{eq:dyn_h} w.r.t.
the new projected density $\mathcal{S}_{H}\tilde{\rho}_{t}$. In addition,
the first two moments $\bar{H}_{t}$ and $C_{t}^{H}$ \eqref{eq:dyn_h}
will stay the same w.r.t. $\mathcal{S}_{H}\tilde{\rho}_{t}$. The
continuous filtering process by letting $s\rightarrow0$ will satisfy
the equations \eqref{eq:dyn_h} with $Q_{t}^{H}\equiv0$. Therefore,
under the same initial condition and the uniqueness of the solution,
the same solution will be reached in both \eqref{eq:KB-quad} and
\eqref{eq:dyn_h}. This leads to the main result of this section concerning
the analysis step update in the approximating filter solution.
\begin{thm}
\label{thm:consist_analy}Consider the analysis step update \eqref{eq:analysis_update}
of the statistical filtering model \eqref{eq:mf-fpf-quad}. Assume
that $a_{t},K_{t}$ in the statistical filtering SDE are designed
to satisfy \eqref{eq:coeff_sde} and the probability set $\mathcal{V}_{H}$
according to $H$ defined in \eqref{eq:symm_proj} is not empty. Under
the same statistical observations $y_{t},t\in\left[0,T\right]$ and
the same initial conditions, the following relations hold for $t\in\left[0,T\right]$
\begin{equation}
\mathcal{H}\hat{\rho}_{t}=\tilde{\mathbb{E}}\left[H\left(\tilde{Z}_{t}\right)\right],\quad\mathcal{H}\mathcal{\hat{C}}_{t}\mathcal{H}^{*}=\tilde{\mathbb{E}}\left[H^{\prime}\left(\tilde{Z}_{t}\right)H^{\prime}\left(\tilde{Z}_{t}\right)^{\intercal}\right],\label{eq:consist_analy}
\end{equation}
where $\left(\hat{\rho}_{t},\hat{\mathcal{C}}_{t}\right)$ is the
solution of \eqref{eq:KB-quad}, and $\tilde{\rho}_{t}$ is given
by the solution of \eqref{eq:KS-fpf-quad} with $H^{\prime}=H-\mathcal{H}\hat{\rho}_{t}$.
\end{thm}

Theorem~\ref{thm:consist_analy} validates the use of the statistical
filtering model density $\tilde{\rho}_{t}$ by solving \eqref{eq:mf-fpf}
to approximate the optimal filter $\hat{\rho}_{t}$ from \eqref{eq:KB-fpf}.
Though restricted only on the first two moments of the observation
function $H$, the resulting consistent statistics during analysis
step play a key role in accurate statistical forecast. Notice that
based on the statistical model in \eqref{eq:closure_model}, accurate
prediction of the important leading statistics, $\bar{u}_{t},R_{t}$,
is determined by key higher-order feedbacks in the related functional
$\mathcal{H}\rho_{t}$ (more specifically, the terms $\mathcal{H}_{m}\rho_{t}$
and $\mathcal{H}_{v}\rho_{t}$). According to Proposition~\ref{prop:optim_mom},
the optimal $\mathcal{H}\hat{\rho}_{t}$ gives the least mean square
estimate of the random variable $\mathcal{H}\rho_{t}$ given the statistical
observations. Thus, consistent approximating filter $\tilde{\mathbb{E}}H\left(\tilde{Z}_{t}\right)$
for $\mathcal{H}\hat{\rho}_{t}$ as well as its error estimate guarantees
accurate recovery of key model statistics. For example, applying the
explicit forms of the observation function \eqref{eq:operators_obs},
the quadratic observation operator $\mathcal{H}_{m}$ in the mean
equation gives
\[
\tilde{\mathbb{E}}H^{m}\left(\tilde{Z}_{t}\right)=\mathcal{H}_{m}\hat{\rho}_{t}\;\Leftrightarrow\;\sum_{p,q}\gamma_{kpq}\tilde{\mathbb{E}}\left[\tilde{Z}_{p,t}\tilde{Z}_{q,t}\right]=\sum_{p,q}\gamma_{kpq}\int z_{p}z_{q}\hat{\rho}_{t}\left(z\right)\mathrm{d}z,
\]
which implies consistent statistical feedbacks from the statistical
filtering model $\tilde{\rho}_{t}$ and the optimal filter solution
$\hat{\rho}_{t}$. This demonstrates that the new approximating filter
maintains the accuracy in the statistical mean prediction $\bar{u}_{t}$.
In addition, the covariance operator characterizes the essential uncertainty
in the optimal filter estimate $\hat{\rho}_{t}$ as a random field,
that is,
\[
\mathcal{H}_{m}\mathcal{\hat{C}}_{t}\mathcal{H}_{m}^{*}=\mathbb{E}_{\mu_{t}}\left[\left(\mathcal{H}_{m}\rho_{t}-\mathcal{H}_{m}\hat{\rho}_{t}\right)\left(\mathcal{H}_{m}\rho_{t}-\mathcal{H}_{m}\hat{\rho}_{t}\right)^{\intercal}\right].
\]
It is also linked to the approximation by $C_{t}^{H^{m}}=\tilde{\mathbb{E}}\left[\left(H_{t}^{m}-\bar{H}_{t}^{m}\right)\left(H_{t}^{m}-\bar{H}_{t}^{m}\right)^{\intercal}\right]=\mathcal{H}_{m}\mathcal{\hat{C}}_{t}\mathcal{H}_{m}^{*}$,
demonstrating a consistent error estimate in the statistical filtering
model. Similar conclusion can be reached for the accurate prediction
in the model covariance prediction for $R_{t}$ based on the cubic
observation operator $\mathcal{H}_{v}$. 
\begin{rem*}
Still, the statistical consistency in the analysis step does not guarantee
the consistency in the entire two-step updating procedure in \eqref{eq:filter_op}.
In particular, the forecast models of \eqref{eq:KB-fpf} and \eqref{eq:mf-fpf}
have the following updating equations during the forecast step
\[
\begin{aligned}\partial_{t}\hat{\rho}_{t}= & \:\mathcal{L}_{t}^{*}\hat{\rho}_{t}\\
\partial_{t}\tilde{\rho}_{t}= & \:\mathcal{L}_{t}^{*}\tilde{\rho}_{t}
\end{aligned}
\;\Rightarrow\;\begin{aligned}\partial_{t}\left(\mathcal{H}\hat{\rho}_{t}\right)= & \:\mathcal{H}\mathcal{L}_{t}^{*}\hat{\rho}_{t}=\int\left(\mathcal{L}_{t}H\right)\hat{\rho}_{t}\left(z\right)\mathrm{d}z\\
\partial_{t}\left(\tilde{\mathbb{E}}H\right)= & \:\tilde{\mathbb{E}}\mathcal{L}_{t}H=\int\left(\mathcal{L}_{t}H\right)\tilde{\rho}_{t}\left(z\right)\mathrm{d}z
\end{aligned}
\]
where the generator is defined with $y_{t}=\left(\bar{u}_{t},R_{t}\right)$
\[
\mathcal{L}_{t}\left(y_{t}\right)\varphi=\nabla_{z}\varphi\cdot\left\{ \left[H\left(z\right)-R_{t}\right]+L\left(\bar{u}_{t}\right)z\right\} +\frac{1}{2}\Sigma_{t}^{2}:\nabla_{z}\nabla_{z}\varphi.
\]
The analysis step update only gives consistent first two moments of
$H$, while higher moments may be included in $\mathcal{L}_{t}H$.
It is not guaranteed that the forecast model can give consistent forecast
in $\mathcal{H}\hat{\rho}_{t}$ and $\tilde{\mathbb{E}}H$ as well
as their covariances. More work is still needed for the complete consistency
analysis combining the approximations in both the forecast and analysis
step of the filtering method.
\end{rem*}

\section{Stability and convergence of the statistical filtering model \protect\label{sec:Stability-and-convergence}}

Here, we discuss the long-time performance of the full filtering problem
for the statistical filtering system \eqref{eq:fpf} for finding the
optimal filter PDF $\hat{\rho}_{t}$ based on statistical observations
$\mathcal{G}_{t}=\sigma\left\{ y_{s},s\leq t\right\} $. In Section~\ref{sec:Data-assimilation},
it shows that the statistical filtering model \eqref{eq:mf-fpf} constitutes
the approximate filter solution $\tilde{\rho}_{t}$ with consistent
mean and covariance, $\bar{H}_{t}$ and $C_{t}^{H}$, in the analysis
step update. We show further here that the full filter approximation
of the observation function $\bar{H}_{t}$ will approach the optimal
filter $\mathcal{H}\hat{\rho}_{t}$ at the long-time limit as $t\rightarrow\infty$.
This guarantees the stable performance of the proposed new filtering
strategy. 

\subsection{Complete statistical filter equations based on the observation operator}

We consider the optimal filter solution based on the conditional Gaussian
model \eqref{eq:fpf}. The finite-dimensional statistical states $\mathcal{H}\hat{\rho}_{t}\in\mathbb{R}^{p}$
and $\hat{\mathcal{C}}_{t}^{\mathcal{H}}=\mathcal{H}\hat{\mathcal{C}}_{t}\mathcal{H}^{*}\in\mathbb{R}^{p\times p}$
under the observation operator are solved by the Kalman-Bucy equations
\eqref{eq:KB-fpf} as
\begin{equation}
\begin{aligned}\mathrm{d}\left(\mathcal{H}\hat{\rho}_{t}\right)= & \left\langle \mathcal{L}_{t}\left(y_{t}\right)H,\hat{\rho}_{t}\right\rangle \mathrm{d}t+\hat{\mathcal{C}}_{t}^{\mathcal{H}}\Gamma^{-2}\left\{ \mathrm{d}y_{t}-\left[\mathcal{H}\hat{\rho}_{t}+h_{t}\left(y_{t}\right)\right]\mathrm{d}t\right\} ,\\
\mathrm{d}\hat{\mathcal{C}}_{t}^{\mathcal{H}}= & \left[\left\langle \mathcal{L}_{t}\left(y_{t}\right)H,\hat{\mathcal{C}}_{t}H^{\intercal}\right\rangle +\left\langle \hat{\mathcal{C}}_{t}H,\mathcal{L}_{t}\left(y_{t}\right)H^{\intercal}\right\rangle \right]\mathrm{d}t-\hat{\mathcal{C}}_{t}^{\mathcal{H}}\Gamma^{-2}\hat{\mathcal{C}}_{t}^{\mathcal{H}}\mathrm{d}t.
\end{aligned}
\label{eq:filter_optim}
\end{equation}
Above, $\left\langle F,G\right\rangle =\int F\left(z\right)G\left(z\right)\mathrm{d}z$
denotes the componentwise integration of the product of matrix-valued
functions $F,G$. For simplicity, we use constant observation noise,
$\Gamma_{t}\equiv\Gamma$. In the first equation for the mean, we
rewrite the forecast step dynamics as
\[
\mathcal{H}\mathcal{L}_{t}^{*}\left(y_{t}\right)\hat{\rho}_{t}=\int H\left(z\right)\mathcal{L}_{t}^{*}\left(y_{t}\right)\hat{\rho}_{t}\left(z\right)\mathrm{d}z=\left\langle \mathcal{L}_{t}\left(y_{t}\right)H,\hat{\rho}_{t}\right\rangle .
\]
Similarly in the covariance equation, we rewrite using the definition
of $\hat{\mathcal{C}}_{t}$ in \eqref{eq:optim_v} under the conditional
measure $\mu_{t}\left(\cdot\right)=\mathbb{P}\left(\rho\in\cdot\mid\mathcal{G}_{t}\right)$
\begin{align*}
\mathcal{H}\mathcal{L}_{t}^{*}\left(y_{t}\right)\mathcal{\hat{C}}_{t}\mathcal{H}^{*} & =\mathbb{E}_{\rho\sim\mu_{t}}\int H\left(x\right)\mathcal{L}_{t}^{*}\left(\rho-\hat{\rho}_{t}\right)\left(x\right)\mathrm{d}x\int\left(\rho-\hat{\rho}_{t}\right)\left(z\right)H\left(z\right)^{\intercal}\mathrm{d}z=\left\langle \mathcal{L}_{t}H,\hat{\mathcal{C}}_{t}H^{\intercal}\right\rangle ,\\
\mathcal{H}\mathcal{\hat{C}}_{t}\mathcal{L}_{t}\left(y_{t}\right)\mathcal{H}^{*} & =\mathbb{E}_{\rho\sim\mu_{t}}\int H\left(z\right)\left(\rho-\hat{\rho}_{t}\right)\left(z\right)\mathrm{d}z\int\left(\rho-\hat{\rho}_{t}\right)\left(x\right)\mathcal{L}_{t}H\left(x\right)^{\intercal}\mathrm{d}x=\left\langle \hat{\mathcal{C}}_{t}H,\mathcal{L}_{t}H^{\intercal}\right\rangle ,
\end{align*}
where $\hat{\mathcal{C}}_{t}H\left(x\right)=\mathbb{E}_{\rho}\left[\left(\rho-\hat{\rho}_{t}\right)\left(x\right)\int\left(\rho-\hat{\rho}_{t}\right)\left(z\right)H\left(z\right)\mathrm{d}z\right]\in\mathbb{R}^{d}$
can be viewed as a unnormalized density. 

Correspondingly, consider the approximating filter model \eqref{eq:mf-fpf}
with $\Sigma_{t}\equiv0$. Let $\tilde{\rho}_{t}$ be the PDF after
the symmetric projection \eqref{eq:symm_proj}, that is, satisfying
vanishing third-order moments $\tilde{\mathbb{E}}\left[H_{n}^{\prime}H_{p}^{\prime}H_{q}^{\prime}\right]=0$
for all $n,p,q$. The moments $\bar{H}_{t}=\tilde{\mathbb{E}}\left[H\left(\tilde{Z}_{t}\right)\right]$
and $C_{t}^{H}=\tilde{\mathbb{E}}\left[H_{t}^{\prime}\left(\tilde{Z}_{t}\right)H_{t}^{\prime}\left(\tilde{Z}_{t}\right)^{\intercal}\right]$
according to the observation function $H$ satisfy the following equations
(by combining the forecast step update with generator $\mathcal{L}_{t}$
and the analysis step dynamics \eqref{eq:dyn_h})
\begin{equation}
\begin{aligned}\mathrm{d}\bar{H}_{t}= & \left\langle \mathcal{L}_{t}\left(y_{t}\right)H,\tilde{\rho}_{t}\right\rangle \mathrm{d}t+C_{t}^{H}\Gamma^{-2}\left\{ \mathrm{d}y_{t}-\left[\bar{H}_{t}+h_{t}\left(y_{t}\right)\right]\mathrm{d}t\right\} ,\\
\mathrm{d}C_{t}^{H}= & \left[\left\langle \mathcal{L}_{t}\left(y_{t}\right)H,\tilde{\rho}_{t}H_{t}^{\prime\intercal}\right\rangle +\left\langle \tilde{\rho}_{t}H_{t}^{\prime},\mathcal{L}_{t}\left(y_{t}\right)H^{\intercal}\right\rangle \right]\mathrm{d}t-C_{t}^{H}\Gamma^{-2}C_{t}^{H}\mathrm{d}t,
\end{aligned}
\label{eq:filter_approx}
\end{equation}
where $H_{t}^{\prime}\left(z\right)=H\left(z\right)-\bar{H}_{t}$.

First, we introduce the following assumptions for the approximating
filter PDF $\tilde{\rho}_{t}$ approximating the optimal filter solution
$\hat{\rho}_{t}$ as random fields according to the same observation
process $Y=\left\{ y_{t},t\geq0\right\} $:
\begin{assumption*}
\label{assu:limit}Assume that the optimal filter \eqref{eq:KB-fpf}
and the approximating filter model \eqref{eq:mf-fpf} have a continuous
probability density functions $\hat{\rho}_{t}$ and $\tilde{\rho}_{t}$
respectively that satisfy the following conditions:
\begin{itemize}
\item Unique equilibrium solutions $\left(\mathcal{H}\hat{\rho}_{\infty},\hat{\mathcal{C}}_{\infty}^{\mathcal{H}}\right)$
and $\left(\mathcal{H}\tilde{\rho}_{\infty},C_{\infty}^{H}\right)$
exist to the systems \eqref{eq:filter_optim} and \eqref{eq:filter_approx}
respectively.
\item There exist deterministic matrices $L_{\infty}^{m},L_{\infty}^{v}\in\mathbb{R}^{p\times p}$.
The generator $\mathcal{L}_{t}$ \eqref{eq:dyn_pdf} reaches the same
statistical limit under $H$ w.r.t. both $\hat{\rho}_{t}$ and $\tilde{\rho}_{t}$,
that is,
\begin{equation}
\begin{aligned}\left\langle \mathcal{L}_{t}\left(y_{t}\right)H,\hat{\rho}_{t}\right\rangle \rightarrow L_{\infty}^{m}\mathcal{H}\hat{\rho}_{\infty}, & \;\left\langle \mathcal{L}_{t}\left(y_{t}\right)H,\hat{\mathcal{C}}_{t}H^{\intercal}\right\rangle \rightarrow L_{\infty}^{v}\hat{\mathcal{C}}_{\infty}^{\mathcal{H}},\\
\left\langle \mathcal{L}_{t}\left(y_{t}\right)H,\tilde{\rho}_{t}\right\rangle \rightarrow L_{\infty}^{m}\mathcal{H}\tilde{\rho}_{\infty}, & \;\left\langle \mathcal{L}_{t}\left(y_{t}\right)H,\tilde{\rho}_{t}H_{t}^{\prime\intercal}\right\rangle \rightarrow L_{\infty}^{v}C_{\infty}^{H},
\end{aligned}
\label{eq:limit_filter}
\end{equation}
a.s. as $t\rightarrow\infty$ given the observation data $y_{t}$.
\item Further, the real parts of eigenvalues of the limit matrices $L_{\infty}^{m}-\hat{\mathcal{C}}_{\infty}^{\mathcal{H}}\Gamma^{-2}$
and $L_{\infty}^{v}-\hat{\mathcal{C}}_{\infty}^{\mathcal{H}}\Gamma^{-2}$
are all negative.
\end{itemize}
\end{assumption*}
The first condition in assumption \eqref{eq:limit_filter} guarantees
that both filter equations will finally converge to finite solutions.
The second condition requires consistent first and second-order moments
under $H$ w.r.t. the optimal and approximate model PDF, such as for
the covariance as $t\rightarrow\infty$
\[
\begin{aligned}\int\mathcal{L}_{t}\left(y_{t}\right)H\left(z\right)\hat{\mathcal{C}}_{t}H\left(z\right)^{\intercal}\mathrm{d}z\rightarrow & L_{\infty}^{v}\int H\left(z\right)\left(\hat{\mathcal{C}}_{\infty}H^{\intercal}\right)\left(z\right)\mathrm{d}z,\\
\int\mathcal{L}_{t}\left(y_{t}\right)H_{t}^{\prime}\left(z\right)H_{t}^{\prime}\left(z\right){}^{\intercal}\tilde{\rho}_{t}\left(z\right)\mathrm{d}z\rightarrow & L_{\infty}^{v}\int H_{\infty}^{\prime}\left(z\right)H_{\infty}^{\prime}\left(z\right){}^{\intercal}\tilde{\rho}_{\infty}\left(z\right)\mathrm{d}z.
\end{aligned}
\]
In addition, we may further introduce the convergence rate, that is,
there exist constants $\bar{\lambda}>0$ and $K>0$ such that a.s.
\begin{equation}
\begin{aligned}\left|\left\langle \mathcal{L}_{t}\left(y_{t}\right)H,\hat{\rho}_{t}\right\rangle -L_{\infty}^{m}\mathcal{H}\hat{\rho}_{\infty}\right|\leq Ke^{-\bar{\lambda}t}, & \quad\left\Vert \left\langle \mathcal{L}_{t}\left(y_{t}\right)H,\hat{\mathcal{C}}_{t}H^{\intercal}\right\rangle -L_{\infty}^{v}\hat{\mathcal{C}}_{\infty}^{\mathcal{H}}\right\Vert \leq Ke^{-\bar{\lambda}t},\\
\left|\left\langle \mathcal{L}_{t}\left(y_{t}\right)H,\tilde{\rho}_{t}\right\rangle -L_{\infty}^{m}\mathcal{H}\tilde{\rho}_{\infty}\right|\leq Ke^{-\bar{\lambda}t}, & \quad\left\Vert \left\langle \mathcal{L}_{t}\left(y_{t}\right)H,\tilde{\rho}_{t}H_{t}^{\prime\intercal}\right\rangle -L_{\infty}^{v}C_{\infty}^{H}\right\Vert \leq Ke^{-\bar{\lambda}t},
\end{aligned}
\label{eq:limit_exp}
\end{equation}
where $\left\Vert \cdot\right\Vert $ is the matrix norm. And the
third condition asks that the filter solutions will be stabilized
at the long time limit. Next, we ask the limit behaviour in the mean
and covariance $\left(\bar{H}_{t},C_{t}^{H}\right)$ from the approximation
model \eqref{eq:filter_approx} in comparison with the optimal filter
solution $\left(\mathcal{H}\hat{\rho}_{t},\hat{\mathcal{C}}_{t}^{\mathcal{H}}\right)$
from \eqref{eq:filter_optim}.

\subsection{Asymptotic stability of the equilibrium covariance matrix}

With the assumption \eqref{eq:limit_filter}, we have consistent equilibrium
covariance at $t\rightarrow\infty$ in the two model solutions
\[
\begin{aligned}L_{\infty}^{v}\hat{\mathcal{C}}_{\infty}^{\mathcal{H}}+\hat{\mathcal{C}}_{\infty}^{\mathcal{H}}L_{\infty}^{v\intercal}-\hat{\mathcal{C}}_{\infty}^{\mathcal{H}}\Gamma^{-2}\hat{\mathcal{C}}_{\infty}^{\mathcal{H}} & =0,\\
L_{\infty}^{v}C_{\infty}^{H}+C_{\infty}^{H}L_{\infty}^{v\intercal}-C_{\infty}^{H}\Gamma^{-2}C_{\infty}^{H} & =0.
\end{aligned}
\]
Uniqueness of the solution directly implies that the final equilibrium
covariances satisfy $\hat{\mathcal{C}}_{\infty}^{\mathcal{H}}=C_{\infty}^{H}$
with no randomness. Further, we have that the covariance $C_{t}^{H}$
in the approximating filter will approach the optimal equilibrium
covariance $\hat{\mathcal{C}}_{\infty}^{\mathcal{H}}$ as described
in the following result.
\begin{lem}
\label{lem:cov_limit}Suppose that Assumption~\ref{assu:limit} is
satisfied, $C_{t}^{H}$ is the covariance solution to the statistical
filtering model \eqref{eq:filter_approx} and $\hat{\mathcal{C}}_{\infty}^{\mathcal{H}}$
is the unique equilibrium solutions to the optimal model \eqref{eq:filter_optim}.
Then there is
\begin{equation}
\left\Vert C_{t}^{H}-\hat{\mathcal{C}}_{\infty}^{\mathcal{H}}\right\Vert \rightarrow0,\label{eq:conv_cov}
\end{equation}
a.s. as $t\rightarrow\infty$. Further, the convergence rate will
be exponential if \eqref{eq:limit_exp} is also satisfied.
\end{lem}

\begin{proof}
Combining the covariance equation in \eqref{eq:filter_approx} and
the equilibrium equation of \eqref{eq:filter_optim}, we have
\begin{align*}
\mathrm{d}\left(C_{t}^{H}-\hat{\mathcal{C}}_{\infty}^{\mathcal{H}}\right) & =\left[L_{\infty}^{v}-\frac{1}{2}\left(C_{t}^{H}+\hat{\mathcal{C}}_{\infty}^{\mathcal{H}}\right)\Gamma^{-2}\right]\left(C_{t}^{H}-\hat{\mathcal{C}}_{\infty}^{\mathcal{H}}\right)\mathrm{d}t\\
 & +\left(C_{t}^{H}-\hat{\mathcal{C}}_{\infty}^{\mathcal{H}}\right)\left[L_{\infty}^{v\intercal}-\frac{1}{2}\Gamma^{-2}\left(C_{t}^{H}+\hat{\mathcal{C}}_{\infty}^{\mathcal{H}}\right)\right]\mathrm{d}t\\
 & +\left[\left\langle \mathcal{L}_{t}H,\tilde{\rho}_{t}H_{t}^{\prime}{}^{\intercal}\right\rangle -L_{\infty}^{v}C_{t}^{H}\right]\mathrm{d}t+\left[\left\langle \tilde{\rho}_{t}H_{t}^{\prime},\mathcal{L}_{t}H^{\intercal}\right\rangle -C_{t}^{H}L_{\infty}^{v\intercal}\right]\mathrm{d}t.
\end{align*}
For the last row of the above equation, using the uniqueness of the
solution we have $C_{t}^{H}\rightarrow C_{\infty}^{H}$ as $t\rightarrow\infty$.
Then denote
\[
F_{t}=\left\langle \mathcal{L}_{t}H,\tilde{\rho}_{t}H_{t}^{\prime}{}^{\intercal}\right\rangle -L_{\infty}^{v}C_{t}^{H}.
\]
We get $\left|F_{t}\right|\rightarrow0$ as $t\rightarrow\infty$
by using
\[
\left\Vert \left\langle \mathcal{L}_{t}H,\tilde{\rho}_{t}H_{t}^{\prime}{}^{\intercal}\right\rangle -L_{\infty}^{v}C_{t}^{H}\right\Vert \leq\left\Vert \left\langle \mathcal{L}_{t}H,\tilde{\rho}_{t}H_{t}^{\prime}{}^{\intercal}\right\rangle -L_{\infty}^{v}C_{\infty}^{H}\right\Vert +\left\Vert L_{\infty}^{v}\left(C_{\infty}^{H}-C_{t}^{H}\right)\right\Vert \rightarrow0.
\]
Also by taking $\ \lambda_{v}=\min\left\{ \mathrm{Re}\lambda:\lambda\:\mathrm{is\:the\:eighenvalue\:of}\:-L_{\infty}^{v}+\hat{\mathcal{C}}_{\infty}^{\mathcal{H}}\Gamma^{-2}\right\} $,
we have $\lambda_{v}>0$ from Assumption~\ref{assu:limit}. This
implies $\left\Vert e^{\int_{s}^{t}\left[L_{\infty}^{v}-\frac{1}{2}\left(C_{\tau}^{H}+\hat{\mathcal{C}}_{\infty}^{\mathcal{H}}\right)\Gamma^{-2}\right]\mathrm{d}\tau}\right\Vert \leq Ke^{-\lambda_{v}\left(t-s\right)}$
. Together with $F_{t}$ vanishing as $t\rightarrow\infty$, we have
for any $t>T$
\[
\left\Vert C_{t}^{H}-\hat{\mathcal{C}}_{\infty}^{\mathcal{H}}\right\Vert \leq K\left\Vert C_{T}^{H}-\hat{\mathcal{C}}_{\infty}^{\mathcal{H}}\right\Vert e^{-2\lambda_{v}\left(t-T\right)}+2K\int_{T}^{t}e^{-2\lambda_{v}\left(t-s\right)}\left|F_{s}\right|\mathrm{d}s\leq K_{1}e^{-2\lambda_{v}t}+K_{2}\sup_{s\geq T}\left|F_{s}\right|.
\]
Therefore, by first letting $t\rightarrow\infty$ then letting $T\rightarrow\infty$,
we reach a.s. $\lim_{t\rightarrow\infty}\left\Vert C_{t}^{H}-\hat{\mathcal{C}}_{\infty}^{\mathcal{H}}\right\Vert =0$.

Further, if we assume exponential convergence rate $\bar{\lambda}$
in the covariances under the generator $\mathcal{L}_{t}$ as in \eqref{eq:limit_exp},
we can have exponential convergence rate in both $\hat{\mathcal{C}}_{t}$
and $C_{t}^{H}$ as $t\rightarrow\infty$ a.s.
\begin{equation}
\left\Vert \hat{\mathcal{C}}_{t}^{\mathcal{H}}-\hat{\mathcal{C}}_{\infty}^{\mathcal{H}}\right\Vert \leq K_{H}e^{-\min\left\{ \lambda_{v},\bar{\lambda}\right\} t},\;\left\Vert C_{t}^{H}-\hat{\mathcal{C}}_{\infty}^{\mathcal{H}}\right\Vert \leq K_{H}e^{-\min\left\{ \lambda_{v},\bar{\lambda}\right\} t}.\label{eq:conv_exp}
\end{equation}
\end{proof}
In addition, from the definition of the optimal filter solution \eqref{eq:optim_mom},
the covariance is defined as
\begin{align*}
\mathrm{tr}\mathbb{E}\hat{\mathcal{C}}_{t}^{\mathcal{H}} & =\mathbb{E}_{\mu_{t}}\mathrm{tr}\left[\left(\mathcal{H}\rho_{t}-\mathcal{H}\hat{\rho}_{t}\right)\left(\mathcal{H}\rho_{t}-\mathcal{H}\hat{\rho}_{t}\right)^{\intercal}\right]\\
 & =\mathbb{E}_{\mu_{t}}\mathrm{tr}\left[\left(\mathcal{H}\rho_{t}-\mathcal{H}\hat{\rho}_{t}\right)^{\intercal}\left(\mathcal{H}\rho_{t}-\mathcal{H}\hat{\rho}_{t}\right)\right]\\
 & =\mathbb{E}_{\mu_{t}}\left[\left|\mathcal{H}\rho_{t}-\mathcal{H}\hat{\rho}_{t}\right|^{2}\right].
\end{align*}
Notice that $\mathcal{H}\hat{\rho}_{t}$ and $\hat{\mathcal{C}}_{t}^{\mathcal{H}}$
are still random field in $\mathcal{G}_{t}$. Under Assumption~\ref{assu:limit}
$\hat{\mathcal{C}}_{t}^{\mathcal{H}}\rightarrow\hat{\mathcal{C}}_{\infty}^{\mathcal{H}}$
a.s. with the observations in $\mathcal{G}_{t}$, we have as $t\rightarrow\infty$
\begin{equation}
\mathbb{E}\left[\left|\mathcal{H}\rho_{t}-\mathcal{H}\hat{\rho}_{t}\right|^{2}\right]=\mathrm{tr}\mathbb{E}\left[\hat{\mathcal{C}}_{t}^{\mathcal{H}}\right]\rightarrow\mathrm{tr}\hat{\mathcal{C}}_{\infty}^{\mathcal{H}}.\label{eq:uncertainty_v}
\end{equation}
This confirms that the total uncertainty at equilibrium in the optimal
filter solution $\mathcal{H}\hat{\rho}_{t}$ is estimated by the total
variance $\mathrm{tr}\hat{\mathcal{C}}_{\infty}^{\mathcal{H}}$.

\subsection{Convergence of the statistical state under the observation operator}

Next, we consider the convergence of the statistical observation function
$\bar{H}_{t}$ from the approximating filter to the optimal filter
solution $\mathcal{H}\hat{\rho}_{t}$. We have the long-term stability
in the statistical solution in the following theorem.
\begin{thm}
\label{thm:mean_conv}Suppose that Assumption~\ref{assu:limit} holds
and the covariances goes to the same deterministic limit
\[
\hat{\mathcal{C}}_{t}^{\mathcal{H}}\rightarrow\hat{\mathcal{C}}_{\infty}^{\mathcal{H}},\quad C_{t}^{H}\rightarrow\hat{\mathcal{C}}_{\infty}^{\mathcal{H}},
\]
a.s. as $t\rightarrow\infty$. Then there is
\begin{equation}
\mathbb{E}\left[\left|\mathcal{H}\hat{\rho}_{t}-\bar{H}_{t}\right|^{2}\right]\rightarrow0,\label{eq:conv_mean}
\end{equation}
as $t\rightarrow\infty$. Furthermore, assume exponential convergence
rate in $\hat{\mathcal{C}}_{t}^{\mathcal{H}},C_{t}^{H}$ as in \eqref{eq:conv_exp},
and let $\lambda=\min\left\{ \lambda_{m},\lambda_{v}\right\} $ where
$\lambda_{m}$ and $\lambda_{v}$ are the minimums of the real parts
of the eigenvalues of $-L_{\infty}^{m}+\hat{\mathcal{C}}_{\infty}^{\mathcal{H}}\Gamma^{-2}$
and $-L_{\infty}^{v}+\hat{\mathcal{C}}_{\infty}^{\mathcal{H}}\Gamma^{-2}$
respectively. There is also exponential convergence as
\begin{equation}
\mathbb{E}\left[\left|\mathcal{H}\hat{\rho}_{t}-\bar{H}_{t}\right|^{2}\right]\leq K_{H}e^{-\lambda_{1}t},\label{eq:conv_mean_exp}
\end{equation}
with some $\lambda_{1}\leq\min\left\{ \lambda,\bar{\lambda}\right\} $
and $K_{H}$ a constant only dependent on the observation function
$H$.
\end{thm}

\begin{proof}
By taking the difference of the mean equations in \eqref{eq:filter_optim}
and \eqref{eq:filter_approx}, we have
\begin{align*}
\mathrm{d}\left(\mathcal{H}\hat{\rho}_{t}-\bar{H}_{t}\right) & =\left(L_{\infty}^{m}-\hat{\mathcal{C}}_{\infty}^{\mathcal{H}}\Gamma^{-2}\right)\left(\mathcal{H}\hat{\rho}_{t}-\bar{H}_{t}\right)\mathrm{d}t+\left(\hat{\mathcal{C}}_{t}^{\mathcal{H}}-C_{t}^{H}\right)\Gamma^{-2}\left[\mathrm{d}y_{t}-h_{t}\left(y_{t}\right)\mathrm{d}t\right]\\
 & +\left[\left(\hat{\mathcal{C}}_{\infty}^{\mathcal{H}}-\hat{\mathcal{C}}_{t}^{\mathcal{H}}\right)\Gamma^{-2}\mathcal{H}\hat{\rho}_{t}-\left(\hat{\mathcal{C}}_{\infty}^{\mathcal{H}}-C_{t}^{H}\right)\Gamma^{-2}\bar{H}_{t}\right]\mathrm{d}t\\
 & +\left[\left\langle \mathcal{L}_{t}\left(y_{t}\right)H,\hat{\rho}_{t}\right\rangle -L_{\infty}^{m}\mathcal{H}\hat{\rho}_{t}\right]\mathrm{d}t+\left[\left\langle \mathcal{L}_{t}\left(y_{t}\right)H,\tilde{\rho}_{t}\right\rangle -L_{\infty}^{m}\bar{H}_{t}\right]\mathrm{d}t.
\end{align*}
By applying It\^{o}'s formula to the above equation, there is
\begin{align*}
\mathrm{d}\left[e^{-t\left(L_{\infty}^{m}-\hat{\mathcal{C}}_{\infty}^{\mathcal{H}}\Gamma^{-2}\right)}\left(\mathcal{H}\hat{\rho}_{t}-\bar{H}_{t}\right)\right] & =e^{-t\left(L_{\infty}^{m}-\hat{\mathcal{C}}_{\infty}^{\mathcal{H}}\Gamma^{-2}\right)}\left\{ \left(\hat{\mathcal{C}}_{t}^{\mathcal{H}}-C_{t}^{H}\right)\Gamma^{-2}\left[\mathrm{d}y_{t}-h_{t}\left(y_{t}\right)\mathrm{d}t\right]+G_{t}\left(y_{t}\right)\mathrm{d}t\right\} \\
 & +e^{-t\left(L_{\infty}^{m}-\hat{\mathcal{C}}_{\infty}^{\mathcal{H}}\Gamma^{-2}\right)}\left[\left(\hat{\mathcal{C}}_{\infty}^{\mathcal{H}}-\hat{\mathcal{C}}_{t}^{\mathcal{H}}\right)\Gamma^{-2}\mathcal{H}\hat{\rho}_{t}-\left(\hat{\mathcal{C}}_{\infty}^{\mathcal{H}}-C_{t}^{H}\right)\Gamma^{-2}\bar{H}_{t}\right]\mathrm{d}t.
\end{align*}
Above, we denote the residual term as
\[
G_{t}\left(y_{t}\right)=\left[\left\langle \mathcal{L}_{t}\left(y_{t}\right)H,\hat{\rho}_{t}\right\rangle -L_{\infty}^{m}\mathcal{H}\hat{\rho}_{t}\right]+\left[\left\langle \mathcal{L}_{t}\left(y_{t}\right)H,\tilde{\rho}_{t}\right\rangle -L_{\infty}^{m}\bar{H}_{t}\right].
\]
By similar computation in Lemma~\ref{lem:cov_limit} according to
assumption \eqref{eq:limit_filter}, there is
\[
\left|G_{t}\right|\leq\begin{aligned}\left|\left\langle \mathcal{L}_{t}\left(y_{t}\right)H,\hat{\rho}_{t}\right\rangle -L_{\infty}^{m}\mathcal{H}\hat{\rho}_{\infty}\right| & +\left|L_{\infty}^{m}\left(\mathcal{H}\hat{\rho}_{\infty}-\mathcal{H}\hat{\rho}_{t}\right)\right|\\
+\left|\left\langle \mathcal{L}_{t}\left(y_{t}\right)H,\tilde{\rho}_{t}\right\rangle -L_{\infty}^{m}\bar{H}_{\infty}\right| & +\left|L_{\infty}^{m}\left(\bar{H}_{\infty}-\bar{H}_{t}\right)\right|
\end{aligned}
\leq Ke^{-\bar{\lambda}t}\rightarrow0
\]
a.s. as $t\rightarrow\infty$ using the uniqueness of the solutions
$\mathcal{H}\hat{\rho}_{t}\rightarrow\mathcal{H}\hat{\rho}_{\infty}$
and $\bar{H}_{t}\rightarrow\bar{H}_{\infty}$ (with exponential decay
in the stronger convergence case \eqref{eq:conv_exp}). Equivalently,
the above SDE can be written as (ignoring initial condition by assuming
$\mathcal{H}\hat{\rho}_{0}=\bar{H}_{0}$)
\begin{align*}
\mathcal{H}\hat{\rho}_{t}-\bar{H}_{t} & =\int_{0}^{t}e^{\left(t-s\right)\left(L_{\infty}^{m}-\hat{\mathcal{C}}_{\infty}^{\mathcal{H}}\Gamma^{-2}\right)}\left\{ \left(\hat{\mathcal{C}}_{s}^{\mathcal{H}}-C_{s}^{H}\right)\Gamma^{-2}\left[\mathrm{d}y_{s}-h_{s}\left(y_{s}\right)\mathrm{d}s\right]+G_{s}\left(y_{s}\right)\mathrm{d}s\right\} \\
 & +\int_{0}^{t}e^{\left(t-s\right)\left(L_{\infty}^{m}-\hat{\mathcal{C}}_{\infty}^{\mathcal{H}}\Gamma^{-2}\right)}\left[\left(\hat{\mathcal{C}}_{\infty}^{\mathcal{H}}-\hat{\mathcal{C}}_{s}^{\mathcal{H}}\right)\Gamma^{-2}\mathcal{H}\hat{\rho}_{s}-\left(\hat{\mathcal{C}}_{\infty}^{\mathcal{H}}-C_{s}^{H}\right)\Gamma^{-2}\bar{H}_{s}\right]\mathrm{d}s.
\end{align*}
Therefore, by taking the expectation for each term on the right hand
side of the above identity
\begin{align*}
\mathbb{E}\left[\left|\mathcal{H}\hat{\rho}_{t}-\bar{H}_{t}\right|^{2}\right] & \leq5\mathbb{E}\left|\int_{0}^{t}e^{\left(t-s\right)\left(L_{\infty}^{m}-\hat{\mathcal{C}}_{\infty}^{\mathcal{H}}\Gamma^{-2}\right)}\left(\hat{\mathcal{C}}_{s}^{\mathcal{H}}-C_{s}^{H}\right)\Gamma^{-2}\mathrm{d}y_{s}\right|^{2}\\
 & +5\mathbb{E}\left|\int_{0}^{t}e^{\left(t-s\right)\left(L_{\infty}^{m}-\hat{\mathcal{C}}_{\infty}^{\mathcal{H}}\Gamma^{-2}\right)}G_{s}\left(y_{s}\right)\mathrm{d}s\right|^{2}\\
 & +5\mathbb{E}\left|\int_{0}^{t}e^{\left(t-s\right)\left(L_{\infty}^{m}-\hat{\mathcal{C}}_{\infty}^{\mathcal{H}}\Gamma^{-2}\right)}\left(\hat{\mathcal{C}}_{s}^{\mathcal{H}}-C_{s}^{H}\right)\Gamma^{-2}h_{s}\left(y_{s}\right)\mathrm{d}s\right|^{2}\\
 & +5\mathbb{E}\left|\int_{0}^{t}e^{\left(t-s\right)\left(L_{\infty}^{m}-\hat{\mathcal{C}}_{\infty}^{\mathcal{H}}\Gamma^{-2}\right)}\left(\hat{\mathcal{C}}_{\infty}^{\mathcal{H}}-\hat{\mathcal{C}}_{s}^{\mathcal{H}}\right)\Gamma^{-2}\mathcal{H}\hat{\rho}_{s}\mathrm{d}s\right|^{2}\\
 & +5\mathbb{E}\left|\int_{0}^{t}e^{\left(t-s\right)\left(L_{\infty}^{m}-\hat{\mathcal{C}}_{\infty}^{\mathcal{H}}\Gamma^{-2}\right)}\left(\hat{\mathcal{C}}_{\infty}^{\mathcal{H}}-C_{s}^{H}\right)\Gamma^{-2}\bar{H}_{s}\mathrm{d}s\right|^{2}.
\end{align*}
The second line above follows the same argument as in Lemma~\ref{lem:cov_limit}
and using Cauchy-Schwarz inequality
\begin{align*}
\mathbb{E}\left|\int_{0}^{t}e^{\left(t-s\right)\left(L_{\infty}^{m}-\hat{\mathcal{C}}_{\infty}^{\mathcal{H}}\Gamma^{-2}\right)}G_{s}\left(y_{s}\right)\mathrm{d}s\right|^{2} & \leq\int_{0}^{t}\left\Vert e^{\frac{1}{2}\left(t-s\right)\left(L_{\infty}^{m}-\hat{\mathcal{C}}_{\infty}^{\mathcal{H}}\Gamma^{-2}\right)}\right\Vert ^{2}\mathrm{d}s\mathbb{E}\int_{0}^{t}\left|e^{\frac{1}{2}\left(t-s\right)\left(L_{\infty}^{m}-\hat{\mathcal{C}}_{\infty}^{\mathcal{H}}\Gamma^{-2}\right)}G_{s}\left(y_{s}\right)\right|^{2}\mathrm{d}s.\\
 & \leq\lambda_{m}^{-1}\mathbb{E}\int_{0}^{t}e^{-\lambda_{m}\left(t-s\right)}\left|G_{s}\left(y_{s}\right)\right|^{2}\mathrm{d}s\\
 & \leq Ke^{-\lambda_{1}t}\rightarrow0.
\end{align*}
Similar results can be achieved for line three to five following the
convergence (or exponential convergence) of the integrants. Finally,
for the first line using the observation equation in \eqref{eq:fpf},
that is, $\mathrm{d}y_{t}=\left[\mathcal{H}\rho_{t}+h_{t}\left(y_{t}\right)\right]\mathrm{d}t+\Gamma\mathrm{d}B_{t}$,
there is
\begin{align*}
 & \mathbb{E}\left|\int_{0}^{t}e^{\left(t-s\right)\left(L_{\infty}^{m}-\hat{\mathcal{C}}_{\infty}^{\mathcal{H}}\Gamma^{-2}\right)}\left(\hat{\mathcal{C}}_{s}^{\mathcal{H}}-C_{s}^{H}\right)\Gamma^{-2}\mathrm{d}y_{s}\right|^{2}\\
\leq & \mathbb{E}\left|\int_{0}^{t}e^{\left(t-s\right)\left(L_{\infty}^{m}-\hat{\mathcal{C}}_{\infty}^{\mathcal{H}}\Gamma^{-2}\right)}\left(\hat{\mathcal{C}}_{s}^{\mathcal{H}}-C_{s}^{H}\right)\Gamma^{-2}\left[\mathcal{H}\rho_{t}+h_{t}\left(y_{t}\right)\right]\mathrm{d}s\right|^{2}\\
+ & \mathbb{E}\int_{0}^{t}e^{2\left(t-s\right)\left(L_{\infty}^{m}-\hat{\mathcal{C}}_{\infty}^{\mathcal{H}}\Gamma^{-2}\right)}\left\Vert \hat{\mathcal{C}}_{s}^{\mathcal{H}}-C_{s}^{H}\right\Vert ^{2}\mathrm{d}s\\
\leq & K_{1}\int_{0}^{t}e^{-\lambda_{m}\left(t-s\right)}\left(\mathbb{E}\left\Vert \hat{\mathcal{C}}_{s}^{\mathcal{H}}-\hat{\mathcal{C}}_{\infty}^{\mathcal{H}}\right\Vert ^{2}+\mathbb{E}\left\Vert C_{s}^{H}-\hat{\mathcal{C}}_{\infty}^{\mathcal{H}}\right\Vert ^{2}\right)\mathrm{d}s\\
\leq & Ke^{-\lambda_{1}t}\rightarrow0.
\end{align*}
Above in the second to the last inequality, we use the uniform boundedness
of $\mathbb{E}\left|\mathcal{H}\rho_{t}\right|$ and $\mathbb{E}\left|h_{t}\left(y_{t}\right)\right|$
as $t\rightarrow\infty$, and the last line uses the convergence (or
exponential convergence) of the covariance \eqref{eq:conv_cov} or
\eqref{eq:conv_exp}. Combining all the above bounds, we finally get
\eqref{eq:conv_mean} and \eqref{eq:conv_mean_exp}.
\end{proof}
Together with the equilibrium estimate in \eqref{eq:uncertainty_v}
combining the result in Theorem~\ref{thm:mean_conv}, we can also
get the same error estimate compared with the target field as $t\rightarrow\infty$
\[
\mathbb{E}\left[\left|\mathcal{H}\rho_{t}-\bar{H}_{t}\right|^{2}\right]\leq\mathbb{E}\left[\left|\mathcal{H}\rho_{t}-\mathcal{H}\hat{\rho}_{t}\right|^{2}\right]+\mathbb{E}\left[\left|\mathcal{H}\hat{\rho}_{t}-\bar{H}_{t}\right|^{2}\right]\rightarrow\mathrm{tr}\hat{\mathcal{C}}_{\infty}^{\mathcal{H}}.
\]
Thus, we show that the statistical filtering  model solution $\bar{H}_{t}$
converges to the optimal filter solution with the same mean square
error, and exponential convergence is reached if the forecast model
has exponential convergence to the equilibrium.

As a final comment, we can further relax Assumption~\ref{assu:limit}
as there exist deterministic uniformly continuous functions, $L_{\infty}^{m}\left(y\right)$
and $L_{\infty}^{v}\left(y\right)$, so that for any $y_{t}\rightarrow y_{\infty},t\rightarrow\infty$
\begin{equation}
\begin{aligned}\left\langle \mathcal{L}_{t}\left(y_{t}\right)H,\hat{\rho}_{t}\right\rangle \rightarrow L_{\infty}^{m}\left(y_{\infty}\right)\mathcal{H}\hat{\rho}_{\infty}, & \;\left\langle \mathcal{L}_{t}H,\hat{\mathcal{C}}_{t}H^{\intercal}\right\rangle \rightarrow L_{\infty}^{v}\left(y_{\infty}\right)\hat{\mathcal{C}}_{\infty}^{\mathcal{H}},\\
\left\langle \mathcal{L}_{t}\left(y_{t}\right)H,\tilde{\rho}_{t}\right\rangle \rightarrow L_{\infty}^{m}\left(y_{\infty}\right)\mathcal{H}\tilde{\rho}_{\infty}, & \;\left\langle \mathcal{L}_{t}H,\tilde{\rho}_{t}H_{t}^{\prime}{}^{\intercal}\right\rangle \rightarrow L_{\infty}^{v}\left(y_{\infty}\right)C_{\infty}^{H},
\end{aligned}
\label{eq:limit_filter-1}
\end{equation}
a.s. as $t\rightarrow\infty$. And the limiting matrices are uniformly
bounded by negative-definite matrices $A_{m},A_{v}$
\begin{equation}
L_{\infty}^{m}\left(y\right)-\hat{\mathcal{C}}_{\infty}^{\mathcal{H}}\Gamma^{-2}\preceq A_{m}\prec0,\quad L_{\infty}^{v}\left(y\right)-\hat{\mathcal{C}}_{\infty}^{\mathcal{H}}\Gamma^{-2}\preceq A_{v}\prec0.\label{eq:boundedness}
\end{equation}
In addition, non-zero noise in \eqref{eq:mf-fpf} can be included
satisfying $\Sigma_{t}\rightarrow0$ as $t\rightarrow\infty$. Then,
the same result applies for the convergence of the observation mean
and covariance at long-time limit as in Lemma~\ref{lem:cov_limit}
and Theorem~\ref{thm:mean_conv}.

\section{Ensemble approximation of the statistical filtering model\protect\label{sec:Ensemble-approximation}}

Finally, we discuss the construction of practical ensemble methods
for solving the statistical filtering model \eqref{eq:mf-fpf} with
explicit model parameters that is easy to implement. We demonstrate
the accuracy of the numerical approximation at the limit of large
number of samples and small integration time step.

\subsection{Numerical methods for the approximating filter with discrete observations\protect\label{subsec:Numerical-methods}}

Assume that the observation data comes at discrete times $t_{n}=\delta n$
with a constant observation frequency $\delta$. We can introduce
the linear interpolation of $y_{n}=\left(\bar{u}_{t_{n}},R_{t_{n}}\right)\in\mathbb{R}^{p}$
as
\begin{equation}
\frac{\mathrm{d}y_{t}^{\delta}}{\mathrm{d}t}=\frac{y_{n+1}-y_{n}}{t_{n+1}-t_{n}}=\frac{\Delta y_{n+1}}{\delta},\label{eq:obs_interp}
\end{equation}
during time interval $t\in\left[t_{n},t_{n+1}\right]$. We propose
ensemble algorithms to approximate the filtering distribution $\tilde{\rho}_{t_{n}}\sim\tilde{Z}_{t_{n}}$
conditional on the statistical observations $Y_{t}^{\delta}=\left\{ y_{s}^{\delta},s\leq t\right\} $
based on the statistical filter equation \eqref{eq:mf-fpf}. First,
$N$ independent particles, $\tilde{\mathbf{Z}}_{t}=\left\{ \tilde{Z}_{t}^{\left(i\right)}\right\} _{i=1}^{N}$,
are drawn to sample the initial distribution of the stochastic state.
Then, the particles are evolved according to the following SDE with
drift terms $a_{t}^{m},a_{t}^{v}$ and control gains $K_{t}^{m},K_{t}^{v}$
\begin{equation}
\begin{aligned}\mathrm{d}\tilde{Z}_{t}^{\left(i\right)}= & \;L\left(\bar{u}_{t}^{N}\right)\tilde{Z}_{t}^{\left(i\right)}\mathrm{d}t+\varGamma\left(\tilde{Z}_{t}^{\left(i\right)}\tilde{Z}_{t}^{\left(i\right)\intercal}-R_{t}^{N}\right)\mathrm{d}t+\Sigma_{t}\mathrm{d}\tilde{W}_{t}^{\left(i\right)}\\
 & +a_{t}^{m}\left(\tilde{Z}_{t}^{\left(i\right)}\right)\mathrm{d}t+K_{t}^{m}\left(\tilde{Z}_{t}^{\left(i\right)}\right)\left\{ \mathrm{d}\bar{u}_{t}^{\delta}-\left[H^{m}\left(\tilde{Z}_{t}^{\left(i\right)}\right)+h_{m,t}\left(\bar{u}_{t}^{\delta}\right)\right]\mathrm{d}t-\Gamma_{m,t}\mathrm{d}\tilde{B}_{m,t}^{\left(i\right)}\right\} \\
 & +a_{t}^{v}\left(\tilde{Z}_{t}^{\left(i\right)}\right)\mathrm{d}t+K_{t}^{v}\left(\tilde{Z}_{t}^{\left(i\right)}\right)\left\{ \mathrm{d}R_{t}^{\delta}-\left[H^{v}\left(\tilde{Z}_{t}^{\left(i\right)}\right)+h_{v,t}\left(\bar{u}_{t}^{\delta},R_{t}^{\delta}\right)\right]\mathrm{d}t-\Gamma_{v,t}\mathrm{d}\tilde{B}_{v,t}^{\left(i\right)}\right\} ,
\end{aligned}
\label{eq:model_full}
\end{equation}
where the expressions for $h_{m},h_{v}$ are defined in \eqref{eq:dyn_obs},
$H^{m},H^{v}$ are defined in \eqref{eq:operators_obs}, and $\tilde{B}_{m,t}^{\left(i\right)},\tilde{B}_{v,t}^{\left(i\right)}$
are independent white noises. Above, in the first line of \eqref{eq:model_full}
for the forecast step of the filter, the first two moments $\left(\bar{u}_{t}^{N},R_{t}^{N}\right)$
can be explicitly solved by the statistical equations according to
the stochastic-statistical model \eqref{eq:closure_model}
\begin{equation}
\begin{aligned}\frac{\mathrm{d}\bar{u}_{t}^{N}}{\mathrm{d}t}= & \:M\left(\bar{u}_{t}^{N}\right)+F_{t}+\mathbb{E}^{N}\left[H^{m}\left(\tilde{\mathbf{Z}}_{t}\right)\right],\\
\frac{\mathrm{d}R_{t}^{N}}{\mathrm{d}t}= & \:L\left(\bar{u}_{t}^{N}\right)R_{t}^{N}+R_{t}^{N}L\left(\bar{u}_{t}^{N}\right)^{\intercal}+\Sigma_{t}\Sigma_{t}^{\intercal}+\mathbb{E}^{N}\left[H^{v}\left(\tilde{\mathbf{Z}}_{t}\right)\right]+\epsilon^{-1}\left(\mathbb{E}^{N}\left[\tilde{\mathbf{Z}}_{t}\tilde{\mathbf{Z}}_{t}^{\intercal}\right]-R_{t}^{N}\right).
\end{aligned}
\label{eq:obs_full}
\end{equation}
The expectation is computed through the empirical average of the ensemble,
$\mathbb{E}^{N}f\left(\tilde{\mathbf{Z}}\right)=\frac{1}{N}\sum_{i=1}^{N}f\left(\tilde{Z}^{\left(i\right)}\right)$.
In this way, the particle simulation of \eqref{eq:model_full} can
be carried out easily for each individual sample $\tilde{Z}_{t}^{\left(i\right)}$,
and the dependence on the distribution of the whole interacting particles
is only introduced through the empirical average in the statistical
equations \eqref{eq:obs_full}. We summarize the ensemble filtering
strategy in Algorithm~\ref{alg:filter_full}.

\begin{algorithm}
\begin{algorithmic}[1] 
\Ensure{Get the interpolated sequence of observation \eqref{eq:obs_interp} of the mean and covariance $y^{\delta}_t=\left\{ \bar{u}_{t}^{\delta},R_{t}^{\delta}\right\}, t\in\left[0,T\right]$; and  introduce discrete time integration step $\tau$, and the initial distribution $\rho_0$ of the model state.}
\Require{At the initial time $t=0$, draw an ensemble of samples $\left\{\tilde{Z}_{0}^{\left(i\right)}\right\}_{i=1}^{N}$ from the initial distribution $\rho_0$, and compute initial mean and covariance $\left\{\bar{u}^{N}_{0}, R^{N}_{0}\right\}$ consistent with the statistics w.r.t. $\rho_0$. }

\For{$n = 0$ while $n < \left\lfloor T/\tau \right\rfloor $, during the time updating interval $t\in\left[t_{n},t_{n+1}\right]$  with $t_{n}=n\tau$}
	\State{Compute the gain functions $K^{m}_{t_n}$ and $K^{v}_{t_n}$ using \eqref{eq:kalman_algm} and the associated drift terms $a^{m}_{t_n}$ and $a^{v}_{t_n}$.}
    \State{Update the samples $\left\{\tilde{Z}_{t_{n+1}}^{\left(i\right)}\right\}_{i=1}^{N}$ using \eqref{eq:model_full} with the statistical states $\left\{\bar{u}^{N}_{t}, R^{N}_{t}\right\}$ and observation data $y^{\delta}_{t}$.}
    \State{Update the statistical mean and covariance $\left\{\bar{u}^{N}_{t_{n+1}}, R^{N}_{t_{n+1}}\right\}$ by integrating \eqref{eq:obs_full} to the next time step using the average of all samples.}
\EndFor

\end{algorithmic}

\caption{Ensemble statistical filter with observations in mean and covariance\protect\label{alg:filter_full}}
\end{algorithm}

\begin{rem*}
Solving the equations \eqref{eq:model_full} may still demand high
computational cost for resolving the multiple nonlinear coupling terms
in high dimension $d\gg1$. One potential approach to address the
computational challenge is to adopt the efficient random batch approach
\cite{qi2023high,qi2023random} developed for the coupled models \eqref{eq:closure_model}.
A detailed investigation of the numerical methods will be performed
in the follow-up research.
\end{rem*}

\subsection{Construction of explicit model operators in the analysis step}

In the second and third lines of \eqref{eq:model_full} for the analysis
step update of filtering, we still need to propose explicit expressions
for functions $a_{t}^{m},K_{t}^{m}$ and $a_{t}^{v},K_{t}^{v}$ according
to the observations of the mean and covariance respectively. According
to Theorem~\ref{thm:consist_analy}, the gain function $K_{t}$ needs
to be solved from equation \eqref{eq:coeff_sde}, that is,
\[
-\nabla\cdot\left(K_{t}^{\intercal}\tilde{\rho}_{t}\right)=\tilde{\rho}_{t}\Gamma_{t}^{-2}\left(H\left(z\right)-\bar{H}_{t}\right).
\]
Then the drift function $a_{t}$ can be directly computed from the
solution of $K_{t}$ as 
\[
a_{t}=\nabla\cdot\left(K_{t}\Gamma_{t}^{2}K_{t}^{\intercal}\right)-K_{t}\Gamma_{t}^{2}\nabla\cdot K_{t}^{\intercal}.
\]
In general, it is still difficult to find solutions of the above equations.
By multiplying $H$ on both sides and integrating about $z$, the
identity for $K_{t}$ implies a necessary condition
\begin{equation}
\tilde{\mathbb{E}}\left[K_{t}^{\intercal}\nabla H\right]=\Gamma_{t}^{-2}C_{t}^{H},\label{eq:kalman_gain}
\end{equation}
where $C_{t}^{H}=\tilde{\mathbb{E}}\left[\left(H\left(\tilde{Z}_{t}\right)-\bar{H}_{t}\right)\left(H\left(\tilde{Z}_{t}\right)-\bar{H}_{t}\right)^{\intercal}\right]$
is the covariance of $H$. Therefore, we can first design proper gain
functions $K_{t}$ by solving \eqref{eq:kalman_gain} according to
the specific structures of $H^{m}:\mathbb{R}^{d}\rightarrow\mathbb{R}^{d}$
and $H^{v}:\mathbb{R}^{d}\rightarrow\mathbb{R}^{d^{2}}$ required
in our problem in \eqref{eq:operators_obs} 
\begin{equation}
\begin{aligned}H_{k}^{m}\left(z\right)= & \sum_{m,n}\gamma_{kmn}z_{m}z_{n}=z^{\intercal}A_{k}z,\\
H_{kl}^{v}\left(z\right)= & \sum_{m,n}\gamma_{kmn}z_{m}z_{n}z_{l}+\gamma_{lmn}z_{m}z_{n}z_{k}=\left(z^{\intercal}A_{k}z\right)z_{l}+z_{k}\left(z^{\intercal}A_{l}z\right),
\end{aligned}
\label{eq:obs_func}
\end{equation}
for all $1\leq k,l\leq d$ where we rewrite the quadratic and cubic
functions using the symmetric coefficient matrix, $A_{k}^{\intercal}=A_{k}\in\mathbb{R}^{d\times d}$,
using the assumed structural symmetry in the coupling coefficient
$\gamma_{kmn}$. The resulting gain functions are then constructed
with the following specific expressions.
\begin{prop}
\label{prop:Kalman_gain}The matrix-valued functions $K_{t}^{m}=\tilde{K}^{m}\Gamma_{m,t}^{-2}$
and $K_{t}^{v}=\tilde{K}^{v}\Gamma_{v,t}^{-2}$ with $\tilde{K}^{m}\left(z\right)\in\mathbb{R}^{d\times d}$
and $\tilde{K}^{v}\left(z\right)\in\mathbb{R}^{d\times d^{2}}$ in
the following expressions
\begin{equation}
\begin{aligned}\tilde{K}_{j,k}^{m}\left(z\right)= & \frac{1}{2}z_{j}\left[\left(z^{\intercal}A_{k}z\right)-\bar{H}_{k}^{m}\right],\\
\tilde{K}_{j,kl}^{v}\left(z\right)= & \frac{1}{3}z_{j}\left(z^{\intercal}A_{k}z\right)z_{l}+\frac{1}{3}z_{j}\left(z^{\intercal}A_{l}z\right)z_{k}-\frac{1}{3}\bar{H}_{kl}^{v}.
\end{aligned}
\label{eq:kalman_algm}
\end{equation}
for $1\leq k,l\leq d$ and $1\leq j\leq d$ satisfy the equation \eqref{eq:kalman_gain}
according to the structures of the functions $H^{m}$ and $H^{v}$
in the form of \eqref{eq:obs_func} respectively, and $\bar{H}^{m}=\tilde{\mathbb{E}}\left[H^{m}\left(\tilde{Z}\right)\right]$
and $\bar{H}^{v}=\tilde{\mathbb{E}}\left[H^{v}\left(\tilde{Z}\right)\right]$.
\end{prop}

The proof of Proposition~\ref{prop:Kalman_gain} is put in Appendix~\ref{sec:Detailed-proofs}.
The average terms, $\bar{H}^{m},\bar{H}^{v}$, are already computed
in the statistical equations \eqref{eq:obs_full} thus no additional
computational cost is needed. On the other hand, it is noticed that
\eqref{eq:kalman_algm} can only give a necessary condition for the
gain operators and may not guarantee the original identity for $K_{t}$
in general. However, in the proof of Theorem~\ref{thm:consist_analy},
it shows that \eqref{eq:kalman_gain} is the main relation used to
derive the consistent analysis statistics on the mean of $H\left(\tilde{Z}_{t}\right)$.
Therefore, \eqref{eq:kalman_algm} provides a desirable candidate
for practical implementations of the algorithm concerning the consistency
in the leading moments.

\subsection{Ensemble approximation in the discrete filtering model}

At last, we discuss the accuracy in the finite ensemble approximation.
Let $\tilde{\rho}_{t}^{N}\left(z\right)=\frac{1}{N}\sum_{i=1}^{N}\delta\left(z-\tilde{Z}_{t}^{\left(i\right)}\right)$
be the random field from the empirical PDF of the ensemble approximation
with finite ensemble size $N$, and $\tilde{\rho}_{t}$ is the random
field from \eqref{eq:mf-fpf} conditional on the observations $\mathcal{G}_{t}$.
Additional assumptions are needed for the structures of the model.
\begin{assumption}
\label{assu:model_disc}We assume that the functions $M:\mathbb{R}^{d}\rightarrow\mathbb{R}^{d}$
and $L:\mathbb{R}^{d}\rightarrow\mathbb{R}^{d\times d}$ in the mean
and covariance equations \eqref{eq:model_conti} are Lipschitz continuous,
that is, there is a constant $\beta>0$ so that
\[
\left|M\left(u\right)-M\left(v\right)\right|\leq\beta\left|u-v\right|,\quad\left\Vert L\left(u\right)-L\left(v\right)\right\Vert \leq\beta\left|u-v\right|.
\]
In addition, the coefficients in $\Gamma_{k}\left(R\right)=\sum_{m,n}\gamma_{kmn}R_{mn}$
are uniformly bounded, that is, there exists a constant $C>0$, so
that for all $k,m,n$
\[
\left|\gamma_{kmn}\right|\leq C.
\]
\end{assumption}

First, we direct use the conclusion from the McKean-Vlasov limit of
the $N$-particle system \cite{oelschlager1984martingale,graham1996asymptotic}.
Under the above assumptions of the model coefficients, the ensemble
representation \eqref{eq:model_full} will converge to the continuous
system \eqref{eq:mf-fpf} as $N\rightarrow\infty$. In particular,
assume that the a unique solution exists for the McKean-Vlasov SDE,
and the initial ensemble $\left\{ \tilde{Z}_{0}^{\left(i\right)}\right\} _{i=1}^{N}$
is drawn from i.i.d. random samples. Given the observations $\mathcal{G}_{t}$,
we have for any $\varphi\in C_{b}^{2}\left(\mathbb{R}^{d}\right)$
\begin{equation}
\left\langle \varphi,\tilde{\rho}_{t}^{N}\right\rangle =\frac{1}{N}\sum_{i=1}^{N}\varphi\left(\tilde{Z}_{t}^{\left(i\right)}\right)\rightarrow\mathbb{E}\left[\varphi\left(\tilde{Z}_{t}\right)\mid\mathcal{G}_{t}\right]=\left\langle \varphi,\tilde{\rho}_{t}\right\rangle ,\label{eq:lln}
\end{equation}
a.s. as $N\rightarrow\infty$. Furthermore, there is the error estimate
for the empirical estimate $\mathbb{E}^{N}\varphi=\frac{1}{N}\sum\varphi\left(\tilde{Z}_{t}^{\left(i\right)}\right)$
\begin{equation}
\mathbb{E}\left[\sup_{0\leq t\leq T}\left|\left\langle \varphi,\tilde{\rho}_{t}^{N}\right\rangle -\left\langle \varphi,\tilde{\rho}_{t}\right\rangle \right|^{2}\right]\leq\frac{C_{T}}{N}\left\Vert \varphi\right\Vert _{\infty}^{2}.\label{eq:conv_pdf}
\end{equation}
Detailed proofs on \eqref{eq:lln} and \eqref{eq:conv_pdf} can be
found in such as Chapter 9 of \cite{bain2009fundamentals}.

Next, we consider the finite ensemble and discrete time estimation
of the leading-order mean and covariance from the filter model. With
the ensemble approximation of the density function $\tilde{\rho}_{t}^{N}$,
the mean and covariance estimates are computed from the equations
\eqref{eq:obs_full}
\begin{equation}
\begin{aligned}\frac{\mathrm{d}\bar{u}_{t}^{N,\tau}}{\mathrm{d}t}= & \:M\left(\bar{u}_{\sigma\left(t\right)}^{N,\tau}\right)+F_{\sigma\left(t\right)}+\frac{1}{N}\sum_{i=1}^{N}H^{m}\left(\tilde{Z}_{\sigma\left(t\right)}^{\left(i\right)}\right),\\
\frac{\mathrm{d}R_{t}^{N,\tau}}{\mathrm{d}t}= & \:L\left(\bar{u}_{\sigma\left(t\right)}^{N,\tau}\right)R_{\sigma\left(t\right)}^{N,\tau}+R_{\sigma\left(t\right)}^{N,\tau}L\left(\bar{u}_{\sigma\left(t\right)}^{N,\tau}\right)^{\intercal}+\Sigma_{\sigma\left(t\right)}^{2}+\frac{1}{N}\sum_{i=1}^{N}H^{v}\left(\tilde{Z}_{\sigma\left(t\right)}^{\left(i\right)}\right),
\end{aligned}
\label{eq:model_stat_ens}
\end{equation}
where a finite time integration scheme is also applied with $\sigma\left(t\right)=n\tau$
for $t\in\left[t_{n},t_{n+1}\right]$. Above, we neglect the last
relaxation term with $\epsilon$ since it will automatically vanish
with the resulting consistency. For clarity, we show below again the
continuous forecast model from \eqref{eq:closure_model} 
\begin{equation}
\begin{aligned}\frac{\mathrm{d}\bar{u}_{t}}{\mathrm{d}t}= & \:M\left(\bar{u}_{t}\right)+F_{t}+\tilde{\mathbb{E}}H^{m}\left(\tilde{Z}_{t}\right),\\
\frac{\mathrm{d}R_{t}}{\mathrm{d}t}= & \:L\left(\bar{u}_{t}\right)R_{t}+R_{t}L\left(\bar{u}_{t}\right)^{\intercal}+\Sigma_{t}^{2}+\tilde{\mathbb{E}}H^{v}\left(\tilde{Z}_{t}\right).
\end{aligned}
\label{eq:model_conti}
\end{equation}
Notice that $\bar{u}_{t}^{N},R_{t}^{N}$ and $\bar{u}_{t},R_{t}^ {}$
are stochastic processes due to the random samples $\left\{ \tilde{Z}_{t}^{\left(i\right)}\right\} $
and $\tilde{Z}_{t}$ is still dependent on the observations $\mathcal{G}_{t}$.
We have the following convergence result for the finite ensemble $N$
and finite time step $\tau$ approximation to the continuous model.
\begin{prop}
\label{thm:conv_stat}If Assumption~\ref{assu:model_disc} is satisfied,
under the same initial condition the statistical solution of the ensemble
approximation model \eqref{eq:model_stat_ens} with discrete time
step $\tau$ converges to that of the continuous model \eqref{eq:model_conti}
with
\begin{equation}
\begin{aligned}\mathbb{E}\left[\sup_{n\tau\leq T}\left|\bar{u}_{t_{n}}^{N}-\bar{u}_{t_{n}}\right|^{2}\right] & \leq\left(C_{1,T}\tau+\frac{C_{2,T}}{N}\right)\left\Vert H^{m}\right\Vert _{\infty},\\
\mathbb{E}\left[\sup_{n\tau\leq T}\left|R_{t_{n}}^{N}-R_{t_{n}}\right|^{2}\right] & \leq\left(C_{1,T}^{\prime}\tau+\frac{C_{2,T}^{\prime}}{N}\right)\left(\left\Vert H^{m}\right\Vert _{\infty}+\left\Vert H^{v}\right\Vert _{\infty}\right).
\end{aligned}
\label{eq:stat_bnds}
\end{equation}
\end{prop}

\begin{proof}
First, consider the mean equations from the same initial state, we
have
\begin{align*}
\bar{u}_{t}^{N,\tau}-\bar{u}_{t}= & \int_{0}^{t}\left[M\left(\bar{u}_{\sigma\left(s\right)}^{N,\tau}\right)-M\left(\bar{u}_{s}\right)\right]\mathrm{d}s+\int_{0}^{t}\left[F_{\sigma\left(s\right)}-F_{s}\right]\mathrm{d}s\\
 & +\int_{0}^{t}\left[\left\langle H^{m},\tilde{\rho}_{\sigma\left(s\right)}^{N}\right\rangle -\left\langle H^{m},\tilde{\rho}_{s}\right\rangle \right]\mathrm{d}s.
\end{align*}
Therefore, using H\"{o}lder's inequality and Lipschitz condition
for $M$ there is
\begin{align}
\mathbb{E}\sup_{t\leq T}\left|\bar{u}_{t}^{N,\tau}-\bar{u}_{t}\right|^{2} & \leq3T\beta\mathbb{E}\int_{0}^{T}\left|\bar{u}_{\sigma\left(s\right)}^{N,\tau}-\bar{u}_{s}\right|^{2}\mathrm{d}s\nonumber \\
 & +3T\mathbb{E}\int_{0}^{T}\left|\left\langle H^{m},\tilde{\rho}_{\sigma\left(s\right)}^{N}\right\rangle -\left\langle H^{m},\tilde{\rho}_{s}\right\rangle \right|^{2}\mathrm{d}s+C_{T}\tau\nonumber \\
 & \leq C_{1}\int_{0}^{T}\mathbb{E}\sup_{s^{\prime}\leq s}\left|\bar{u}_{s^{\prime}}^{N,\tau}-\bar{u}_{s^{\prime}}\right|^{2}\mathrm{d}s+C_{2}\tau\left\Vert H^{m}\right\Vert _{\infty}^{2}+\frac{C_{3}}{N}\left\Vert H_{m}\right\Vert _{\infty}^{2}.\label{eq:est_m}
\end{align}
In the first term above on the right hand side, we follow the same
procedure to estimate the error in the corresponding continuous solution
$\bar{u}_{s}^{N,\tau}$ compared to the time discretization solution
$\bar{u}_{\sigma\left(s\right)}^{N,\tau}$
\[
\left|\bar{u}_{\sigma\left(s\right)}^{N,\tau}-\bar{u}_{s}\right|^{2}\leq\left|\bar{u}_{\sigma\left(s\right)}^{N,\tau}-\bar{u}_{s}^{N,\tau}\right|^{2}+\left|\bar{u}_{s}^{N,\tau}-\bar{u}_{s}\right|^{2}\leq C\tau\left\Vert H^{m}\right\Vert _{\infty}^{2}+\left|\bar{u}_{s}^{N,\tau}-\bar{u}_{s}\right|^{2}.
\]
And in the second line for the term related to $H^{m}$, combining
the discrete time estimate with \eqref{eq:conv_pdf} gives
\[
\mathbb{E}\left[\sup_{t\leq T}\left|\left\langle H^{m},\tilde{\rho}_{\sigma\left(s\right)}^{N}\right\rangle -\left\langle H^{m},\tilde{\rho}_{t}\right\rangle \right|^{2}\right]\leq\left(C\tau^{2}+\frac{C_{T}}{N}\right)\left\Vert H^{m}\right\Vert _{\infty}^{2}.
\]
Therefore, applying Gr\"{o}nwall's inequality to \eqref{eq:est_m},
we get the mean estimate in \eqref{eq:stat_bnds}.

Next, under a similar fashion, we can compute from the covariance
equation and using the Lipschitz condition for $L$
\begin{align*}
\mathbb{E}\sup_{t\leq T}\left|R_{t}^{N,\tau}-R_{t}\right|^{2} & \leq C_{1}\beta\mathbb{E}\sup_{t\leq T}\left\Vert \bar{u}_{t}\right\Vert _{\infty}^{2}\int_{0}^{T}\left|R_{\sigma\left(s\right)}^{N,\tau}-R_{s}\right|^{2}\mathrm{d}s\\
 & +C_{2}\beta\mathbb{E}\sup_{t\leq T}\left|\bar{u}_{\sigma\left(t\right)}^{N,\tau}-\bar{u}_{t}\right|^{2}\sup_{t\leq T}\left\Vert R_{\sigma\left(t\right)}\right\Vert ^{2}\\
 & +C_{3}\mathbb{E}\int_{0}^{T}\left|\left\langle H^{v},\tilde{\rho}_{\sigma\left(s\right)}^{N}\right\rangle -\left\langle H^{v},\tilde{\rho}_{s}\right\rangle \right|^{2}\mathrm{d}s+C_{T}\tau.
\end{align*}
Using the uniform boundedness of $\bar{u}_{t},R_{t}$ and \eqref{eq:conv_pdf}
for $H^{v}$ together with the previous estimate for $\mathbb{E}\sup_{t\leq T}\left|\bar{u}_{\sigma\left(t\right)}^{N,\tau}-\bar{u}_{t}\right|^{2}$,
we reach the final covariance estimate in \eqref{eq:stat_bnds}.
\end{proof}
Theorem~ \ref{thm:conv_stat} shows that the discrete approximation
scheme of the approximating filter model can recover the key model
statistics in mean and covariance. It again demonstrates the central
role in achieving accurate prediction of the nonlinear observation
functions $H^{m},H^{v}$ in the filtering method.

\section{Summarizing discussions\protect\label{sec:Summary}}

We developed a systematic statistical filtering strategy that enables
effective ensemble approximation of non-Gaussian probability distributions
of multiscale turbulent states using observations in the leading-order
moments of the mean and covariance. The filtering model is based on
a closed set of coupled stochastic-statistical equations established
for modeling general turbulent dynamical systems involving nonlinear
coupling. The non-Gaussian features then can be characterized by a
McKean-Vlasov SDE taking into account both the stochastic forecast
equation and corrections according to the observed first two moments.
Importantly, the proposed McKean-Vlasov SDE for the finite-dimensional
stochastic state does not require the explicit computation of the
infinite-dimensional probability distribution, but just relies on
the feedbacks from the leading moments that can be computed from the
associated statistical equations. This leads to straightforward numerical
algorithms using ensemble approximations. The stochastic--statistical
formulation offers a flexible approach for recovering essential model
statistics, making it applicable to a wide range of problems in uncertainty
quantification and data assimilation. In the immediate applications
of this research, the performance of the new filtering strategy will
be tested on a series of turbulent systems, starting from prototype
models to realistic applications in really high-dimensional systems.
Also, currently,we are only able to show the statistical consistency
of the approximating filter in the first two moments of the observation
function under restricted conditions. Further explorations exploiting
specific model structures such as the conservation properties will
be used to provide a thorough understanding of the approximation skill
of the filter predictions.

\section*{Acknowledge}

The research of J.-G. L. is partially supported by the National Science
Foundation (NSF) under Grant No. DMS-2106988. The research of D. Q.
is partially supported by Office of Naval Research (ONR) Grant No.
N00014-24-1-2192.

%\newpage
\appendix
\renewcommand\theequation{A\arabic{equation}}
\setcounter{equation}{0}

\section{General background about filtering\protect\label{sec:Background-about-filtering}}

Here, we summarize the useful results needed in the main text of this
paper following mostly \cite{bain2009fundamentals,liptser2013statistics,falb1967infinite}.

\subsection{Filtering equations for general stochastic systems}

Let $\left(\Omega,\mathcal{F},\mathbb{P}\right)$ be the complete
probability space. The \emph{signal process} $u_{t}\in\mathbb{R}^{d}$
and the \emph{observation process} $y_{t}\in\mathbb{R}^{p}$ are defined
on the probability space satisfying the following SDEs\addtocounter{equation}{0}\begin{subequations}\label{eq:filtering}
\begin{align}
\mathrm{d}u_{t} & =F\left(u_{t}\right)\mathrm{d}t+\Sigma\mathrm{d}W_{t},\quad u_{t=0}=u_{0},\label{eq:dyn_filter}\\
\mathrm{d}y_{t} & =H\left(u_{t}\right)\mathrm{d}t+\mathrm{d}B_{t},\quad y_{t=0}=y_{0},\label{eq:obs_filter}
\end{align}
\end{subequations}where $F:\mathbb{R}^{d}\rightarrow\mathbb{R}^{d}$
and $H:\mathbb{R}^{d}\rightarrow\mathbb{R}^{p}$ are bounded and globally
Lipschitz continuous functions, and $W_{t}\in\mathbb{R}^{s},B_{t}\in\mathbb{R}^{p}$
are independent standard Wiener processes with matrix coefficient
$\Sigma\in\mathbb{R}^{d\times s}$. The aim of the general filtering
problem is to determine the conditional probability distribution $\mu_{t}$
of the signal process, $u_{t}$, given the accumulated observation
process, $Y_{t}=\left\{ y_{s},s\leq t\right\} $.

Define the observation filtration $\mathcal{G}_{t}=\sigma\left\{ y_{s},s\leq t\right\} $.
The random conditional distribution $\mu_{t}:\mathbb{R}^{d}\times\Omega\rightarrow\left[0,1\right]$
is defined as the $\mathcal{P}\left(\mathbb{R}^{d}\right)$-valued
stochastic process which is measurable w.r.t. $\mathcal{G}_{t}$,
so that for any function $\varphi\in C_{b}^{2}\left(\mathbb{R}^{d}\right)$
a.s.
\[
\mathbb{E}\left[\varphi\left(u_{t}\right)\mid\mathcal{G}_{t}\right]=\mu_{t}\left(\varphi\right)\coloneqq\int_{\mathbb{R}^{d}}\varphi\left(u\right)\mu_{t}\left(\mathrm{d}u\right).
\]
In particular, the optimal filter solution, $\hat{u}_{t}=\mathbb{E}\left[u_{t}\mid\mathcal{G}_{t}\right]$,
can be defined based on $\mu_{t}$. It shows that $\hat{u}_{t}$ is
the minimizer $\mathbb{E}\left[\left|\hat{u}_{t}-u_{t}\right|^{2}\right]=\min_{v}\mathbb{E}\left[\left|v-u_{t}\right|^{2}\right]$
among all $v\in L^{2}\left(\Omega,\mathcal{G}_{t},\mathbb{P}\right)$
in the set of $\mathcal{G}_{t}$-measurable square-integrable random
variables for any fixed $t$. The filtering equation for the conditional
probability distribution $\mu_{t}$ is verified to satisfy the Kushner-Stratonovich
equation 
\begin{equation}
\mathrm{d}\mu_{t}\left(\varphi\right)=\mu_{t}\left(\mathcal{L}\varphi\right)\mathrm{d}t+\sigma\left(H,\varphi;\mu_{t}\right)\mathrm{d}\nu_{t}.\label{eq:KS_eqn}
\end{equation}
On the right hand side of the above equation, the first term is the
drift due to the infinitesimal generator $\mathcal{L}=F\cdot\nabla+\frac{1}{2}\Sigma\Sigma^{\intercal}:\nabla\nabla$
of the signal process \eqref{eq:dyn_filter}; the second term represents
the correction from the observation process \eqref{eq:obs_filter}.
The innovation process, $\mathrm{d}\nu_{t}=\mathrm{d}y_{t}-\mu_{t}\left(H\right)\mathrm{d}t\in\mathbb{R}^{p}$,
is a $\mathcal{G}_{t}$-Brownian motion under the probability measure
$\mathbb{P}$, and $\sigma\left(H,\varphi;\mu_{t}\right)=\mu_{t}\left(\varphi H^{\intercal}\right)-\mu_{t}\left(\varphi\right)\mu_{t}\left(H^{\intercal}\right)\in\mathbb{R}^{1\times p}$
gives the coefficient with finite quadratic variation, where $\mu_{t}\left(H\right)$
is the componentwise measure of the vector-valued function $H$. In
addition, the filtering equation \eqref{eq:KS_eqn} is shown to have
a unique solution under proper conditions (Theorem 3.30 and 4.19 in
\cite{bain2009fundamentals} and Theorem 7.7 in \cite{xiong2008introduction})
that is also stable (Theorem 2.7 in \cite{budhiraja2003asymptotic}).
Therefore, this guarantees that the solution of the Kushner-Stratonovich
equation \eqref{eq:KS_eqn} uniquely characterizes the filter distribution
$\mu_{t}$ as a $\mathcal{P}\left(\mathbb{R}^{d}\right)$-valued stochastic
process.

Next, assume that the conditional probability $\mu_{t}$ possesses
a square integrable density, $\mu_{t}\left(\mathrm{d}x\right)=\rho_{t}\left(x\right)\mathrm{d}x$,
with respect to the Lebesgue measure. It can be shown under proper
conditions (Corollary 7.18 in \cite{bain2009fundamentals} and \cite{kurtz1999particle}),
the conditional probability solution $\mu_{t}$ of \eqref{eq:filtering}
has a probability density $\rho_{t}\in W_{k}^{2}\left(\mathbb{R}^{d}\right)$.
According to \eqref{eq:KS_eqn} for the conditional probability, $\rho_{t}$
can be found to be the unique solution to the following SPDE
\begin{equation}
\partial_{t}\rho_{t}=\mathcal{L}^{*}\rho_{t}\mathrm{d}t+\rho_{t}\left(H-\bar{H}_{t}\right)^{\intercal}\left(\mathrm{d}y_{t}-\bar{H}_{t}\mathrm{d}t\right),\quad\rho_{t=0}=\rho_{0},\label{eq:KS_den}
\end{equation}
where $\bar{H}_{t}=\int_{\mathbb{R}^{d}}H\left(x\right)\rho_{t}\left(x\right)\mathrm{d}x$
and $\rho_{0}\in L^{2}\left(\mathbb{R}^{d}\right)$ is the absolute
continuous density of $\mu_{0}$. The randomness of the above SPDE
only comes from the innovation process $\mathrm{d}\nu_{t}=\mathrm{d}y_{t}-\mu_{t}\left(H\right)\mathrm{d}t$
as a finite-dimensional white noise in time.

At last, if the functions on the right hand sides of \eqref{eq:filtering}
satisfy the Ornstein-Uhlenbeck processes with matrix coefficients,
that is, $F\left(u\right)=Fu+f_{t}$ and $H\left(u\right)=Hu+h_{t}$.
With Gaussian initial condition, $u_{0}\sim\mathcal{N}\left(\hat{u}_{0},C_{0}\right)$,
the conditional distribution $\mu_{t}=\mathcal{N}\left(\hat{u}_{t},C_{t}\right)$
given $\mathcal{G}_{t}$ in \eqref{eq:KS_eqn} becomes a multivariate
normal distribution, where $\hat{u}_{t}=\mathbb{E}\left[u_{t}\mid\mathcal{G}_{t}\right]$
and $C_{t}=\mathbb{E}\left[\left(u_{t}-\hat{u}_{t}\right)\left(u_{t}-\hat{u}_{t}\right)^{\intercal}\mid\mathcal{G}_{t}\right]$.
The filtering equations for $\hat{u}_{t}\in\mathbb{R}^{d}$ and $C_{t}\in\mathbb{R}^{d\times d}$
are given by the Kalman-Bucy filter \cite{kalman1961new} \addtocounter{equation}{0}\begin{subequations}\label{eq:KB_fil}
\begin{align}
\mathrm{d}\hat{u}_{t} & =\left(F\hat{u}_{t}+f_{t}\right)\mathrm{d}t+K_{t}\left[\mathrm{d}y_{t}-\left(H\hat{u}_{t}+h_{t}\right)\mathrm{d}t\right],\label{eq:KB_m}\\
\dot{C}_{t} & =FC_{t}+C_{t}F^{\intercal}-K_{t}K_{t}^{\intercal}+\Sigma\Sigma^{\intercal}.\label{eq:KB_v}
\end{align}
\end{subequations}with the Kalman gain matrix $K_{t}=C_{t}H^{\intercal}$.
Above, \eqref{eq:KB_m} is an SDE coupled with the deterministic Riccati
equation \eqref{eq:KB_v}.

\subsection{Infinite dimensional filtering in Hilbert space}

It is shown that the linear Kalman-Bucy filter can be generalized
to linear stochastic equations on a Hilbert space \cite{curtain1975survey,falb1967infinite}.
Let $H$ be a Hilbert space. Denote $L^{2}\left(\Omega,\mathcal{G},\mathbb{P};H\right)$
as the collection of all $H$-valued $\mathcal{G}$-measurable square-integrable
random variables. The expectation of $u\in L^{2}\left(\Omega,\mathcal{G},\mathbb{P};H\right)$
is denoted by
\begin{equation}
\mathbb{E}\left[u\right]=\int_{\Omega}u\left(\omega\right)\mathrm{d}\mathbb{P}\left(\omega\right).\label{eq:exp}
\end{equation}
The inner produce for $u,v\in L^{2}\left(\Omega,\mathcal{G},\mathbb{P};H\right)$
can be defined as $\left\langle u,v\right\rangle _{2}=\mathbb{E}\left[\left\langle u,v\right\rangle _{H}\right]=\int_{\Omega}\left\langle u\left(\omega\right),v\left(\omega\right)\right\rangle _{H}\mathrm{d}\mathbb{P}\left(\omega\right)$.
With the above notations, the \emph{covariance operator} $\mathcal{C}$
can be introduced as an element in the linear transformations $\mathcal{L}\left(H,H\right)$.
\begin{defn}
Let $u,v\in L^{2}\left(\Omega,\mathcal{G},\mathbb{P};H\right)$ be
two $H$-valued random variables. Then the covariance of $u$ and
$v$ is given by
\begin{equation}
\mathcal{C}\left(u,v\right)=\mathbb{E}\left[u\otimes v\right]-\mathbb{E}\left[u\right]\otimes\mathbb{E}\left[v\right],\label{eq:cov_op}
\end{equation}
where $u\left(\omega\right)\otimes v\left(\omega\right)\in\mathcal{L}\left(H,H\right)$
is a linear transformation of $H$ into $H$ defined for any $f\in H$
as
\[
\left(u\otimes v\right)f=u\left\langle v,f\right\rangle _{H}.
\]
\end{defn}

It is easy to check that the adjoint $\mathcal{C}\left(u,v\right)^{*}=\mathcal{C}\left(v,u\right)$
and $\mathcal{C}\left(u,u\right)^{*}=\mathcal{C}\left(u,u\right)$
is self-adjoint since
\[
\left\langle f,\left(u\otimes v\right)g\right\rangle _{H}=\left\langle u,f\right\rangle _{H}\left\langle v,g\right\rangle _{H}
\]
Notice that if $H=\mathbb{R}^{d}$ is finite-dimensional, for any
$x,y\in L^{2}\left(\Omega,\mathcal{G},\mathbb{P};\mathbb{R}^{d}\right)$,
$x\otimes y=xy^{\intercal}\in\mathbb{R}^{d\times d}$, then the covariance
$\mathcal{C}\left(x,y\right)\in\mathbb{R}^{d\times d}$ becomes the
$d\times d$ matrix 
\[
\mathcal{C}\left(x,y\right)=\mathbb{E}\left[xy^{\intercal}\right]-\mathbb{E}\left[x\right]\mathbb{E}\left[y\right]^{\intercal}=\mathbb{E}\left[\left(x-\mathbb{E}\left[x\right]\right)\left(y-\mathbb{E}\left[y\right]\right)^{\intercal}\right].
\]

Then, we call the $u_{t}\left(\omega\right)$ from $\left[0,T\right]\times\Omega$
to $H$ an $H$-valued stochastic process. An infinite-dimensional
$H$-valued Wiener process $W_{t}$ can be defined accordingly and
the It\^{o} integral can be generalized to infinite-dimensional Hilbert
space accordingly (see Chapter 2 of \cite{pardoux2007stochastic}
with precise validations). Therefore, the signal process $u_{t}$
of filtering can be given by the following $H$-valued SDE 
\begin{equation}
\mathrm{d}u_{t}=\mathcal{A}_{t}u_{t}\mathrm{d}t+\mathcal{Q}_{t}\mathrm{d}W_{t},\quad u_{t=0}=u_{0},\label{eq:signal}
\end{equation}
where $\mathcal{A}_{t}\in L^{\infty}\left(\left[0,T\right];\mathcal{L}\left(H,H\right)\right)$
is a regulated mapping of $\left[0,T\right]$ into $\mathcal{L}\left(H,H\right)$
(which is further generalized to unbounded operators in \cite{curtain1975infinite})
and $\mathcal{Q}_{t}\in L^{2}\left(\left[0,T\right];\mathcal{L}\left(H,H\right)\right)$,
and $W_{t}$ is the $H$-valued Wiener process. It can be shown (Theorem
2.13 in \cite{pardoux2007stochastic} and Theorem 5.1 in \cite{falb1967infinite})
that \eqref{eq:signal} has unique solution in $L^{2}\left(\Omega,\mathcal{G},\mathbb{P};H\right)$
with initial value $\mathbb{E}\left[\left|u_{0}\right|_{H}^{2}\right]<\infty$.
The observation stochastic process $y_{t}\in\mathbb{R}^{p}$ can be
generated by the SDE
\begin{equation}
\mathrm{d}y_{t}=\mathcal{H}_{t}u_{t}\mathrm{d}t+\mathrm{d}B_{t},\quad y_{t=0}=y_{0},\label{eq:obs-1}
\end{equation}
where $\mathcal{H}_{t}$ is a linear mapping of $\left[0,T\right]$
into $\mathcal{L}\left(H,\mathbb{R}^{p}\right)$, and $B_{t}$ is
the Wiener process in $\mathbb{R}^{p}$ independent of $W_{t}$. The
infinite-dimensional filtering problem can be then described as: given
$y_{s},s\leq t$, determine the optimal estimate $\hat{u}_{t}$ of
$u_{t}$ that minimizes
\begin{equation}
\mathbb{E}\left[\left|u_{t}-v\right|_{H}^{2}\right],\quad v\in L^{2}\left(\Omega,\mathcal{G}_{t},\mathbb{P};H\right),\label{eq:optim}
\end{equation}
where $\mathcal{G}_{t}=\sigma\left\{ y_{s},s\leq t\right\} $ is generated
by the observations up to time $t$.

Finally, in parallel to the Kalman-Bucy filter \eqref{eq:KB_fil}
in finite-dimensional space, similar result can be extended to the
above infinite-dimensional filtering problem \eqref{eq:signal} and
\eqref{eq:obs-1}. Below, we summarize the main results in Theorem
7.10, 7.14 of \cite{falb1967infinite}.
\begin{thm}
The optimal filter solution $\hat{u}_{t}$ of \eqref{eq:optim} exists
and is unique, which satisfies the following infinite-dimensional
SDE
\begin{equation}
\mathrm{d}\hat{u}_{t}=\mathcal{A}_{t}\hat{u}_{t}\mathrm{d}t+\mathcal{K}_{t}\left(\mathrm{d}y_{t}-\mathcal{H}_{t}\hat{u}_{t}\mathrm{d}t\right),\quad\hat{u}_{t=0}=u_{0},\label{eq:fil_mean}
\end{equation}
where $\mathcal{K}_{t}=\mathcal{C}_{t}\mathcal{H}_{t}^{*}\in\mathcal{L}\left(\mathbb{R}^{p},H\right)$
with $\mathcal{H}_{t}^{*}$ the adjoint of $\mathcal{H}_{t}$. And
the covariance operator $\mathcal{C}_{t}=\mathbb{E}\left[\left(u_{t}-\hat{u}_{t}\right)\otimes\left(u_{t}-\hat{u}_{t}\right)\right]$
satisfies the following Riccati equation
\begin{equation}
\dot{\mathcal{C}}_{t}=\mathcal{A}_{t}\mathcal{C}_{t}+\mathcal{C}_{t}\mathcal{A}_{t}^{*}-\mathcal{K}_{t}\mathcal{K}_{t}^{*}+\mathcal{Q}_{t}\mathcal{Q}_{t}^{*},\quad\mathcal{C}_{t=0}=\mathcal{C}\left(u_{0},u_{0}\right).\label{eq:fil_cov}
\end{equation}
\end{thm}

\subsection{Kalman-Bucy filter with conditional Gaussian processes}

The linear Kalman-Bucy filter can be generalized to nonlinear filtering
accepting the conditional Gaussian processes \cite{liptser2013statistics}.
The conditional Gaussian process $\left(v_{t},y_{t}\right),0\leq t\leq T$
is given by the solution of the following coupled equations\addtocounter{equation}{0}\begin{subequations}\label{eq:condGau}
\begin{align}
\mathrm{d}v_{t} & =\left[F_{t}\left(y_{t}\right)v_{t}+f_{t}\left(y_{t}\right)\right]\mathrm{d}t+\Sigma_{t}\mathrm{d}W_{t},\quad v_{t=0}=v_{0},\label{eq:dyn_condGau}\\
\mathrm{d}y_{t} & =\left[H_{t}\left(y_{t}\right)v_{t}+h_{t}\left(y_{t}\right)\right]\mathrm{d}t+\Gamma_{t}\mathrm{d}B_{t},\quad y_{t=0}=y_{0},\label{eq:obs_condGau}
\end{align}
\end{subequations}where $W_{t},B_{t}$ are mutually independent standard
Gaussian white noise processes, and the initial states $\left(v_{0},y_{0}\right)$
are random variables independent of $W_{t},B_{t}$. In general, $v_{t}\in\mathbb{R}^{d}$
represents the signal process and $y_{t}\in\mathbb{R}^{p}$ represents
the observation process. The functions $f_{t},F_{t}$ and $h_{t},H_{t}$
are globally Lipschitz continuous and uniformly bounded on the observed
state $y_{t}$ over the time interval $0\leq t\leq T$. Assume that
the sequence $\left(v_{t},y_{t}\right)$ is obtained from a realization
$\omega$ and the initial condition, $\left\{ v_{0},y_{0}\right\} $.
We can then define the observation sequence $Y_{t}=\left\{ y_{s}\left(\omega\right),s\leq t\right\} $
as well as the unobserved target signal $v_{t}=v_{t}\left(\omega\right)$.
The above system \eqref{eq:condGau} is called the conditional Gaussian
process since the conditional distribution $\mu_{t}=\mathbb{P}\left(v_{t}\in\cdot\mid Y_{t}\right)$
given $Y_{t}$ becomes a Gaussian distribution a.s. if $\mu_{0}=\mathbb{P}\left(v_{0}\in\cdot\mid y_{0}\right)$
is Gaussian (Theorem 12.6 of \cite{liptser2013statistics}).

Next, let $\mathcal{G}_{t}=\sigma\left\{ y_{s},s\leq t\right\} $,
and define the mean $\hat{v}_{t}=\mathbb{E}\left[v_{t}\mid\mathcal{G}_{t}\right]$
and covariance $\hat{C}_{t}=\mathbb{E}\left[\left(v_{t}-\hat{v}_{t}\right)\left(v_{t}-\hat{v}_{t}\right)^{\intercal}\mid\mathcal{G}_{t}\right]$
w.r.t. the conditional Gaussian distribution $\mu_{t}=\mathcal{N}\left(\hat{v}_{t},\hat{C}_{t}\right)$
. Then, it shows that the explicit dynamical equations for $\left(\hat{v}_{t},\hat{C}_{t}\right)$
can be derived based on the conditional Gaussian process \eqref{eq:condGau}.
As a result, filtering equations from the linear Kalman-Bucy filter
\eqref{eq:KB_fil} can be directly applied to the conditional linear
system regardless of its essentially nonlinear dynamics. The equations
of the conditional Gaussian processes and the uniqueness of the solutions
are proved under suitable conditions for the model coefficients in
Chapter 12 of \cite{liptser2013statistics}. We summarize the results
according to Theorem 12.7 of \cite{liptser2013statistics} for the
nonlinear conditional Gaussian filter.
\begin{thm}
The conditional distribution $\mu_{t}$ of the stochastic processes
$v_{t}$ given $Y_{t}$ from \eqref{eq:condGau} is Gaussian, $\mathcal{N}\left(\hat{v}_{t},\hat{C}_{t}\right)$.
Then, with $\Gamma_{t}\Gamma_{t}^{\intercal}\succ0$ and the initial
mean and covariance $\hat{v}_{0},\hat{C}_{0}$, the solutions for
the mean $\hat{v}_{t}$ and covariance matrix $\hat{C}_{t}$ are uniquely
given by the following closed equations\addtocounter{equation}{0}\begin{subequations}\label{eq:KB_cond}
\begin{align}
\mathrm{d}\hat{v}_{t} & =\left[F_{t}\left(y_{t}\right)\hat{v}_{t}+f_{t}\left(y_{t}\right)\right]\mathrm{d}t\nonumber \\
 & +K_{t}\left(y_{t}\right)\left(\Gamma_{t}\Gamma_{t}^{\intercal}\right)^{-1}\left\{ \mathrm{d}y_{t}-\left[H_{t}\left(y_{t}\right)\hat{v}_{t}+h_{t}\left(y_{t}\right)\right]\mathrm{d}t\right\} ,\label{eq:cond_m}\\
\mathrm{d}\hat{C}_{t} & =\left[F_{t}\left(y_{t}\right)\hat{C}_{t}+\hat{C}_{t}F_{t}\left(y_{t}\right)^{\intercal}+\Sigma_{t}\Sigma_{t}^{\intercal}\right]\mathrm{d}t\nonumber \\
 & -K_{t}\left(y_{t}\right)\left(\Gamma_{t}\Gamma_{t}^{\intercal}\right)^{-1}K_{t}^{\intercal}\left(y_{t}\right)\mathrm{d}t,\label{eq:cond_v}
\end{align}
\end{subequations}where $K_{t}\left(y_{t}\right)=\hat{C}_{t}\left(y_{t}\right)H_{t}\left(y_{t}\right)^{\intercal}$.
The matrix $\hat{C}_{t}$ will remain positive-definite for all $0\leq t\leq T$
if $\hat{C}_{0}\succ0$.
\end{thm}

\renewcommand\theequation{B\arabic{equation}}
\setcounter{equation}{0}

\section{Detailed proofs of theorems\protect\label{sec:Detailed-proofs}}
\begin{proof}
[Proof of Lemma \ref{lem:consist_model}]First, consider the model
\eqref{eq:closure_model} without the relaxation term in $R_{t}$,
that is, set $\epsilon^{-1}=0$. Applying It\^{o}'s formula for $Z_{t}$
with any test function $\varphi\in C_{b}^{2}\left(\mathbb{R}^{d}\right)$
gives
\begin{align}
\mathrm{d}\varphi\left(Z_{t}\right)= & \mathcal{L}\left(\bar{u}_{t},R_{t}\right)\varphi\left(Z_{t}\right)\mathrm{d}t+\nabla\varphi\left(Z_{t}\right)^{\intercal}\Sigma_{t}\mathrm{d}W_{t}\nonumber \\
= & \nabla\varphi\left(Z_{t}\right)^{\intercal}\left[L\left(\bar{u}_{t}\right)Z_{t}+\varGamma\left(Z_{t}Z_{t}^{\intercal}-R_{t}\right)\right]\mathrm{d}t\\
 & +\frac{1}{2}\Sigma_{t}\Sigma_{t}^{\intercal}:\nabla\nabla\varphi\left(Z_{t}\right)\mathrm{d}t+\nabla\varphi\left(Z_{t}\right)^{\intercal}\Sigma_{t}\mathrm{d}W_{t},\label{eq:ito_z}
\end{align}
where $\mathcal{L}$ is the generator of $Z_{t}$. Given any statistical
solution $\left(\bar{u}_{t},R_{t}\right)$ and taking expectation
using $\varphi\left(Z\right)=Z$, the equation for the first moment
of $Z_{t}$ can be found as
\begin{equation}
\frac{\mathrm{d}}{\mathrm{d}t}\mathbb{E}\left[Z_{t}\right]=\left[L\left(\bar{u}_{t}\right)\mathbb{E}\left[Z_{t}\right]+\varGamma\left(\mathbb{E}\left[Z_{t}Z_{t}^{\intercal}\right]-R_{t}\right)\right].\label{eq:mean_z}
\end{equation}
Next by taking $\varphi\left(Z\right)=Z_{k}Z_{l}$, we have
\begin{align*}
\mathrm{d}\left(Z_{k.t}Z_{l,t}\right)= & \sum_{m}\left[L_{km}\left(\bar{u}_{t}\right)Z_{m.t}Z_{l,t}+Z_{k.t}Z_{m,t}L_{lm}\left(\bar{u}_{t}\right)\right]\mathrm{d}t+\Sigma_{km,t}\Sigma_{lm,t}\mathrm{d}t\\
 & +\sum_{m,n}\gamma_{mnk}\left(Z_{m,t}Z_{n,t}Z_{l,t}-R_{mn,t}Z_{l,t}\right)\mathrm{d}t+\gamma_{mnl}\left(Z_{m,t}Z_{n,t}Z_{k,t}-R_{mn,t}Z_{k,t}\right)\mathrm{d}t\\
 & +\sum_{m}\Sigma_{km,t}Z_{l,t}\mathrm{d}W_{m,t}+\Sigma_{lm,t}Z_{k,t}\mathrm{d}W_{m,t}.
\end{align*}
This implies the second moment equation of $Z_{t}$ as
\begin{align}
\frac{\mathrm{d}}{\mathrm{d}t}\mathbb{E}\left[Z_{t}Z_{t}^{\intercal}\right]= & \sum_{m}\left[L_{km}\left(\bar{u}_{t}\right)\mathbb{E}\left[Z_{m.t}Z_{l,t}\right]+\mathbb{E}\left[Z_{k.t}Z_{m,t}\right]L_{lm}\left(\bar{u}_{t}\right)\right]\mathrm{d}t+\Sigma_{t}\Sigma_{t}^{\intercal}\label{eq:cov_z}\\
 & +\sum_{m,n}\gamma_{mnk}\left(\mathbb{E}\left[Z_{m,t}Z_{n,t}Z_{l,t}\right]-R_{mn,t}\mathbb{E}\left[Z_{l,t}\right]\right)\mathrm{d}t+\gamma_{mnl}\left(\mathbb{E}\left[Z_{m,t}Z_{n,t}Z_{k,t}\right]-R_{mn,t}\mathbb{E}\left[Z_{k,t}\right]\right)\mathrm{d}t.\nonumber 
\end{align}
Assuming $\mathbb{E}\left[Z_{t}\right]=0$ and $\mathbb{E}\left[Z_{t}Z_{t}^{\intercal}\right]=R_{t}$
for any time instant $t$, first we see that the right hand side of
\eqref{eq:mean_z} will always stay zero. Then, with the same statistics
in the third moments of $Z_{t}$, the right hand side of \eqref{eq:cov_z}
becomes equal to the right hand of the statistical equation of $R_{t}$
in \eqref{eq:closure_model} with \eqref{eq:stat_feedbacks}. Uniqueness
of solution in the statistical equations given the same initial values
implies that the leading two moments of $Z_{t}$ satisfy $\mathbb{E}\left[Z_{t}\right]=0$
and $\mathbb{E}\left[Z_{t}Z_{t}^{\intercal}\right]=R_{t}$ for all
$t>0$.

Finally, by adding the additional relaxation term, $\epsilon^{-1}\left(\mathbb{E}\left[Z_{t}Z_{t}^{\intercal}\right]-R_{t}\right)$
in the dynamics of $R_{t},$ the statistical consistency condition
\eqref{eq:consistent} guarantees $R_{t}=\mathbb{E}\left[Z_{t}Z_{t}^{\intercal}\right]$
for all the time. Thus this term gives zero contribution, so the model
\eqref{eq:closure_model} will also produce the same statistical solution.
\end{proof}
\
\begin{proof}
[Proof of Proposition \ref{prop:model_equiv}]First, consider $\mathbb{E}\varphi\left(u_{t}\right)$
with any test function $\varphi\in C_{b}^{2}\left(\mathbb{R}^{d}\right)$
w.r.t. the PDF $p_{t}$ for the state of the original system \eqref{eq:abs_formu}.
It\^{o}'s lemma shows that
\[
\frac{\mathrm{d}\mathbb{E}\varphi\left(u_{t}\right)}{\mathrm{d}t}=\mathbb{E}\left[\left(\Lambda u_{t}+B\left(u_{t},u_{t}\right)+F_{t}\right)\cdot\nabla\varphi\left(u_{t}\right)\right]+\frac{1}{2}\mathbb{E}\left[\sigma_{t}\sigma_{t}^{\intercal}:\nabla\nabla\varphi\left(u_{t}\right)\right],
\]
where $A:\nabla\nabla f=\sum_{mn}a_{mn}\partial_{u_{m}}\partial_{u_{n}}f$.
By taking $\varphi=u$ and introducing the decomposition $u^{\prime}=u-\mathbb{E}u=\sum_{k}u_{k,t}^{\prime}\hat{v}_{k}$
with $u_{k,t}^{\prime}=\hat{v}_{k}\cdot u_{t}^{\prime}$, we have
\begin{align}
\frac{\mathrm{d}\mathbb{E}u_{t}}{\mathrm{d}t} & =\mathbb{E}\left[\Lambda u_{t}+B\left(u,u\right)+F_{t}\right]\nonumber \\
 & =\Lambda\mathbb{E}u_{t}+B\left(\mathbb{E}u,\mathbb{E}u\right)+\sum_{k,l}\mathbb{E}\left[u_{k,t}^{\prime}u_{l,t}^{\prime}\right]B\left(\hat{v}_{k},\hat{v}_{l}\right)+F_{t}.\label{eq:mean_t}
\end{align}
Above, we use the bilinearity of the operator $B$ and $\mathbb{E}u^{\prime}=0$,
such that
\begin{align*}
\mathbb{E}B\left(u,u\right) & =\mathbb{E}B\left(\mathbb{E}u+u^{\prime},\mathbb{E}u+u^{\prime}\right)\\
 & =\mathbb{E}\left[B\left(\mathbb{E}u,\mathbb{E}u\right)+B\left(\mathbb{E}u,u^{\prime}\right)+B\left(u^{\prime},\mathbb{E}u\right)+B\left(u^{\prime},u^{\prime}\right)\right]\\
 & =B\left(\mathbb{E}u,\mathbb{E}u\right)+\mathbb{E}\sum_{k,l}B\left(u_{k,t}^{\prime}\hat{v}_{k},u_{l,t}^{\prime}\hat{v}_{l}\right)\\
 & =B\left(\mathbb{E}u,\mathbb{E}u\right)+\sum_{k,l}\mathbb{E}\left[u_{k,t}^{\prime}u_{l,t}^{\prime}\right]B\left(\hat{v}_{k},\hat{v}_{l}\right).
\end{align*}
Similarly, by taking $\varphi=\left(\hat{v}_{k}\cdot u_{t}^{\prime}\right)\left(u_{t}^{\prime}\cdot\hat{v}_{l}\right)=u_{t}^{\prime\intercal}Au_{t}^{\prime}$
with $A=\hat{v}_{l}\hat{v}_{k}^{\intercal}$, we find 
\begin{align}
\frac{\mathrm{d}}{\mathrm{d}t}\mathbb{E}\left[u_{k,t}^{\prime}u_{l,t}^{\prime}\right]= & \mathbb{E}\left[\left(\Lambda u_{t}+B\left(u_{t},u_{t}\right)+F_{t}\right)\cdot\left(A+A^{\intercal}\right)u_{t}^{\prime}\right]+\frac{1}{2}\sigma_{t}\sigma_{t}^{\intercal}:\left(A+A^{\intercal}\right)\nonumber \\
= & \mathbb{E}\left[\left(u_{t}^{\prime}\cdot\hat{v}_{l}\right)\left(\hat{v}_{k}^{\intercal}\Lambda u_{t}^{\prime}\right)+\left(u_{t}^{\prime\intercal}\Lambda^{\intercal}\hat{v}_{l}\right)\left(\hat{v}_{k}\cdot u_{t}^{\prime}\right)\right]\nonumber \\
 & +\mathbb{E}\left[\hat{v}_{l}\cdot B\left(\mathbb{E}u,u^{\prime}\right)u_{k,t}^{\prime}+\hat{v}_{k}\cdot B\left(\mathbb{E}u,u^{\prime}\right)u_{l,t}^{\prime}\right]+\mathbb{E}\left[\hat{v}_{l}\cdot B\left(u^{\prime},\mathbb{E}u\right)u_{k,t}^{\prime}+\hat{v}_{k}\cdot B\left(u^{\prime},\mathbb{E}u\right)u_{l,t}^{\prime}\right]\nonumber \\
 & +\mathbb{E}\left[B\left(u_{t}^{\prime},u_{t}^{\prime}\right)\cdot\left(\hat{v}_{l}u_{k,t}^{\prime}+\hat{v}_{k}u_{l,t}^{\prime}\right)\right]+\left(\sigma_{t}\cdot\hat{v}_{k}\right)\left(\hat{v}_{l}\cdot\sigma_{t}\right)\nonumber \\
= & \sum_{m}\left(\hat{v}_{k}^{\intercal}\Lambda\hat{v}_{m}\right)\mathbb{E}\left[u_{m,t}^{\prime}u_{l,t}^{\prime}\right]+\mathbb{E}\left[u_{k,t}^{\prime}u_{m,t}^{\prime}\right]\left(\hat{v}_{l}^{\intercal}\Lambda^{\intercal}\hat{v}_{m}\right)+\left(\sigma_{t}\cdot\hat{v}_{k}\right)\left(\hat{v}_{l}\cdot\sigma_{t}\right)\label{eq:cov_t}\\
 & +\sum_{m}\left[\hat{v}_{k}^{\intercal}B\left(\mathbb{E}u,\hat{v}_{m}\right)\mathbb{E}\left[u_{m,t}^{\prime}u_{l,t}^{\prime}\right]+\mathbb{E}\left[u_{k,t}^{\prime}u_{m,t}^{\prime}\right]\hat{v}_{l}^{\intercal}B\left(\mathbb{E}u,\hat{v}_{m}\right)\right]\nonumber \\
 & +\sum_{m}\left[\hat{v}_{k}^{\intercal}B\left(\hat{v}_{m},\mathbb{E}u\right)\mathbb{E}\left[u_{m,t}^{\prime}u_{l,t}^{\prime}\right]+\mathbb{E}\left[u_{k,t}^{\prime}u_{m,t}^{\prime}\right]\hat{v}_{l}^{\intercal}B\left(\hat{v}_{m},\mathbb{E}u\right)\right]\nonumber \\
 & +\sum_{m,n}\left[\hat{v}_{k}^{\intercal}B\left(\hat{v}_{m},\hat{v}_{n}\right)\mathbb{E}\left[u_{l,t}^{\prime}u_{m,t}^{\prime}u_{n,t}^{\prime}\right]+\hat{v}_{l}^{\intercal}B\left(\hat{v}_{m},\hat{v}_{n}\right)\mathbb{E}\left[u_{k,t}^{\prime}u_{m,t}^{\prime}u_{n,t}^{\prime}\right]\right].\nonumber 
\end{align}
Above in the second equality, we use the projection of modes $\left(A+A^{\intercal}\right)u_{t}^{\prime}=\left(\hat{v}_{l}\hat{v}_{k}^{\intercal}+\hat{v}_{k}\hat{v}_{l}^{\intercal}\right)u_{t}^{\prime}=\hat{v}_{l}u_{k,t}^{\prime}+\hat{v}_{k}u_{k,t}^{\prime}$,
and the bilinearity of the quadratic operator $B$; and in the third
equality, we use the decomposition of fluctuation modes, $u_{t}^{\prime}=\sum_{k}u_{k,t}^{\prime}\hat{v}_{k}$.
We find the coupling operator $L_{km}=\left(\hat{v}_{k}^{\intercal}\Lambda\hat{v}_{m}\right)+\hat{v}_{k}^{\intercal}B\left(\hat{v}_{m},\mathbb{E}u\right)+\hat{v}_{k}^{\intercal}B\left(\hat{v}_{m},\mathbb{E}u\right)$,
the third-order coupling coefficients $\gamma_{mnk}=\hat{v}_{k}^{\intercal}B\left(\hat{v}_{m},\hat{v}_{n}\right)$,
as well as the noise term $\left(\sigma_{t}\cdot\hat{v}_{k}\right)\left(\hat{v}_{l}\cdot\sigma_{t}\right)$.

In addition, by subtracting the mean equation \eqref{eq:mean_t} from
the original system \eqref{eq:abs_formu}, we find the SDE for the
stochastic state
\begin{align*}
\frac{\mathrm{d}u_{t}^{\prime}}{\mathrm{d}t}=\frac{\mathrm{d}}{\mathrm{d}t}\left(u_{t}-\mathbb{E}u_{t}\right)= & \Lambda u_{t}+B\left(u_{t},u_{t}\right)+\sigma_{t}\dot{W}_{t}\\
 & -\Lambda\mathbb{E}u_{t}-B\left(\mathbb{E}u,\mathbb{E}u\right)-\sum_{m,n}\mathbb{E}\left[u_{m,t}^{\prime}u_{n,t}^{\prime}\right]B\left(\hat{v}_{m},\hat{v}_{n}\right)\\
= & \sum_{m}u_{m,t}^{\prime}\left[\Lambda\hat{v}_{m}+B\left(\mathbb{E}u,\hat{v}_{m}\right)+B\left(\hat{v}_{m},\mathbb{E}u\right)\right]\\
 & +\sum_{m,n}\left[u_{m,t}^{\prime}u_{n,t}^{\prime}-\mathbb{E}\left[u_{m,t}^{\prime}u_{n,t}^{\prime}\right]\right]B\left(\hat{v}_{m},\hat{v}_{n}\right)+\sigma_{t}\dot{W}_{t}.
\end{align*}
Again in the second equality, we use the spectral decomposition of
the fluctuation state $u_{k,t}^{\prime}=u_{t}^{\prime}\cdot\hat{v}_{k}$.
By projecting the state $u_{t}^{\prime}$ on the basis $\hat{v}_{k}$,
we have
\begin{equation}
\frac{\mathrm{d}u_{k,t}^{\prime}}{\mathrm{d}t}=\sum_{k}L_{km}\left(\mathbb{E}u_{t}\right)u_{m,t}^{\prime}+\sum_{m,n}\gamma_{mnk}\left[u_{m,t}^{\prime}u_{n,t}^{\prime}-\mathbb{E}\left[u_{m,t}^{\prime}u_{n,t}^{\prime}\right]\right]+\hat{v}_{k}^{\intercal}\sigma_{t}\dot{W}_{t},\label{eq:sde_t}
\end{equation}
with the same parameters $L_{km}$ and $\gamma_{mnk}$ defined before.
The generator $\mathcal{L}_{t}^{u}$ of $u_{t}^{\prime}$ can be written
as 
\[
\mathcal{L}_{t}^{u}\left(p_{t}\right)=\left[\sum_{k}L_{km}\left(\mathbb{E}u_{t}\right)u_{m,t}^{\prime}+\sum_{m,n}\gamma_{mnk}\left(u_{m,t}^{\prime}u_{n,t}^{\prime}-\mathbb{E}\left[u_{m,t}^{\prime}u_{n,t}^{\prime}\right]\right)\right]\cdot\nabla_{u^{\prime}}+\frac{1}{2}\Sigma_{t}\Sigma_{t}^{\intercal}:\nabla_{u^{\prime}}\nabla_{u^{\prime}},
\]
where $\mathcal{L}_{t}^{u}$ is dependent on $p_{t}$ in computing
the expectations. 

Next, we consider the closure model \eqref{eq:closure_model} without
the relaxation term
\begin{equation}
\begin{aligned}\frac{\mathrm{d}\bar{u}_{t}}{\mathrm{d}t}=\Lambda\bar{u}_{t}+B\left(\bar{u}_{t},\bar{u}_{t}\right) & +Q_{m}\left(\mathbb{E}\left[Z_{t}\otimes Z_{t}\right]\right)+F_{t},\\
\frac{\mathrm{d}R_{t}}{\mathrm{d}t}=L\left(\bar{u}_{t}\right)R_{t}+R_{t}L\left(\bar{u}_{t}\right) & +Q_{v}\left(\mathbb{E}\left[Z_{t}\otimes Z_{t}\otimes Z_{t}\right]\right)+\Sigma_{t}\Sigma_{t}^{\intercal},\\
\frac{\mathrm{d}}{\mathrm{d}t}\mathbb{E}\left[\varphi\left(Z_{t}\right)\right]= & \:\mathbb{E}\left[\mathcal{L}_{t}\left(\bar{u}_{t},R_{t}\right)\varphi\left(Z_{t}\right)\right].
\end{aligned}
\label{eq:model}
\end{equation}
where $\mathcal{L}_{t}$ is the generator from \eqref{eq:dyn_pdf}
defined from the McKean-Vlasov SDE of $Z_{t}$
\[
\mathcal{L}_{t}\left(\bar{u}_{t},R_{t}\right)=\left[L\left(\bar{u}_{t}\right)Z_{t}+\varGamma\left(Z_{t}Z_{t}^{\intercal}-R_{t}\right)\right]\cdot\nabla_{z}+\frac{1}{2}\Sigma_{t}\Sigma_{t}^{\intercal}:\nabla_{z}\nabla_{z}.
\]
The closure terms $Q_{m}$ and $Q_{v}$ in \eqref{eq:stat_feedbacks}
have exactly the same structure as that in the original system derived
in \eqref{eq:mean_t} and \eqref{eq:cov_t}. In addition, by comparing
the above SDEs \eqref{eq:sde_t} and \eqref{eq:dyn_stoc}, it is realized
that their generators, $\mathcal{L}_{t}^{u}\left(p_{t}\right)$ and
$\mathcal{L}_{t}\left(\bar{u}_{t},R_{t}\right)$, share the same dynamical
structure with the dependence on the first two moments w.r.t. $p_{t}$
and the statistical solutions $\bar{u}_{t},R_{t}$ in \eqref{eq:model}.
Therefore, at any time instant $t$ if we assume consistent statistics
\[
\mathbb{E}_{p_{t}}\left[u_{t}\right]=\bar{u}_{t},\quad\mathbb{E}_{p_{t}}\left[\left(u_{t}^{\prime}\cdot\hat{v}_{k}\right)\left(u_{t}^{\prime}\cdot\hat{v}_{l}\right)\right]=R_{kl,t},
\]
as well as 
\[
\mathbb{E}_{p_{t}}\left[\varphi\left(u_{t}^{\prime}\right)\right]=\mathbb{E}\left[\varphi\left(\sum_{k=1}^{d}Z_{k,t}\hat{v}_{k}\right)\right],
\]
the right hand sides of the original model \eqref{eq:mean_t},  \eqref{eq:cov_t},
and \eqref{eq:sde_t} become the same as that of the closure model
\eqref{eq:model}. Starting from the same initial condition with uniqueness
of the solution, it directly implies that the statistical solutions
of the two systems \eqref{eq:abs_formu} and \eqref{eq:model} will
remain the same during the entire time evolution.
\end{proof}
\
\begin{proof}
[Proof of Lemma \ref{lem:analysis}]We rewrite the filter model \eqref{eq:mf-fpf-quad}
for $\tilde{Z}_{t}$ by substituting the explicit equation for the
observation process, $\mathrm{d}y_{t}=\left[\mathcal{H}\rho_{t}+h_{t}\left(y_{t}\right)\right]\mathrm{d}t+\Gamma_{t}\mathrm{d}B_{t}$,
in \eqref{eq:dyn_obs}
\[
\mathrm{d}\tilde{Z}_{t}=a_{t}\left(\tilde{Z}_{t}\right)\mathrm{d}t+K_{t}\left(\tilde{Z}_{t}\right)\left\{ \left[\mathcal{H}\rho_{t}-H\left(\tilde{Z}_{t}\right)\right]\mathrm{d}t+\Gamma_{t}\mathrm{d}B_{t}-\Gamma_{t}\mathrm{d}\tilde{B}_{t}\right\} .
\]
By applying It\^{o}'s formula on the above SDE, we have for $\varphi\in C_{b}^{2}\left(\mathbb{R}^{d}\right)$
\begin{align}
\mathrm{d}\varphi\left(\tilde{Z}_{t}\right)= & \nabla\varphi\cdot\left[\left(a_{t}-K_{t}\left(H\left(\tilde{Z}_{t}\right)+h_{t}\left(y_{t}\right)\right)\right)\right]\mathrm{d}t\nonumber \\
 & -\nabla\varphi\cdot K_{t}\Gamma_{t}\mathrm{d}\tilde{B}_{t}+\nabla\varphi\cdot K_{t}\mathrm{d}y_{t}+K_{t}\Gamma_{t}^{2}K_{t}^{\intercal}:\nabla\nabla\varphi\mathrm{d}t,\label{eq:ito_phi}
\end{align}
where we define $A:\nabla\nabla\varphi=\sum_{i,j=1}^{d}A_{ij}\partial_{z_{i}z_{j}}\varphi$
and take the convention $\left(\nabla f\right)_{ij}=\partial_{z_{i}}f_{j}$
for the gradient of vector-valued functions $f\in C^{1}\left(\mathbb{R}^{d};\mathbb{R}^{p}\right)$.
Notice that above the coefficient in the last term is $1$ considering
the additional contributions from the independent white noise process
$\Gamma_{t}\mathrm{d}B_{t}=\mathrm{d}y_{t}-\left[\mathcal{H}\rho_{t}+h_{t}\left(y_{t}\right)\right]\mathrm{d}t$
in the observation process besides the original $\mathrm{d}\tilde{B}_{t}$,
that is, 
\[
\frac{1}{2}\nabla\nabla\varphi:\mathrm{d}\left\langle K\Gamma\tilde{B},K\Gamma\tilde{B}\right\rangle _{t}+\frac{1}{2}\nabla\nabla\varphi:\mathrm{d}\left\langle K\Gamma B,K\Gamma B\right\rangle _{t}=\nabla\nabla\varphi:K_{t}\Gamma_{t}^{2}K_{t}^{\intercal}\mathrm{d}t,
\]
where we denote $\left\langle M,N\right\rangle _{t}$ as the Meyer's
process of two martingales $M_{t}$ and $N_{t}$.

First, by taking $\varphi\left(z\right)=H\left(z\right)$ and taking
expectation $\tilde{\mathbb{E}}$ w.r.t. $\tilde{\rho}_{t}$ conditional
on $Y_{t}=\left\{ y_{s},s\leq t\right\} \in\mathcal{G}_{t}$, we have
\begin{align}
\mathrm{d}\tilde{\mathbb{E}}H\left(\tilde{Z}_{t}\right)= & \tilde{\mathbb{E}}\left[\nabla H\left(\tilde{Z}_{t}\right)^{\intercal}a_{t}\right]\mathrm{d}t-\tilde{\mathbb{E}}\left[\nabla H\left(\tilde{Z}_{t}\right)^{\intercal}K_{t}\left(\bar{H}_{t}+H_{t}^{\prime}+h_{t}\left(y_{t}\right)\right)\right]\mathrm{d}t\nonumber \\
 & +\tilde{\mathbb{E}}\left[\nabla H\left(\tilde{Z}_{t}\right)^{\intercal}K_{t}\right]\mathrm{d}y_{t}+\tilde{\mathbb{E}}\left[K_{t}\Gamma_{t}^{2}K_{t}^{\intercal}:\nabla\nabla H\left(\tilde{Z}_{t}\right)\right]\mathrm{d}t\nonumber \\
= & \tilde{\mathbb{E}}\left[\nabla H\left(\tilde{Z}_{t}\right)^{\intercal}a_{t}\right]\mathrm{d}t+\tilde{\mathbb{E}}\left[K_{t}\Gamma_{t}^{2}K_{t}^{\intercal}:\nabla\nabla H\left(\tilde{Z}_{t}\right)\right]\mathrm{d}t-\tilde{\mathbb{E}}\left[\nabla H\left(\tilde{Z}_{t}\right)^{\intercal}K_{t}H_{t}^{\prime}\right]\mathrm{d}t\label{eq:mean_h}\\
 & +\tilde{\mathbb{E}}\left[\nabla H\left(\tilde{Z}_{t}\right)^{\intercal}K_{t}\right]\left\{ \mathrm{d}y_{t}-\left[\bar{H}_{t}+h_{t}\left(y_{t}\right)\right]\mathrm{d}t\right\} .\nonumber 
\end{align}
In the first line above, we split $H\left(\tilde{Z}_{t}\right)=\bar{H}_{t}+H_{t}^{\prime}$.
Notice that the observation process $y_{t}\in\mathcal{G}_{t}$ can
be brought out of the expectation $\tilde{\mathbb{E}}\left[\cdot\right]=\mathbb{E}\left[\cdot\mid\mathcal{G}_{t}\right]$.
Using the first identity in \eqref{eq:coeff_sde} for $a_{t}$, there
is
\begin{align*}
\tilde{\mathbb{E}}\left[\nabla H\left(\tilde{Z}_{t}\right)^{\intercal}a_{t}\right]= & \int\nabla H\left(z\right)^{\intercal}\left[\nabla\cdot\left(K_{t}\Gamma_{t}^{2}K_{t}^{\intercal}\right)-K_{t}\Gamma_{t}^{2}\nabla\cdot K_{t}^{\intercal}\right]\tilde{\rho}_{t}\left(z\right)\mathrm{d}z\\
= & \int\nabla H\left(z\right)^{\intercal}\left[\nabla\cdot\left(\tilde{\rho}_{t}K_{t}\Gamma_{t}^{2}K_{t}^{\intercal}\right)-K_{t}\Gamma_{t}^{2}K_{t}^{\intercal}\nabla\tilde{\rho}_{t}-K_{t}\Gamma_{t}^{2}\nabla\cdot K_{t}^{\intercal}\tilde{\rho}_{t}\right]\mathrm{d}z\\
= & -\int\nabla\nabla H\left(z\right):\left(K_{t}\Gamma_{t}^{2}K_{t}^{\intercal}\right)\tilde{\rho}_{t}\mathrm{d}z-\int\nabla H\left(z\right)^{\intercal}K_{t}\Gamma_{t}^{2}\nabla\cdot\left(\tilde{\rho}_{t}K_{t}^{\intercal}\right)\mathrm{d}z\\
= & -\tilde{\mathbb{E}}\left[\left(K_{t}\Gamma_{t}^{2}K_{t}^{\intercal}\right):\nabla\nabla H\left(\tilde{Z}_{t}\right)\right]-\int\nabla H\left(z\right)^{\intercal}K_{t}\Gamma_{t}^{2}\nabla\cdot\left(\tilde{\rho}_{t}K_{t}^{\intercal}\right)\mathrm{d}z.
\end{align*}
Then using the second identity in \eqref{eq:coeff_sde} for $K_{t}$
and denoting $H_{t}^{\prime}=H-\bar{H}_{t}$, the last term above
gets simplified to
\[
-\int\nabla H{}^{\intercal}K_{t}\Gamma_{t}^{2}\nabla\cdot\left(\tilde{\rho}_{t}K_{t}^{\intercal}\right)\mathrm{d}z=\int\nabla H_{t}^{\prime\intercal}K_{t}\tilde{\rho}_{t}H_{t}^{\prime}\mathrm{d}z=\tilde{\mathbb{E}}\left[\nabla H_{t}^{\prime\intercal}K_{t}H_{t}^{\prime}\right].
\]
With the above identities, first line of \eqref{eq:mean_h} becomes
zero. Further with the second identity in \eqref{eq:coeff_sde} for
$K_{t}$, there is
\[
\tilde{\mathbb{E}}\left(K_{t}^{\intercal}\nabla\psi\right)=\Gamma_{t}^{-2}\tilde{\mathbb{E}}\left[\left(H\left(\tilde{Z}_{t}\right)-\mathbb{E}H\right)\psi\left(\tilde{Z}_{t}\right)^{\intercal}\right],
\]
for any regular function $\psi$ with $\mathbb{E}\psi=0$. By taking
$\psi=H-\bar{H}_{t}$, there is 
\[
\tilde{\mathbb{E}}\left[\nabla H\left(\tilde{Z}_{t}\right)^{\intercal}K_{t}\right]=\tilde{\mathbb{E}}\left[\left(H\left(\tilde{Z}_{t}\right)-\bar{H}_{t}\right)\left(H\left(\tilde{Z}_{t}\right)-\bar{H}_{t}\right)^{\intercal}\right]\Gamma^{-2}=C_{t}^{H}\Gamma_{t}^{-2}.
\]
This gives the equation for $\bar{H}_{t}=\tilde{\mathbb{E}}H\left(\tilde{Z}_{t}\right)$. 

Next, we take $\varphi\left(z\right)=H_{k}\left(z\right)H_{l}\left(z\right)$.
For the convenience of computation, we separate the mean state $\bar{H}_{t}$
as
\[
\varphi\left(z\right)=\left[\bar{H}_{k,t}+H_{k}^{\prime}\left(z\right)\right]\left[\bar{H}_{l,t}+H_{l}^{\prime}\left(z\right)\right]=H_{k,t}^{\prime}\left(z\right)H_{l,t}^{\prime}\left(z\right)+\bar{H}_{k,t}H_{l}^{\prime}\left(z\right)+H_{k}^{\prime}\left(z\right)\bar{H}_{l,t}+\bar{H}_{k,t}\bar{H}_{l,t}.
\]
The last term above is independent of $z$, thus will vanish after
applying It\^{o}'s formula \eqref{eq:ito_phi}. We have for the first
term on the right hand side
\begin{align}
\mathrm{d}\tilde{\mathbb{E}}\left[H_{k,t}^{\prime}\left(\tilde{Z}_{t}\right)H_{l,t}^{\prime}\left(\tilde{Z}_{t}\right)\right] & =\tilde{\mathbb{E}}\nabla\left(H_{k,t}^{\prime}H_{l,t}^{\prime}\right)^{\intercal}\left[\left(a_{t}-K_{t}\left(\bar{H}_{t}+H_{t}^{\prime}+h_{t}\left(y_{t}\right)\right)\right)\right]\mathrm{d}t\nonumber \\
 & +\tilde{\mathbb{E}}\left[\nabla\left(H_{k,t}^{\prime}H_{l,t}^{\prime}\right)^{\intercal}K_{t}\right]\mathrm{d}y_{t}+\tilde{\mathbb{E}}\left[K_{t}\Gamma_{t}^{2}K_{t}^{\intercal}:\nabla\nabla\left(H_{k,t}^{\prime}H_{l,t}^{\prime}\right)\right]\mathrm{d}t\nonumber \\
 & =\tilde{\mathbb{E}}\left[\nabla\left(H_{k}^{\prime}H_{l}^{\prime}\right)^{\intercal}a_{t}\right]\mathrm{d}t+\tilde{\mathbb{E}}\left[K_{t}\Gamma_{t}^{2}K_{t}^{\intercal}:\nabla\nabla\left(H_{k,t}^{\prime}H_{l,t}^{\prime}\right)\right]\mathrm{d}t\label{eq:cov_h}\\
 & +\tilde{\mathbb{E}}\left[\nabla\left(H_{k,t}^{\prime}H_{l,t}^{\prime}\right)^{\intercal}K_{t}\right]\left\{ \mathrm{d}y_{t}-\left[\bar{H}_{t}+h_{t}\left(y_{t}\right)\right]\mathrm{d}t\right\} -\tilde{\mathbb{E}}\left[\nabla\left(H_{k,t}^{\prime}H_{l,t}^{\prime}\right)^{\intercal}K_{t}H_{t}^{\prime}\right].\nonumber 
\end{align}
Using the identifies \eqref{eq:coeff_sde} for $a_{t}$ and $K_{t}$,
again we can find
\begin{align*}
\tilde{\mathbb{E}}\left[\nabla\left(H_{k,t}^{\prime}H_{l,t}^{\prime}\right)^{\intercal}a_{t}\right] & =-\tilde{\mathbb{E}}\left[K_{t}\Gamma_{t}^{2}K_{t}^{\intercal}:\nabla\nabla\left(H_{k,t}^{\prime}H_{l,t}^{\prime}\right)\right]-\int\nabla\left(H_{k,t}^{\prime}H_{l,t}^{\prime}\right)^{\intercal}K_{t}\Gamma_{t}^{2}\nabla\cdot\left(\tilde{\rho}_{t}K_{t}^{\intercal}\right)\mathrm{d}z\\
 & =-\tilde{\mathbb{E}}\left[K_{t}\Gamma_{t}^{2}K_{t}^{\intercal}:\nabla\nabla\left(H_{k,t}^{\prime}H_{l,t}^{\prime}\right)\right]+\tilde{\mathbb{E}}\left[\nabla\left(H_{k,t}^{\prime}H_{l,t}^{\prime}\right)^{\intercal}K_{t}H_{t}^{\prime}\right].
\end{align*}
Therefore, we have $\mathrm{d}\tilde{\mathbb{E}}\left[H_{k,t}^{\prime}H_{l,t}^{\prime}\right]=\tilde{\mathbb{E}}\left[\nabla\left(H_{k,t}^{\prime}H_{l,t}^{\prime}\right)^{\intercal}K_{t}\right]\left\{ \mathrm{d}y_{t}-\left[\bar{H}_{t}+h_{t}\left(y_{t}\right)\right]\mathrm{d}t\right\} $.
Further, using the identity for $K_{t}$, the coefficient becomes
third moments of $H_{t}^{\prime}$
\begin{align*}
\tilde{\mathbb{E}}\left[\nabla\left(H_{k,t}^{\prime}H_{l,t}^{\prime}\right)^{\intercal}K_{t}\right] & =\int\nabla\left(H_{k,t}^{\prime}H_{l,t}^{\prime}\right)^{\intercal}K_{t}\tilde{\rho}_{t}\mathrm{d}z\\
 & =-\int\left(H_{k,t}^{\prime}H_{l,t}^{\prime}\right)\left[\nabla\cdot\left(\tilde{\rho}_{t}K_{t}^{\intercal}\right)\right]^{\intercal}\mathrm{d}z\\
 & =\int\left(H_{k,t}^{\prime}H_{l,t}^{\prime}\right)\left[\tilde{\rho}_{t}\Gamma_{t}^{-2}\left(H\left(z\right)-\bar{H}_{t}\right)\right]^{\intercal}\mathrm{d}z\\
 & =\tilde{\mathbb{E}}\left[H_{k,t}^{\prime}H_{l,t}^{\prime}H^{\prime\intercal}\right]\Gamma_{t}^{-2}.
\end{align*}
Similarly, by repeating the same procedure for $\bar{H}_{k,t}H_{l}^{\prime}\left(z\right)$,
we have
\begin{align*}
\mathrm{d}\tilde{\mathbb{E}}\left[\bar{H}_{k,t}H_{l}^{\prime}\left(\tilde{Z}_{t}\right)\right] & =\tilde{\mathbb{E}}\left[\nabla\left(\bar{H}_{k,t}H_{l,t}^{\prime}\right)^{\intercal}K_{t}\right]\left\{ \mathrm{d}y_{t}-\left[\bar{H}_{t}+h_{t}\left(y_{t}\right)\right]\mathrm{d}t\right\} \\
 & =\tilde{\mathbb{E}}\left[H_{l,t}^{\prime}H_{t}^{\prime\intercal}\right]\Gamma_{t}^{-2}\left[\mathrm{d}y_{t}-\left(\bar{H}_{t}+h_{t}\right)\mathrm{d}t\right]\bar{H}_{k,t}.
\end{align*}
And similar result can be achieved for $\tilde{\mathbb{E}}\left[H_{k}^{\prime}\left(\tilde{Z}_{t}\right)\bar{H}_{l,t}\right]$.

Finally, applying It\^{o}'s formula for $\bar{H}_{t}\bar{H}_{t}^{\intercal}$
where $\mathrm{d}\bar{H}_{t}=C_{t}^{H}\Gamma_{t}^{-2}\left(\mathcal{H}\rho_{t}-\bar{H}_{t}\right)\mathrm{d}t+C_{t}^{H}\Gamma_{t}^{-1}\mathrm{d}B_{t}$
as we have derived, there is,
\begin{align*}
\mathrm{d}\left(\bar{H}_{t}\bar{H}_{t}^{\intercal}\right) & =\left(\mathrm{d}\bar{H}_{t}\right)\bar{H}_{t}^{\intercal}+\bar{H}_{t}\left(\mathrm{d}\bar{H}_{t}^{\intercal}\right)+\mathrm{d}\left\langle C^{H}\Gamma^{-1}B,C^{H}\Gamma^{-1}B\right\rangle _{t}\\
 & =C_{t}^{H}\Gamma_{t}^{-2}\left[\mathrm{d}y_{t}-\left(\bar{H}_{t}+h_{t}\right)\mathrm{d}t\right]\bar{H}_{t}^{\intercal}\\
 & +\bar{H}_{t}\left[\mathrm{d}y_{t}^{\intercal}-\left(\bar{H}_{t}+h_{t}\right)^{\intercal}\mathrm{d}t\right]\Gamma_{t}^{-2}C_{t}^{H}\\
 & +C_{t}^{H}\Gamma_{t}^{-2}C_{t}^{H}\mathrm{d}t.
\end{align*}
Notice again that the white noise process, $C_{t}^{H}\Gamma_{t}^{-1}\mathrm{d}B_{t}$,
gives the last term in the first equality above. Putting all the above
equations together, we get the equation for $\mathrm{d}C_{t}^{H}=\mathrm{d}\tilde{\mathbb{E}}\left[H\left(\tilde{Z}_{t}\right)H\left(\tilde{Z}_{t}\right)^{\intercal}\right]-\mathrm{d}\left(\bar{H}_{t}\bar{H}_{t}^{\intercal}\right)$
where $C_{kl,t}^{H}=\tilde{\mathbb{E}}\left[H_{k,t}^{\prime}H_{l,t}^{\prime}\right]$.
\end{proof}
\
\begin{proof}
[Proof of Proposition \ref{prop:Kalman_gain}]According to \eqref{eq:kalman_gain}
with $K=\tilde{K}\Gamma^{-2}$, we need to show
\[
-\nabla\cdot\left(\tilde{K}^{\intercal}\tilde{\rho}\right)H^{\intercal}=\tilde{\rho}H^{\prime}H^{\intercal}\:\Rightarrow\:\tilde{\mathbb{E}}\left[\tilde{K}^{\intercal}\nabla H\right]=\tilde{\mathbb{E}}\left[H^{\prime}\left(\bar{H}+H^{\prime}\right)^{\intercal}\right]=C^{H},
\]
according to the specific expressions $H=H^{m}$ and $H=H^{v}$. First,
we can compute
\[
\nabla_{z}H_{l}^{m}=2A_{l}z,\quad\nabla_{z}H_{pq}^{v}=2z_{q}A_{p}z+\left(z^{\intercal}A_{p}z\right)\delta_{qj}\hat{e}_{j}.
\]
Above in $H^{v}$ for simplicity, we only compute half of the symmetric
function and $\hat{e}_{j}$ is the unit vector with value $1$ in
the $j$-th entry.

From direct computations for $H^{m}$ and using $H_{k}^{m}=z^{\intercal}A_{k}z$,
we have
\begin{align*}
\sum_{j}\tilde{K}_{j,k}^{m}\frac{\partial H_{l}^{m}}{\partial z_{j}} & =\frac{1}{2}\left[\left(z^{\intercal}A_{k}z\right)-\bar{H}_{k}^{m}\right]\sum_{j}z_{j}2\left(A_{l}z\right)_{j}\\
 & =\left[\left(z^{\intercal}A_{k}z\right)-\bar{H}_{k}^{m}\right]\left(z^{\intercal}A_{l}z\right)=H_{k}^{m\prime}H_{l}^{m}.
\end{align*}
Similarly for $H^{v}$, we can compute
\begin{align*}
\sum_{j}\tilde{K}_{j,kl}^{v}\frac{\partial H_{pq}^{v}}{\partial z_{j}} & =\frac{1}{3}\left[\left(z^{\intercal}A_{k}z\right)z_{l}-\bar{H}_{kl}^{v}\right]\sum_{j}z_{j}\left[2z_{q}\left(A_{p}z\right)_{j}+\left(z^{\intercal}A_{p}z\right)\delta_{qj}\right]\\
 & =\frac{2}{3}\left[\left(z^{\intercal}A_{k}z\right)z_{l}-\bar{H}_{kl}^{v}\right]\left(z^{\intercal}A_{p}z\right)z_{q}+\frac{1}{3}\left[\left(z^{\intercal}A_{k}z\right)z_{l}-\bar{H}_{kl}^{v}\right]z_{q}\left(z^{\intercal}A_{p}z\right)\\
 & =\left[\left(z^{\intercal}A_{k}z\right)z_{l}-\bar{H}_{kl}^{v}\right]\left(z^{\intercal}A_{p}z\right)z_{q}=H_{kl}^{v\prime}H_{pq}^{v}.
\end{align*}
This finishes the proof.
\end{proof}
\
\begin{proof}
[Proof of Proposition \ref{thm:conv_stat}]First, consider the mean
equations from the same initial state, we have
\begin{align*}
\bar{u}_{t}^{N,\tau}-\bar{u}_{t}= & \int_{0}^{t}\left[M\left(\bar{u}_{\sigma\left(s\right)}^{N,\tau}\right)-M\left(\bar{u}_{s}\right)\right]\mathrm{d}s+\int_{0}^{t}\left[F_{\sigma\left(s\right)}-F_{s}\right]\mathrm{d}s\\
 & +\int_{0}^{t}\left[\left\langle H^{m},\tilde{\rho}_{\sigma\left(s\right)}^{N}\right\rangle -\left\langle H^{m},\tilde{\rho}_{s}\right\rangle \right]\mathrm{d}s.
\end{align*}
Therefore, using H\"{o}lder's inequality and Lipschitz condition
for $M$ there is
\begin{align}
\mathbb{E}\sup_{t\leq T}\left|\bar{u}_{t}^{N,\tau}-\bar{u}_{t}\right|^{2} & \leq3T\beta\mathbb{E}\int_{0}^{T}\left|\bar{u}_{\sigma\left(s\right)}^{N,\tau}-\bar{u}_{s}\right|^{2}\mathrm{d}s\nonumber \\
 & +3T\mathbb{E}\int_{0}^{T}\left|\left\langle H^{m},\tilde{\rho}_{\sigma\left(s\right)}^{N}\right\rangle -\left\langle H^{m},\tilde{\rho}_{s}\right\rangle \right|^{2}\mathrm{d}s+C_{T}\tau\nonumber \\
 & \leq C_{1}\int_{0}^{T}\mathbb{E}\sup_{s^{\prime}\leq s}\left|\bar{u}_{s^{\prime}}^{N,\tau}-\bar{u}_{s^{\prime}}\right|^{2}\mathrm{d}s+C_{2}\tau\left\Vert H^{m}\right\Vert _{\infty}^{2}+\frac{C_{3}}{N}\left\Vert H_{m}\right\Vert _{\infty}^{2}.\label{eq:est_m-1}
\end{align}
In the first term above on the right hand side, we follow the same
procedure to estimate the error in the corresponding continuous solution
$\bar{u}_{s}^{N,\tau}$ compared to the time discretization solution
$\bar{u}_{\sigma\left(s\right)}^{N,\tau}$
\[
\left|\bar{u}_{\sigma\left(s\right)}^{N,\tau}-\bar{u}_{s}\right|^{2}\leq\left|\bar{u}_{\sigma\left(s\right)}^{N,\tau}-\bar{u}_{s}^{N,\tau}\right|^{2}+\left|\bar{u}_{s}^{N,\tau}-\bar{u}_{s}\right|^{2}\leq C\tau\left\Vert H^{m}\right\Vert _{\infty}^{2}+\left|\bar{u}_{s}^{N,\tau}-\bar{u}_{s}\right|^{2}.
\]
And in the second line for the term related to $H^{m}$, combining
the discrete time estimate with \eqref{eq:conv_pdf} gives
\[
\mathbb{E}\left[\sup_{t\leq T}\left|\left\langle H^{m},\tilde{\rho}_{\sigma\left(s\right)}^{N}\right\rangle -\left\langle H^{m},\tilde{\rho}_{t}\right\rangle \right|^{2}\right]\leq\left(C\tau^{2}+\frac{C_{T}}{N}\right)\left\Vert H^{m}\right\Vert _{\infty}^{2}.
\]
Therefore, applying Gr\"{o}nwall's inequality to \eqref{eq:est_m-1},
we get the mean estimate in \eqref{eq:stat_bnds}.

Next, under a similar fashion, we can compute from the covariance
equation and using the Lipschitz condition for $L$
\begin{align*}
\mathbb{E}\sup_{t\leq T}\left|R_{t}^{N,\tau}-R_{t}\right|^{2} & \leq C_{1}\beta\mathbb{E}\sup_{t\leq T}\left\Vert \bar{u}_{t}\right\Vert _{\infty}^{2}\int_{0}^{T}\left|R_{\sigma\left(s\right)}^{N,\tau}-R_{s}\right|^{2}\mathrm{d}s\\
 & +C_{2}\beta\mathbb{E}\sup_{t\leq T}\left|\bar{u}_{\sigma\left(t\right)}^{N,\tau}-\bar{u}_{t}\right|^{2}\sup_{t\leq T}\left\Vert R_{\sigma\left(t\right)}\right\Vert ^{2}\\
 & +C_{3}\mathbb{E}\int_{0}^{T}\left|\left\langle H^{v},\tilde{\rho}_{\sigma\left(s\right)}^{N}\right\rangle -\left\langle H^{v},\tilde{\rho}_{s}\right\rangle \right|^{2}\mathrm{d}s+C_{T}\tau.
\end{align*}
Using the uniform boundedness of $\bar{u}_{t},R_{t}$ and \eqref{eq:conv_pdf}
for $H^{v}$ together with the previous estimate for $\mathbb{E}\sup_{t\leq T}\left|\bar{u}_{\sigma\left(t\right)}^{N,\tau}-\bar{u}_{t}\right|^{2}$,
we reach the final covariance estimate in \eqref{eq:stat_bnds}.
\end{proof}

\bibliographystyle{plain}
\bibliography{refs}
\end{document}